\documentclass[]{elsarticle}
\usepackage[T1]{fontenc}
\usepackage{multicol}
\usepackage{esvect}
\usepackage{lineno,hyperref}
\usepackage{amsmath}
\usepackage{dsfont}
\usepackage{caption}
\usepackage{subfig}
\usepackage{amssymb}
\usepackage{url}
\usepackage{txfonts}
\modulolinenumbers[5]
\usepackage{upgreek}
\usepackage{pgfplots}
\usepackage{placeins}
\pgfplotsset{compat=1.3}
\journal{International Journal of Multiphase Flows}

\newcommand{\Nmax}{N_{\text{max}}}
\newcommand{\beq}[1]{\begin{equation}\label{#1}}
\newcommand{\eeq}{\end{equation}}

\newcommand{\lbold}{\boldsymbol{l}}

\newcommand{\tens}[1]{{\boldsymbol{#1}}}
\newcommand{\Ug}{\tens{U}_g}

\newcommand{\R}{\mathbb{R}}

\def\xv{\tens{x}}
\def\yv{\tens{y}}
\def\rv{\tens{r}}
\def\XV{\tens{X}_i}

\newcommand{\U}{\tens{U}}

\newcommand{\size}{S}

\newcommand{\nME}{n^{ME}}

\newcommand{\Hcurv}{\mathrm{H}}

\newcommand{\real}{\mathbb{R}}
\newcommand{\Vol}{\mathcal{V}}
\newcommand{\dual}{\mathcal{M}}
\newcommand{\uparam}{u}
\newcommand{\vparam}{v}
\newcommand{\Uopen}{\mathcal{U}}
\newcommand{\vi}{V_I}
\newcommand{\ve}{V_e}
\newcommand{\Vi}{{\tens{V}_I}}
\newcommand{\Vp}{{\tens{V}_p}}
\newcommand{\Up}{\tens{U}_p}
\newcommand{\Vic}[1]{{V_{I,#1}}}

\newcommand{\normal}{\tens{N}}

\newcommand{\Et}{\tens{e}^*}
\newcommand{\Surf}{A}
\newcommand{\One}{\mathds{1}}

\newcommand{\Stretch}{\mathcal{S}}

\newcommand{\Aver}[1]{\overline{#1}}
\newcommand{\phase}[1]{\widehat{#1}}
\newcommand{\normalizeASDF}{\Theta_h}
\newcommand{\Uf}{\tens{U}}
\newcommand{\Surface}{\varmathbb{S}}
\newcommand{\totalSurf}{\Surface_I}
\newcommand{\portsurf}[1]{\Surface_{#1}}
\newcommand{\gradient}[2]{\tens{\nabla}_{#1}{#2}}
\newcommand{\diverg}[2]{\tens{\nabla}_{#1}\cdot{#2}}
\newcommand{\Gcurv}{G}
\newcommand{\Mcurv}{H}
\newcommand{\vol}{v}
\newcommand{\Areadrop}{S}
\newcommand{\area}{\text{area}}
\newcommand{\SDF}{F}
\newcommand{\sdf}{\SDF'}
\newcommand{\nsdf}{\mathcal{\SDF}'}
\newcommand{\GNDF}{f_{\Sigma}}

\newcommand{\ngndf}{\mathcal{N}'}
\newcommand{\NDF}{f}
\newcommand{\NNDF}{N}
\newcommand{\ndf}{n}
\newcommand{\conditionvar}[1]{\mathsf{#1}}
\newcommand{\Domain}{\Omega_{x}}
\newcommand{\subdomain}{E}
\newcommand{\phasespace}{\Omega_{\xi}}
\newcommand{\phasevar}{\tens{\xi}}
\newcommand{\phasevarv}{\Aver{\tens{\xi}}^{\vol}}
\newcommand{\particle}{\mathit{p}}
\newcommand{\ASDF}{\mathrm{\SDF}^a}

\newcommand{\nasdf}{\mathcal{\sdf}^a}
\newcommand{\aaverg}[1]{\widetilde{#1}}
\newcommand{\paverg}[1]{\overline{#1}}
\newcommand{\Volh}[1]{\mathcal{V}_h{(#1)}}

\newcommand{\Breakup}{\mathrm{\Gamma}}

\newcommand{\ds}{\displaystyle}

\newcommand{\Stokes}{\mathrm{St}}

\newcommand{\EvapRate}{\mathrm{R}_S}

\newcommand{\KEvap}{\mathrm{K}}

\newcommand{\Vertsarray}{\tens{\mathrm{V}}}
\newcommand{\Facearray}{\tens{\mathrm{F}}}
\newcommand{\Neighbour}{\tens{\mathrm{O}}}
\newcommand{\vertex}{\text{V}}
\newcommand{\face}{\text{T}}
\newcommand{\averglen}{k}

\newcommand{\mom}{\mathrm{m}}

\newcommand{\ns}{n}

\newcommand{\Alphad}{\alpha_d}
\newcommand{\Sigmad}{\Sigma_d}

\newcommand{\dint}[2]{d^{#1}#2}
\newcommand{\SEntropy}{E_{sh}}

\newtheorem{theorem}{Theorem}[section]

\newtheorem{proposition}[theorem]{Proposition}

\newenvironment{proof}{\paragraph{Proof:}}{\hfill$\square$}
\newenvironment{definition}[1][Definition]{\begin{trivlist}
\item[\hskip \labelsep {\bfseries #1}]}{\end{trivlist}}

\begin{document}

\begin{frontmatter}

\title{Statistical modeling of the gas-liquid interface using geometrical variables: 
toward a unified description of the disperse and separated phase flows}


\author[ifpadd,em2cadd]{\ \ ESSADKI Mohamed\corref{mycorrespondingauthor}}
\cortext[mycorrespondingauthor]{Corresponding author}
\author[em2cadd]{\ DRUI Florence}
\author[ifpadd]{\ de CHAISEMARTIN St\'ephane}
\author[em2cadd]{LARAT Adam}
\author[coriaadd]{M\'ENARD Thibault}
\author[em2cadd,cpamadd]{MASSOT Marc}

\address[ifpadd]{IFP Energies nouvelles, 1-4 Avenue du Bois Pr\'eau, 92852 Rueil-Malmaison, France}
\address[em2cadd]{Laboratoire EM2C UPR 288, CNRS, CentraleSup\'elec, Universit\'e Paris-Saclay, Grande 3, rue Joliot-Curie 91192 Gif-sur-Yvette cedex France}
\address[coriaadd]{Laboratoire CORIA, Avenue de l'Universit\'e, 76801 Saint-\'Etienne-du-Rouvray, France}
\address[cpamadd]{Centre de Math\'ematiques Appliqu\'ees, Ecole polytechnique, CNRS, Universit\'e Paris-Saclay, Route de Saclay,  
91128 Palaiseau Cedex, FRANCE}

\begin{abstract}
In this work, we investigate an original strategy in order to derive a statistical modeling of the interface in gas-liquid two-phase flows through geometrical variables.
The contribution is two-fold. First it participates in the theoretical design of a unified reduced-order model for the description of two regimes: a disperse phase in a carrier fluid and two separated phases. The first idea is to propose a statistical description of the interface relying on geometrical properties such as the mean and Gauss curvatures and define a Surface Density Function (SDF). The second main idea consists in using such a formalism in the disperse case, where a clear link is proposed between local statistics of the interface and the statistics on objects, such as the number density function in Williams-Boltzmann equation for droplets. This makes essential the use of topological invariants in geometry through the Gauss-Bonnet formula and allows to include the works conducted on sprays of spherical droplets. It yields a statistical treatment of populations of non-spherical objects such as ligaments, as long as they are homeomorphic to a sphere.  Second, it provides an original angle and algorithm in order to build statistics from DNS data of interfacial flows. From the theoretical approach, we identify a kernel for the spatial averaging of geometrical quantities preserving the topological invariants. Coupled to 
a new algorithm for the evaluation of curvatures and surface that preserves these invariants, we analyze two sets of DNS results conducted with the ARCHER code from CORIA
with and without topological changes and assess the approach.
\end{abstract}

\begin{keyword}
disperse/separated phases, gas-liquid interface, moments method, Gauss- Bonnet formula, computational geometry, surface and number density function.
\MSC[2010] 76T10\sep 65D99 \sep 53A17 \sep 76A99 \sep35Q35
\end{keyword}

\end{frontmatter}



\section{Introduction}

Nowadays, direct fuel injection systems are widely used in automotive 
engines to better atomize and mix the fuel with the air, in order to 
improve the combustion, reduce pollutant emissions and save fuel. 
The design of new and efficient 
injectors need to be assisted by predictive numerical simulations 
to test new designs, understand the various physical mechanisms 
involved in this complex problem and at the same time, at a reduced price 
compared to experiments. 
Near the nozzle outlet, of diameter less than $1mm$, 
the fuel is a bulk liquid phase separated from the gaseous phase. In this region, the two-phase flow is 
denoted \textit{separated phases}. After primary atomization, the fuel is disintegrated into a polydisperse spray. Indeed, the 
large drops can undergo secondary breakup leading to a cloud of  
very small disperse droplets of size around $1-10\, \mu m$, called the \textit{disperse phase}. 
The major challenges faced by the modeling of fuel injection are three-fold: the strong multi-scale feature of 
the considered problem, the coexistence of two specific two-phase regimes 
(the separated and the disperse regime) and the entanglement between the first two points within 
the transitional regime. 

%
%
In Direct Numerical Simulation (DNS), each phase dynamics is resolved through
monophasic Navier-Stokes equations while the gas-liquid interface is determined 
using interface capturing methods (VOF, level-set or a combination of both methods) \cite{VOF_Hirt81,level_set, bourlioux95,menard07} or 
interface tracking methods as initiated in \cite{Hirten74}. DNS codes, such as the ARCHER code \cite{menard07},
may be used for simulations of academic and some simple injection configurations.
The results of these simulations help to better understand the phenomena, while accurate information on
detailed physics is sometimes hard to extract from experimental measurements \cite{lebas2009,Desjardins_2013,herrmann,lechenadec2013,vaudor}. 
However, configurations at large Reynolds number, that characterizes the turbulent regime of the flow, 
or at large Weber number, related to the stability of the interface, may be extremely
costly from a numerical point of view and DNS may fail in predicting the 
smallest interfacial structures, like small droplets or thin ligaments. Therefore, while
DNS is of great interest in academic research, it is not appropriate for 
industrial use.

Reduced-order models intend to avoid the simulation of the smallest scales
of the configuration by providing a description and evolution laws only for some 
quantities of interest, such as  for example the volume fraction of the phases in 
mixture zones or the area density of the two-phase interface. 
So far, however, these models
inherently depend on the two-phase flow regime: separated or disperse 
phase regimes.
Building up a multi-scale and accurate model with the capacity of 
resolving the whole injection process is a challenging task, that can be addressed either 
by coupling models associated with the two main flow classes, namely the disperse and separated 
phase flows, or by developing a unified approach.

The reduced-order models for disperse phase flows usually adopt a statistical 
approach. They rely on a number density function (NDF), 
that satisfies a generalized population balance equation (GPBE), also known 
as Williams-Boltzmann Equation (WBE) \cite{williams1958}. The phase space variables of the 
NDF can include different physical properties of the disperse particles (droplets
in the case of atomization) such as velocity, 
size, temperature, \textit{aso.} However, the large phase-space dimension makes the 
deterministic resolution of the WBE quickly unaffordable. Multiple approaches have then been 
developed.
First, one can use Lagrangian Monte-Carlo approach to estimate the NDF 
\cite{bird94}. Even though this approach can be considered as the most accurate for
solving WBE, it still suffers from many drawbacks: the need for a high 
number of numerical particles to ensure the statistical convergence, the requirement of 
complex load balancing algorithms for parallel computations or the difficulty of
coupling with the gaseous phase, which is most of the time solved in an Eulerian framework. Second,
one can derive Eulerian 
moment models from the WBE. With these models, the NDF is not directly solved,
but one solves for a finite set of its moments. 
Then, a most likely NDF can be reconstructed from the finite number of moments:
this is called the Hausdorff moment problem \cite{tagliani99}. 
Among the properties of a disperse phase, the size-distribution
of the droplets has a major impact on the results of the 
combustion simulations \cite{Vie13jx,Hannebique13jx}. Currently, in the context of Eulerian moment methods, 
this distribution can be accounted for in three ways: 
\begin{enumerate}
\item multi-fluid models, where the size phase space is discretized into sections.
Then moments up to order two can be used in each section \cite{deChaisemartin2009,laurent2015}, 
\item quadrature-moment methods such as QMOM, DQMOM or EQMOM \cite{nguyen2016,laurent-12}, consider 
the NDF as a sum of Dirac-delta functions, potentially extended to kernels,
\item high order moments with continuous reconstruction developed in \cite{kah12,emre2015,massot2010,vieJCP2012,Essadki-SIAM}. 
At each step, a continuous NDF maximizing the Shannon entropy is reconstructed from the high order moments \cite{mead84}, with a complete coverage of the whole moment space. 
\end{enumerate}
Using high order moments with continuous reconstruction through entropy 
maximization avoids the use of several sections, that are computationally more expensive \cite{kah2010}. 
Moreover, the continuous reconstruction allows a more accurate and consistent evaluation of the evaporation 
flux compared to the discontinuous approach used in the quadrature-moment methods, as well as a comprehensive description of the moment space.
In the current work, we adopt this approach.

Among the different approaches that may be used in separated-phase configurations,
let us mention the interface capturing methods based on the resolution of two-fluid or
mixture model resulting in a system of PDEs as in \cite{chanteperdrix2002,murrone2005,champmartin2014}. 
These equations may be derived through an averaging procedure \cite{ishii75,DrewPassman99},
or by using a variational principle \cite{gavrilyuk2002,drui2016b}. In both cases, they stand
for a spatial, temporal or ensemble averaged two-phase flow, where one or both phases
can be present at any given location, according to the values of a characteristic function,
that is generally the volume fraction of one of the phases. Traditionally, these models provide
little information about the sub-scale interfacial structures. Recent works, as in \cite{drui2016b}, 
have shown that they can be enriched with further sub-scale physics; yet such models are far from the 
ability to deal with polydispersity such as kinetic models do.
From a numerical point of view, two-fluid models are used for simulations
of separated phases and interfacial configurations, for which the exact location of the interface
is not reconstructed, but lies in a mixture zone due to numerical approximation.
Although the numerical diffusion can be reduced by using more accurate numerical
schemes, the spreading of the interface is still a main bottleneck of these methods
for simulations of atomization.

Recent works have been devoted to the coupling of approaches for separated phases
and those for disperse phase flows. 
Among them, LeTouze et al \cite{letouze2015} proposed to couple a diffuse interface model for separated 
phases with a multi-fluid model for the disperse phase. Up to now, the exchange terms between both models
depend on the configuration of the atomization and cannot predict the distribution in size of the disperse droplets.
One can also mention the Eulerian-Lagrangian Spray Atomization
(ELSA) \cite{vallet99,vallet2001,lebas2009} technique where a two-fluid model is enhanced with an 
equation for the expected surface area density in the dense zones. This model is then coupled with 
a Lagrangian approach for the simulation of the disperse phase in the dilute zones.
In \cite{Bejoy} the same set 
of equations is used with a differentiation between the variables describing the disperse 
and separated phases. Two additional 
equations on the expected density area are also used: one for the separated phases and one for 
the disperse phase.
So far, these approaches do not provide a unified description of the whole
atomization (from the separated phase to the spray of droplets) 
and fail in providing an accurate description of the polydispersion 
for the disperse phase. 
Indeed, in the works presented above, the description of the  
gas-liquid interface geometry relies on one or two variables only,
that are the volume fraction and the expected surface area density. This information 
is not sufficient to reconstruct a NDF of a polydisperse spray. 

The idea of the present work is to study the possibility of describing gas-liquid 
interfaces by using new geometrical information, such as  the mean and Gauss curvatures,
for which evolution equations have been derived in \cite{Drew_Geom}. Moreover, 
the coupling with Eulerian models for sprays is considered through a statistical 
description of the interface. First, inspired by the pioneering works of \cite{pope} for the description of the dynamics of
flames, we define a Surface Density Function (SDF)
with a different phase space: in our work, it is composed of the mean and Gauss 
curvatures and the interfacial velocity. 
The key issue is to make the link between statistics on a local description of the interface through geometrical variables, such as curvatures, 
and statistical description of isolated objects through a number density function in the appropriate phase space.
A discrete formalism of the SDF is then proposed to describe 
disperse phase flows. In this case, the geometrical quantities are averaged on the surface of 
isolated droplets or bubbles. 
We show that this new formalism can be related
to the high-order moment model proposed in  \cite{Essadki-SIAM,OGST} for 
sprays of spherical droplets. Indeed, by considering fractional moments of a NDF, this model 
uses the same geometrical quantities to describe the polydispersity of the droplets.
Going forward into a full generalization of the interface description, we consider spatially averaged SDF (ASDF),
with an averaging kernel bounded to a small region around the interface. This spatial average is applied to each 
separated realization, before an ensemble average is applied. When performed in a 
consistent way, this process preserves some necessary topological properties of the gas-liquid interface description. 
This leads us to define a generalized NDF (GNDF), which can be used for both  
separated and disperse phases and which degenerates to a standard NDF in the 
disperse region. We also propose new numerical algorithms and procedures for the computation of the curvatures 
and of the different distributions 
(SDF and NDF) from the values of a Level Set function. 
In the last section, we use these new algorithms for the post-processing of some DNS simulations, obtained with the ARCHER code, developed at CORIA laboratory \cite{menard07} and to assess the theoretical part of this work and pave the way to a new way of analyzing DNS of interfacial flows.

\section{Surface element properties and probabilistic description of the gas-liquid interface}
\label{sect:gas-liquid-interf}

Immiscible two-phase flows, such as gas-liquid mixtures, are characterized by the presence of a sharp interface.
Indeed, the gas-liquid interface thickness being of the order of the molecular mean free path ($\lambda=10$~nm), it is smaller 
than the microscopic length scales of the vast majority of two-phase flow applications.
Thus, at macroscopic scale, this interface can be represented as a 2D dynamic surface embedded in a 3D domain. 
In the following, we consider a two-phase flow (gas and liquid) within a finite 3D domain $\Domain$.
Let us denote by $\totalSurf(t)$ the moving surface that separates both phases. First, we do not
make any assumption about the flow topology (separated or disperse). 
We start by defining a 2D surface and some intrinsic geometrical variables in subsections \ref{sect:surface-def}-\ref{sect:TimeEvolution}.
Then in subsection \ref{sect:general-SDF}, we introduce a statistical description of the gas-liquid interface similar to Pope's description of flames dynamics and propagation \cite{pope}.

\subsection{Surface definition}
\label{sect:surface-def}

Since we are dealing with disperse and separated phases in the same domain, the moving surface $\totalSurf(t)$ 
is not necessarily a connected set: in general it consists in a set of closed (disperse phase) and unclosed (bulk phase) connected sub-surfaces. 
Therefore, the global surface $\totalSurf(t)$ can be written as a union of  connected sub-surfaces $\portsurf{i}(t)$:
\begin{equation}
\label{eq:definition-general-surface}
	\totalSurf(t)=\bigcup\limits_{i=1}^{\Nmax}\portsurf{i}(t),
\end{equation}
that are separate in the sense that,  for $i\neq j$:
$$\min\limits_{(\tens{x},\tens{y})\in\portsurf{i}(t)\times\portsurf{j}(t)}(||\tens{x}-\tens{y}||)>0.$$

Now, let us consider two standard descriptions of a moving surface.
First, the surface can be defined as the set of zeros of a time-space function $(t;\xv) \mapsto g(t;\xv)$:
\begin{equation}
\label{eq:implicit-surface-definition}
	\totalSurf(t)=\{\xv\in\Domain;\quad g(t;\xv)=0\}.
\end{equation}
The function $g$ satisfies the following kinematic equation:
\begin{equation}
	\frac{\partial g}{\partial t}+\Vi^g\cdot\tens{\nabla}(g)=0,
\label{eq:implicit-kinematic}
\end{equation}
where $\Vi^g (t; \xv)$ is an interface velocity, whose definition is discussed thereafter.
Let us call this approach: the "implicit definition of a surface".

On the other hand, the "explicit definition of a surface", also called the "parametric surface", 
consists in using two real parameters $u,v$ and parametric functions:
\begin{equation}
  \XV(t;.):\left\{
  \begin{matrix}  
    \Uopen_i            & \longrightarrow & \Domain\\ 
    (\uparam,\vparam) & \longmapsto     & \XV(t;\uparam,\vparam) 
  \end{matrix}
  \right. \, , \quad i \in \left\{ 1, ... , \Nmax \right\},
\label{eq:explicit-surface}
\end{equation}
where $\Uopen_i$ is a subset of $\varmathbb{R}^2$, such that:
\begin{equation}
	\portsurf{i}(t)=\{ \XV(t;\uparam,\vparam);\quad (\uparam,\vparam)\in\Uopen_i\}.
\end{equation}

Each couple of parameters $(\uparam,\vparam)$ is associated with a point located on the surface 
and moving according to an interface velocity that depends on the choice of the parametrization:
\begin{equation}
	\frac{\partial \XV(t;\uparam,\vparam)}{\partial t}=\Vi^{\Uopen}(t,\XV(t;\uparam,\vparam)).
\label{eq:explicit-kinematic}
\end{equation}
We will discuss in the next section the relation between the interface velocities used in the implicit definition 
$\Vi^g$ and in the explicit definition $\Vi^{\Uopen}$. 

Although we mostly consider the implicit definition
of the surface \eqref{eq:implicit-surface-definition}
some of the definitions provided in the following will also be given for a parametric description of the surface for the sake of clarity.
From now on, we suppose that each sub-surface $\portsurf{i}$ is a $C^2$ oriented surface, 
such that the space function $g(t;.)$ is a $C^2$ differential function. 
The orientation is chosen such that the gradient $\tens{\nabla}g(t;\xv)$ at the interface points ($g(t;\xv)=0$) 
is strictly oriented toward the gaseous phase.

\subsection{Intrinsic gas-liquid interface variables}
\label{sect:IntrinsicVariables}

The aim of this paragraph is to introduce some local intrinsic properties of the interface, meaning that the quantities associated
with these properties do
not depend on the way the surface is defined (implicitly or explicitly, choice of a parametrization...). These quantities will be useful
for setting a statistical description of the interface. 

\subsubsection{Normal vector and tangent plane} 
Let us consider a point on the surface $\xv_i\in\totalSurf(t)$. 
Using \eqref{eq:implicit-surface-definition}, the normal vector at $\xv_i$ is given by:
\begin{equation}
	\normal(t,\xv_i)=\frac{\gradient{\xv}{g}(t;\xv_i)}{||\gradient{\xv}{g}(t;\xv_i)||}.
\label{eq:normal-def}
\end{equation}
The tangent plane to the surface at $\xv_i$ is the unique plan orthogonal to $\normal$ and passing by $\xv_i$. 
One can consider that the tangent plane provides a first order approximation of the surface at $\xv_i$.

\subsubsection{Curvatures} 

	Curvatures are defined as the infinitesimal variations of $(t, \xv) \mapsto \normal (t, \xv)$ when $\xv$ follows a path over
	the interface.
These variations can be expressed as a function of the Hessian matrix $ \mathcal{H}(g)$ of the 
function $g(t;.)$
according to (see also \cite{curvature_implicit}):
\begin{equation}
  \gradient{\xv}{\normal}^T=-\frac{1}{||\gradient{\xv}{g}||}(\mathbf{I}_{3}-\normal\otimes\normal^T)\mathcal{H}(g).
  \label{eq:curvature_tensor}
\end{equation}
where $\mathcal{H}(g)$ is given by:
\begin{equation}
  \mathcal{H}(g)(t;\xv)= \gradient{\xv}{\left( \gradient{\xv}{g} \right)} =\left(\frac{\partial^2 g}{\partial x_j\partial x_k}(t;\xv)\right)_{j,k=1\hdots 3},
\end{equation}
and $\mathbf{I}_3$ is the identity matrix. 
	
It can be shown that there exists an orthonormal basis $\{\Et_1,\Et_2\}$ of the tangent plane at the surface point $\xv_i$, 
such that the restriction of the matrix $\gradient{\xv}\normal^T$ to the tangent plane is a $2\times2$ diagonal matrix \cite{curvature_implicit}. 
In other words, in the orthonormal basis $\{\Et_1,\Et_2,\normal\}$, $ \gradient{\xv}{\normal}^T$ reads:
\begin{equation}
	\gradient{\xv}{\normal}^T=\left(
		\begin{matrix}
			\kappa_1 & 0 & \sigma_1\\
			0 & \kappa_2 & \sigma_2\\
			0 & 0 & 0\\
		\end{matrix}
	\right) \, ,
\end{equation}
where $\kappa_1\geq\kappa_2$ are the two principal curvatures and $(\sigma_1,\sigma_2)$ are two real variables. 
The eigenvectors $\{\Et_1,\Et_2\}$ corresponding to the eigenvalues $(\kappa_1,\kappa_2)$
are also called the \textit{principal directions} of the surface at $\xv_i$.

Instead of using the two principal curvatures, one can consider the \textit{mean curvature} $\Mcurv$ and 
the \textit{Gauss curvature} $\Gcurv$, defined by:
\begin{equation}
\begin{array}{rcl}
\Mcurv&=&\ds \frac{1}{2}(\kappa_1+\kappa_2),\\
\Gcurv&=&\ds \kappa_1\kappa_2.\\
\end{array}
\end{equation}

One could easily demonstrate that
the mapping 
\beq{eq:mapping} 
  \left\{
    \begin{array}{ccc}
      \left\{(\kappa_1,\kappa_2)\in \R^2; \; \kappa_1\geq\kappa_2\right\} 
      & \longrightarrow & 
      \left\{(\Mcurv,\Gcurv)\in \R_+^2; \; \Mcurv^2\geq\Gcurv\right\}\\
      (\kappa_1,\kappa_2) &\longmapsto&(\Mcurv,\Gcurv)
    \end{array}
  \right.,
\eeq
is one-to-one.

\subsubsection{Area density measure and stretch rate} 

	The last quantity needed in the following is an evaluation of the 
	interface area within any control volume. For this purpose, one defines the \textit{area density
	measure} as follows, see  \cite{Morel2015} and related works:
\begin{equation}
\label{eq:deltaI-definition}
	\delta_I(t;\xv)=||\gradient{\xv}{g}|| \, \delta(g(t,\xv)),
\end{equation}
where $\delta$ denotes the Dirac function. 
Consequently, for any volume $\Vol$, the area $\Surf_{\Vol}(t)$ of surface contained in $\Vol$ 
at time $t$ is given by:
\[
	\Surf_{\Vol}(t) = \int_{\real^3} \One_{\Vol}(\xv) \delta_I(t;\xv), 
\]
where $\One_{\Vol}(\xv)$ is the characteristic function of volume $\Vol$. It is important to note that the measure
$\delta_I$ does not depend on the choice of $g$. It is thus an intrinsic property of the interface. 


\subsubsection{Interface velocity} 

We have seen in paragraph \ref{sect:surface-def} that the interface velocity may have multiple definitions, according to
the definition of the interface (implicit \eqref{eq:implicit-surface-definition} or explicit \eqref{eq:explicit-surface}) 
or the choice of the parametrization. Let us
choose here a unique definition of this velocity.

As underlined in \cite{Drew_Geom}, one can see that the evolution of the interface according to \eqref{eq:implicit-kinematic}
only depends on the normal component of $\Vi^g$. Let us then introduce  $\vi  = \Vi^g \cdot \normal$, that is unambiguously
defined, and the interface velocity by:
\begin{equation}
	\Vi=\vi\normal,
\label{eq:normal-interface-velocity}
\end{equation}

By considering a point $ \XV(t;\uparam,\vparam)$ on the surface for a given 
parametrization and the application $\varphi : (t; u,v) \mapsto g(t;  \XV(t;\uparam,\vparam))$, 
one can show the following relation with the 
the velocity  $\Vi^{\Uopen}$ defined in \eqref{eq:explicit-kinematic}:
\[
	\vi =  \Vi^{\Uopen} \cdot \normal.
\]

Now, let us consider that the interface lies in a medium, whose velocity is $\mathbf{U} (t, \xv)$. Moreover, 
this interface may propagate, normal to itself, at a  velocity $\ve \normal$ relative to the medium \cite{pope}. This is for example the case
of a flame that separates burnt and unburnt gases, for which the flame speed $\ve$ depends on the chemical
reactions rates. Or in the case of two-phase flows, $\ve$ may characterize the rate of phase transitions, such as
evaporation. The interfacial velocity is then related to these quantities through:
\[
	\vi = (\mathbf{U}\cdot\normal)+\ve.
\]

\subsection{Time evolution of interfacial variables}
\label{sect:TimeEvolution}

In this paragraph, we present the evolution equations for surface element properties based on the works of Drew \cite{Drew_Geom}. 
Let us mention that the choice in the definition of the interface velocity affects the expression of the evolution laws for 
the curvatures and the area density measure. 
Here we consider the interface velocity defined by \eqref{eq:normal-interface-velocity}. 
Then, the time evolution of the curvatures can be expressed as follows:
\begin{equation}
\begin{array}{rcl}
\dot{\Mcurv}&=&\ds-\frac{1}{2}\nabla^2_T(\vi)-(2\Mcurv^2-\Gcurv)\vi,\\
\dot{\Gcurv}&=&\ds-\Mcurv\nabla^2_T(\vi)+\sqrt{\Mcurv^2-\Gcurv}(\frac{\partial^2\vi}{\partial {y_1^{*}}^2}-\frac{\partial^2\vi}{\partial {y_2^{*}}^2})-2 \, \Mcurv \, \Gcurv \, \vi,\\
\end{array}
\label{eq:time-evol-inter-var}
\end{equation}
where $\dot{\bullet}=\partial_t\bullet+\Vi\cdot\nabla_{\xv}\bullet$, denotes the Lagrangian time derivative, 
$\nabla^2_T=\partial_{y_1}^2+\partial_{y_2}^2$ is the tangential Laplacian operator 
($y_1$ and $y_2$ can be any two orthonormal directions tangential to the surface), 
and $y_k^{*}$ is the coordinate along the principal directions $\Et_k, \; k=1,2$.
 In \ref{app:eq-deltaI}, we derive the equation for $\delta_I$. This equation reads:
\begin{equation}
\label{eq:time-evol-deltaI}
	\dot{\delta}_I =   - \left( \nabla_{\xv} \cdot \Vi   \right)  \delta_I + 2H \, \vi \, \delta_I .
\end{equation}

In equation \eqref{eq:time-evol-deltaI}, one can identify the second term in the right-hand side
with the \textit{stretching rate} $\dot{\Stretch}$, which is defined in Pope \cite{pope} as an intrinsic property. 
For an interface velocity that is normal to the interface, its equation reads:
\begin{equation}
	\dot{\Stretch} = 2H \, \vi.
\label{eq:stretch-factor-equation}
\end{equation}

Another important point in the comparison with Pope's equations is to note
that in \cite{pope}, the interface velocity is defined as the sum of the fluid velocity and
the propagation velocity 
of the flame interface due to chemical reaction (combustion). Consequently, this leads 
to different equations for $\dot{\Mcurv}$, $\dot{\Gcurv}$ and $\dot{\Stretch}$.

%

The system \eqref{eq:time-evol-inter-var} is not closed, because of the second order derivatives of the interfacial velocity 
in the two principal directions. 
For future works, these terms need to be modeled, but this is not the objective of the present work. 
An interpretation of the different terms occurring in the system of equation \eqref{eq:time-evol-inter-var} can be found in \cite{Drew_Geom}.

\subsection{General statistical description of the interface}
\label{sect:general-SDF}

Now, we propose to introduce a general statistical description of a gas-liquid interface, that may be useful
when the exact location of the interface is not known, like in turbulent flows \cite{pope} or in two-phase
transition zones. 
Obviously, the statistical description of the interface can not be restricted to one geometrical variable, 
as is often done for disperse phases. Indeed, when assuming that droplets are spherical,
the information about the droplets radii is enough to describe the interface geometry.
This is however not the case with a general interface, such as one described by 
\eqref{eq:definition-general-surface}. 
In the following, we propose to use the mean and Gauss curvatures, as well as the interface velocity
to characterize the local interface properties and their evolution. These variables will be called the phase-space
variables:
	\[
		\phase{\phasevar}=(\phase{\Mcurv},\phase{\Gcurv},\phase{\Vi}),
	\]
	%
and $\phasespace$ is the set of values that can be attained by any realization $\phase{\phasevar} = \phasevar(t,\xv)$
at time $t$ and position $\xv$.
Moreover we need to define an appropriate probabilistic measure. 
When the interface splits the two phases into a certain number of 
discrete and countable particles, a relevant measure is the NDF.
In the case of a general gas-liquid interface, we adopt a measure based on the area concentration 
of the geometric properties of the interface. 
Let us consider the surface density function (SDF) $\SDF(t,\xv;\phase{\phasevar})$ \cite{pope},
defined as follows: the quantity $\SDF(t,\xv;\phase{\phasevar})\dint{5}{\phase{\phasevar}} \dint{3}{\xv}$ 
measures the probable surface area present in the spatial volume $\dint{3}{\xv}$ around $\xv$ and in the phase-space volume 
$\dint{5}{\phase{\phasevar}}$ around $\phase{\phasevar}$. 
	%
The notion "probable" is used here in the sense of an ensemble averaging over different realizations, 
as defined by Drew \cite{DrewPassman99}. 
To clarify this concept of SDF, let us first give a definition of the SDF for one realization, namely the fine-grain  SDF, $\sdf$,
using the implicit definition of the surface  \eqref{eq:implicit-surface-definition}:
\begin{equation}
	\sdf(t,\xv;\phase{\phasevar})=\delta_I(t,\xv)\delta(\phase{\phasevar}-\phasevar(t,\xv)).
\label{eq:SDF-implicit}
\end{equation}
    
Pope also defines a fine-grained SDF in \cite{pope}, using the explicit definition \eqref{eq:explicit-surface}:
\begin{equation}
	\sdf(t,\xv,\phase{\phasevar})=\int_{\Uopen} L(t,\uparam,\vparam;\phase{\phasevar})A(t,\uparam,\vparam)d\uparam d\vparam,
\label{eq:SDF-explicit}
\end{equation}
where 
$A(t,\uparam,\vparam)$ is the area of a surface element, defined by:
\begin{equation}
	A(t;\uparam,\vparam)=||\partial_u\XV\times\partial_v\XV||(t;\uparam,\vparam)
\label{eq:defintion-A}
\end{equation} 
and
\begin{equation}
	L(t,\uparam,\vparam;\phase{\phasevar})=\delta(\xv-\xv(t,\uparam,\vparam))\delta(\phase{\phasevar}-\phasevar(t,\uparam,\vparam)).
\label{eq:L-SDF}
\end{equation}
One can show that the two definitions are equivalent and independent of the choice of the parametrization and the function $g$.

Secondly, the SDF is defined as an ensemble average over all the realizations:
\begin{equation}
\SDF(t,\xv;\phase{\phasevar})=<\sdf(t,\xv;\phase{\phasevar})>,
\end{equation}
where $<\bullet>$ is the ensemble average operator.

Using the same procedure as in \cite{pope}, we can derive the following transport equation with a source term:
\begin{equation}
  \partial_t \SDF+\nabla_{\xv}\cdot\{\phase{\Vi}\SDF\}+\nabla_{\phase{\phasevar}}\cdot\{<\dot{\phasevar}>_c\SDF\}=<\dot{\Stretch}>_c\SDF,
  \label{eq:SDF-evolution}
\end{equation}
where the \textit{conditional expectation} of a scalar $\psi$ is defined by:
\[
	\ds<\psi>_c=\frac{<\psi \sdf(t,\xv;\phase{\phasevar})>}{\SDF(t,\xv;\phase{\phasevar})}.
\]
The SDF evolves due to two main contributions: the left-hand side of equation \eqref{eq:SDF-evolution} 
contributes to the evolution of $\SDF$ in the phase-space $\phasespace$ 
and in the physical space $\Domain$. This evolution is expressed as divergences of conservative fluxes. 
The source terms of the right-hand side of equation \eqref{eq:SDF-evolution} express the evolution of the 
surface area due to stretching.

The conditional expectations of   $\dot{\Mcurv}$, $\dot{\Gcurv}$ and $\dot{\Stretch}$ may be obtained from
their Lagrangian time evolution \eqref{eq:time-evol-inter-var} by applying the ensemble averaging and by using 
the linearity of the average operator, which reads ($\lambda$ is a constant):
\[
	<\lambda a + b> = \lambda <a> + <b>,
\]
as well as Gauss and Leibniz rules:
\[
	< \partial_t a > = \partial_t <a> \quad \text{and} \quad <\partial_x a> = \partial_x <a>.
\]
Then, the averaged evolution equations read:
\begin{equation}
	\begin{array}{rcl}
		<\dot{\Mcurv}>_c	&=&	-\ds\frac{1}{2}<\nabla^2_T(\vi)>_c
											\, - \, (2\phase{\Mcurv}^2-\phase{\Gcurv}) \, (\phase{\Vi}\cdot<\normal>_c),\\
		<\dot{\Gcurv}>_c	&=&	-\ds \phase{\Mcurv}<\nabla^2_T(\vi)>_c 
											\, + \, \sqrt{\phase{\Mcurv}^2-\phase{\Gcurv}} \, \left<\frac{\partial^2\vi}{\partial {y_1^{*}}^2}
											-\frac{\partial^2\vi}{\partial {y_2^{*}}^2} \right>_c
											\, - \, 2\phase{\Mcurv}\phase{\Gcurv} \, (\phase{\Vi}\cdot<\normal>_c),\\
		<\dot{\Stretch}>_c	&=&	2\phase{\Mcurv} \, (\phase{\Vi}\cdot<\normal>_c).
	\end{array}
\label{eq:cond-time-evol-inter-var}
\end{equation} 

One can note that some terms are unclosed in these equations, that are the conditional expectations of:
\begin{itemize}
	\item the normal vector $<\normal>_c$,
	\item the second order derivatives of the interfacial velocity $<\nabla^2_T(\vi)>_c$,
	\item the evolution of the interfacial velocity $<\dot{\Vi}>_c$.
\end{itemize}
These terms need to be modeled and related to the internal flow dynamics of the gas and liquid phases.

\subsection{Averaged quantities and moments of the SDF}

The numerical resolution of equation \eqref{eq:SDF-evolution} can be unreachable for most applications  
because of the large dimension of the phase space $\phasespace$. 
In fact, solving the exact SDF would provide a level of detail on the flow, which is often
not needed.  
Moreover, the large amount of data produced to attain this level of detail may hide the prevailing macroscopic features 
we are looking for. Therefore, we only aim at predicting some macroscopic variables of the flow, that ensure 
a satisfactory result for most industrial applications. 
Drew \cite{Drew_Geom} derived Eulerian equations for the evolution of the following averaged quantities: 
the expected surface density $\Sigma(t,\xv)$,
the interfacial-expected mean and Gauss curvatures, $\widetilde{\Mcurv}$, $\widetilde{\Gcurv}$ and 
the volume fraction $\alpha(t,\xv)$. 
In the present approach, we can not express the volume fraction as a function of the SDF 
without making a topological assumption on the gas-liquid interface. 
Therefore, in this part we temporarily restrict our study to the first three  averaged interfacial quantities.

In terms of moments of the SDF, $\Sigma(t,\xv)$, $\widetilde{\Mcurv}$ and $\widetilde{\Gcurv}$ 
can be expressed as follows:
\begin{equation}
  \begin{array}{rcl}
    \Sigma(t,\xv)&=&M_{0,0,\tens{0}}(t,\xv),\\
    \Sigma(t,\xv)\widetilde{\Mcurv}&=&M_{1,0,\tens{0}}(t,\xv),\\
    \Sigma(t,\xv)\widetilde{\Gcurv}&=& M_{0,1,\tens{0}}(t,\xv),\\
  \end{array}
  \label{eq:expected_interfacial_variables}
\end{equation}
where the moments read:
\begin{equation}
M_{i,j,\tens{l}}(t,\xv)=\int_{\phasespace}\Mcurv^{i}\Gcurv^{j}\Vic{x}^{l_x}\Vic{y}^{l_y}\Vic{z}^{l_z}\SDF(t,\xv;\phase{\phasevar})\dint{5}{\phase{\phasevar}},
\end{equation}
and $\tens{l}=(l_x,l_y,l_z)$.

Then, by multiplying the SDF evolution equation \eqref{eq:SDF-evolution} by 
$\Mcurv^{i}\Gcurv^{j}\Vic{x}^{l_x}\Vic{y}^{l_y}\Vic{z}^{l_z}$ and integrating
over the whole phase-space domain,
we can derive the evolution equations for the three quantities of interest:
\begin{equation}
  \begin{array}{rcl}
    \partial_t\Sigma+\diverg{\xv}{ \left\{\Sigma \, \overline{\Vi} \right\}}
    		&=&
    		\int_{\phase{\phasevar}}<\dot{\Stretch}>_cF(t,\xv;\phase{\phasevar}) \,
		\dint{5}{\phase{\phasevar}},\\[10pt]
    \partial_t\Sigma\widetilde{\Mcurv}+\diverg{\xv}{ \left\{\Sigma \, \widetilde{\Mcurv} \, \overline{\Vi} \right\}}
    		&=&
    		\diverg{\xv}{ \left\{\Sigma\widetilde{\Mcurv}(\overline{\Vi}-\overline{\Vi}^H) \right\}}
		+
    		\int_{\phase{\phasevar}}<\dot{\Stretch}+\dot{H}/\phase{H}>_c\phase{\Mcurv}\SDF(t,\xv;\phase{{\phasevar}}) \,
    		\dint{5}{\phase{\phasevar}},\\[10pt]
    \partial_t\Sigma\widetilde{\Gcurv}+\diverg{\xv}{ \left\{\Sigma \, \widetilde{\Gcurv} \, \overline{\Vi} \right\}}
    		&=&
		\diverg{\xv}{\left\{ \Sigma\widetilde{\Gcurv}(\overline{\Vi}-\overline{\Vi}^G) \right\}}
		+
    		\int_{\phase{\phasevar}}<\dot{\Stretch}+\dot{G}/\phase{G}>_c\phase{\Gcurv}\SDF(t,\xv;\phase{{\phasevar}}) \,
    		\dint{5}{\phase{\phasevar}}.\\
  \end{array}
  \label{eq:averaged-Drew-eq-general-case}
\end{equation}
We can notice three types of averaged interfacial velocity that read:
\begin{equation}\label{eq:expected_vi}
  \begin{array}{rcl}
    \Sigma \, \widetilde{\Vi}
    &=&
    \int_{\phase{\phasevar}}\phase{\Vi}\SDF(t,\xv;\phase{{\phasevar}}) \, \dint{5}{\phase{\phasevar}},\\[9pt]
    \Sigma \, \widetilde{\Mcurv} \, \overline{\Vi}^H
    &=&
    \int_{\phase{\phasevar}}\phase{\Vi}\phase{\Mcurv}\SDF(t,\xv;\phase{{\phasevar}}) \, \dint{5}{\phase{\phasevar}},\\[9pt]
    \Sigma \, \widetilde{\Gcurv} \, \overline{\Vi}^G 
    &=&
    \int_{\phase{\phasevar}}\phase{\Vi}\phase{\Gcurv}\SDF(t,\xv;\phase{{\phasevar}}) \, \dint{5}{\phase{\phasevar}}.
  \end{array}
\end{equation}

Using the time evolution of both curvatures and of the stretch factor \eqref{eq:time-evol-inter-var} 
in the system of equations \eqref{eq:averaged-Drew-eq-general-case}, 
we can show that this system is equivalent to the system derived by Drew in \cite{Drew_Geom}. 
However, two closure issues need to be tackled at two different stages for future works. 
First, we need time evolution laws to close the conditional averages 
$(<\dot{\Mcurv}>_c,<\dot{\Gcurv}>_c,<\dot{\Stretch}>_c)$, as already mentioned for the transport equation of the SDF. 
Then, we need to propose a procedure to reconstruct the SDF from the known moments. 
Let us underline that, by using only the first three moments $M_{0,0,\tens{0}}$, $M_{1,0,\tens{0}}$, $M_{0,1,\tens{0}}$, 
we can not capture the variance of the phase-space coordinates. 
One should consider higher order moments. 
However, one must be careful when reconstructing the SDF in this case. 
Indeed, the curvatures $(\phase{\Mcurv},\phase{\Gcurv})$ must satisfy 
$\phase{\Mcurv}^2\geq\phase{\Gcurv}$ (see
equation \eqref{eq:mapping}), and the SDF reconstruction has to respect this constraint.

\section{Statistical description of the gas-liquid interface for a disperse phase}



\subsection{Surface density function in the context of discrete particles}
\label{sect:discrte-SDF}

The SDF measures the pointwise probable area concentration for a probabilistic event, 
which is characterized by the phase space variables $\phase{\phasevar}=\{\phase{\Mcurv},\phase{\Gcurv},\phase{\Vi}\}\in\phasespace$
and
evaluated at a local point $\xv\in\Domain$. 
This probabilistic description is a pointwise description of the interface, 
which makes it general and valid for both separated and disperse phases. 
However, in the case of a disperse flow, a point-particle approach is usually considered. 
Each particle, droplet or bubble, is reduced to a point, characterized by some averaged and global 
quantities, such as: the velocity of the center of mass, the particle size (measured by volume, surface area or diameter), etc... and the statistics are conducted on a number of objects.
In order to adapt the statistical approach presented in the previous section to a discrete formalism 
for a disperse phase, we first define interfacial quantities that are averaged over the surface of each 
particle. Then, we define a \textit{discrete surface density function} that may characterize the dispersion
in size and in shape of the population of particles.

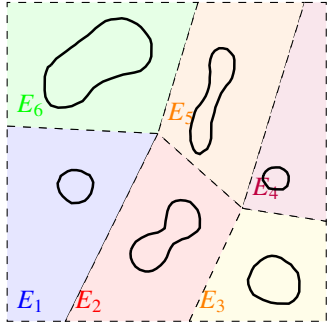
\begin{figure}
\begin{center}
	\begin{tikzpicture}[scale=0.6]

\draw[black, dashed, fill=blue!10] (0,0.601368*7.2) -- (0.463867*7.2,0.578227*7.2) -- (0.179284 *7.2,0) -- (0,0) -- cycle;
\draw[black, dashed, fill=green!10] (0,0.601368*7.2) -- (0.463867*7.2,0.578227*7.2) -- (0.587531*7.2,7.07) -- (0,7.07) -- cycle;
\draw[black, dashed, fill=orange!10] (0.587531*7.2,7.07) -- (0.463867*7.2,0.578227*7.2) -- (0.723886*7.2,0.348004*7.2) -- (0.914976*7.2, 7.07)--cycle;
\draw[black, dashed, fill=red!10] (0.463867*7.2,0.578227*7.2) -- (0.179284 *7.2,0)  -- (0.560646*7.2,0) --  (0.723886*7.2,0.348004*7.2) -- cycle;
\draw[black, dashed, fill=yellow!10] (0.560646*7.2,0) --  (0.723886*7.2,0.348004*7.2)  -- (7.07,0.308141*7.2)-- (7.07,0) -- cycle;
\draw[black, dashed, fill=purple!10] (0.723886*7.2,0.348004*7.2)  -- (7.07,0.308141*7.2) -- (7.07,7.07) -- (0.914976*7.2, 7.07) --  cycle;

\node[above right, blue] at (0,0) {$E_1$};
\node[above right, red] at (0.179284*7.2,0) {$E_2$};
\node[above right, orange] at (0.560646*7.2,0) {$E_3$};
\node[above right, purple] at (0.723886*7.2, 0.348004*7.2) {$E_4$};
\node[above right, orange] at (0.463867*7.2,0.578227*7.2) {$E_5$};
\node[above right, green] at (0,0.601368*7.2) {$E_6$};


\begin{axis}[
unit vector ratio*=1 1 1,
width=10cm,
hide axis,
xmin=0, xmax=1,
ymin=0, ymax=1,
xtick={0, 1},
xticklabels={,},
ytick={0, 1},
yticklabels={,},
]
\addplot [ultra thick, black]
table {%
0.137788 0.806551
0.154203 0.826625
0.170520 0.838590
0.186837 0.850555
0.200443 0.861202
0.214033 0.870498
0.224927 0.879826
0.239881 0.890457
0.253470 0.899752
0.267092 0.911751
0.278003 0.922430
0.295684 0.935731
0.309274 0.945026
0.326857 0.950219
0.341696 0.951390
0.360593 0.953863
0.376731 0.950963
0.392819 0.944010
0.408891 0.935705
0.420905 0.926098
0.430206 0.915173
0.440823 0.901528
0.450092 0.887900
0.453955 0.872986
0.453759 0.856770
0.453481 0.833797
0.447895 0.817646
0.435571 0.801577
0.419237 0.788260
0.400307 0.783084
0.386799 0.780546
0.365190 0.776754
0.344896 0.770243
0.325949 0.763716
0.301579 0.754551
0.286642 0.745272
0.271689 0.734642
0.264853 0.726615
0.252529 0.710547
0.237494 0.693159
0.219879 0.685264
0.203610 0.677353
0.183349 0.673545
0.160425 0.672472
0.137599 0.679508
0.126982 0.693152
0.120408 0.706748
0.117958 0.727051
0.118236 0.750024
0.118481 0.770294
0.122703 0.785110
0.128273 0.799909
0.137788 0.806551
};
\addplot [ultra thick, black]
table {%
0.577091 0.579553 
0.575988 0.599840 
0.580309 0.622764 
0.585927 0.641617 
0.590166 0.657784 
0.595703 0.669881 
0.606630 0.681912 
0.616226 0.695310 
0.624475 0.708726 
0.630028 0.722173 
0.634234 0.735638 
0.637109 0.750470 
0.637322 0.768037 
0.638882 0.785589 
0.636432 0.805892 
0.639356 0.824778 
0.644942 0.840928 
0.653158 0.851641 
0.666731 0.859585 
0.678858 0.859438 
0.692301 0.856571 
0.704281 0.844261 
0.712252 0.834704 
0.714800 0.822509 
0.714636 0.808995 
0.714457 0.794130 
0.711598 0.780649 
0.707376 0.765834 
0.700491 0.753754 
0.689564 0.741723 
0.681332 0.729659 
0.673100 0.717595 
0.667546 0.704147 
0.660596 0.686661 
0.655043 0.673213 
0.652168 0.658381 
0.647946 0.643566 
0.642327 0.624712 
0.642147 0.609847 
0.641984 0.596334 
0.641738 0.576064 
0.637500 0.559897 
0.629202 0.542427 
0.618340 0.535802 
0.606098 0.526490 
0.594020 0.530691 
0.586065 0.541600 
0.580855 0.556531 
0.576960 0.568742
0.577091 0.579553 
};
\addplot [ultra thick, black]
table {%
0.164416 0.445363 
0.168606 0.457476 
0.178120 0.464118 
0.190346 0.472079 
0.202523 0.475985 
0.218693 0.475789 
0.233483 0.472906 
0.250886 0.463233 
0.258938 0.460432 
0.266860 0.446820 
0.270788 0.437312 
0.269261 0.422463 
0.263625 0.402258 
0.254029 0.388860 
0.244482 0.379515 
0.225536 0.372988 
0.202611 0.371915 
0.187903 0.381555 
0.174575 0.393881 
0.165290 0.406158 
0.160096 0.422439 
0.158895 0.434618
0.164416 0.445363 
};
\addplot [ultra thick, black]
table {%
0.389514 0.225025 
0.396415 0.238457 
0.404647 0.250521 
0.422312 0.262470 
0.433157 0.267744 
0.444003 0.273019 
0.460222 0.276876 
0.473730 0.279415 
0.485939 0.286024 
0.495454 0.292666 
0.496915 0.302109 
0.498426 0.315606 
0.501301 0.330438 
0.504192 0.346622 
0.509762 0.361421 
0.519342 0.373468 
0.530187 0.378743 
0.545043 0.381265 
0.561196 0.379717 
0.575970 0.375483 
0.587950 0.363173 
0.598583 0.350880 
0.603810 0.337301 
0.606325 0.322403 
0.606162 0.308889 
0.605982 0.294025 
0.597766 0.283312 
0.582813 0.272682 
0.569223 0.263386 
0.550293 0.258210 
0.538083 0.251601 
0.523163 0.243673 
0.514931 0.231609 
0.509410 0.220864 
0.502443 0.202027 
0.495559 0.189947 
0.483268 0.176581 
0.469662 0.165934 
0.457453 0.159325 
0.445292 0.156769 
0.434512 0.156901 
0.417060 0.162519 
0.403699 0.172142 
0.394430 0.185771 
0.389204 0.199350 
0.388020 0.212880
0.389514 0.225025 
};
\addplot [ultra thick, black]
table {%
0.806109 0.460534 
0.812977 0.471263 
0.822508 0.479256 
0.831990 0.483196 
0.845481 0.484383 
0.857609 0.484235 
0.871002 0.477315 
0.880337 0.469092 
0.885580 0.456864 
0.885400 0.441999 
0.881161 0.425832 
0.867555 0.415185 
0.859470 0.415284 
0.841953 0.415497 
0.823120 0.418429 
0.813835 0.430706 
0.804566 0.444334 
0.804647 0.451091
0.806109 0.460534 
};
\addplot [ultra thick, black]
table {%
0.763190 0.143442 
0.770140 0.160927 
0.777025 0.173008 
0.787920 0.182336 
0.796103 0.190346 
0.808328 0.198306 
0.819174 0.203581 
0.832682 0.206119 
0.846140 0.204604 
0.862245 0.199002 
0.876937 0.188011 
0.892993 0.178355 
0.900947 0.167446 
0.914242 0.152417 
0.916790 0.140222 
0.919322 0.126676 
0.919126 0.110459 
0.918913 0.0928917 
0.911980 0.0767575 
0.901052 0.0647264 
0.886115 0.0554473 
0.871293 0.0556276 
0.849733 0.0558899 
0.829553 0.0588384 
0.805395 0.0672416 
0.789323 0.0755464 
0.774599 0.0838348 
0.765346 0.0988144 
0.758821 0.116464 
0.757653 0.131345
0.763190 0.143442 
};

\end{axis}

\end{tikzpicture}
	\caption{Illustration of the spatial decomposition, with subspaces containing
	only one particle.}
	\label{fig:space-illustr}
\end{center}
\end{figure}

Let us consider an isolated particle $\particle_k$ of
surface $\portsurf{k}$, supposed to be smooth. Let us also consider 
a partition of the domain $\Domain=\bigcup\limits_{k=1}^{N_{max}}\subdomain_k$, 
such that each sub-domain $\subdomain_k$ contains exactly one particle $\particle_k$,
as illustrated in Figure \ref{fig:space-illustr}. 
For a quantity $\phi(t,\xv)$, that can be a scalar (like curvatures) or a vector (like a velocity), 
we define its interface average $\Aver{\phi}_k(t,\xv)$ over the particle surface $\portsurf{k}$ as follows:
\begin{equation}
  \Aver{\phi}_k(t,\xv_k)=\dfrac{1}{S_k}\int_{\xv\in\subdomain_k}\phi(t,\xv)\delta_I(t,\xv)\dint{3}{\xv},
  \label{eq:average-over-p-surf-implicit}
\end{equation}
where $S_k$ is the total surface area of the particle:
\[
	S_k = \int_{\xv\in\subdomain_k}\delta_I(t,\xv)\dint{3}{\xv},
\]
and $\xv_k$ its  center of mass. 


From now on, we consider each particle $\particle_k$  as punctual, located at its  center of mass $\xv_k(t)$, 
having a surface area $S_k$ and the averaged interface properties 
$\Aver{\phasevar}_k=\{\Aver{\Mcurv}_k,\Aver{\Gcurv}_k,\Aver{\Vi}_k\}$. 
We define the discrete SDF by:
\begin{equation}
  \SDF^{d}(t,\xv,\phase{{\phasevar}})
  =
  <\sum_{k=1}^{N_{max}}S_k(t)\delta(\xv-\xv_k(t))\delta(\phase{{\phasevar}}-\Aver{\phasevar}_k(t))>
  \label{eq:DSDF}
\end{equation}
where $N_{max}$ is the maximum number of particles in the domain $\Domain$ over all realizations. 
One can note that we have chosen to locate the particles at their centers of mass.
However, the surface averaging procedure involves a surface barycenter 
defined by $\widetilde{\xv}_k=1/S_k\int_{\portsurf{k}}\xv dS(\xv)$. In order to cope with our choice, which is more
practical as far as the particle dynamics is considered, we propose to express the mean interface velocity
as:
\[
	\Aver{\Vi}_k(t)=\dfrac{d\xv_k(t)}{dt} + \tilde{V}_k (t),
\]
where $\tilde{V}_k (t)$ is a fluctuation velocity, that stands for the particle shape deformations.
%
In the following, this fluctuation velocity is neglected and we simply consider that
the averaged interface velocity of a particle $\particle_k$ reads:
\begin{equation}
  \Aver{\Vi}_k(t)=\dfrac{d\xv_k(t)}{dt}.
\end{equation}

The discrete SDF \eqref{eq:DSDF} verifies a similar transport equation as \eqref{eq:SDF-evolution}:
\begin{equation}
\partial_t \SDF^d+\nabla_{\xv}\cdot\{\phase{{\Vi}}\SDF^d\}+\gradient{\phase{\phasevar}}{\cdot\{<\dot{\Aver{\phasevar}}}>_c\SDF^d\}=<\dot{\Aver{\Stretch}}>_c\SDF^d.
\label{eq:GBPE-DSDF}
\end{equation}
Once more, the time evolution of the averaged curvatures needs to be modeled in order to close the system of equations.
This is work in progress and is out of the scope of the present paper.

The "localized" SDF defined in section \ref{sect:general-SDF} and the discrete SDF are two different functions.
Nonetheless, they contain similar pieces of information about the gas-liquid interface properties. 
Indeed, both functions provide the probable surface area of the interface having some geometrical properties 
given by the phase-space variables. For the "localized" SDF, the phase-space variables are given by 
the localized interface properties at a surface point, while for the discrete SDF, we have considered 
interface properties averaged by object (particle). 
The link between the two functions can be seen through the first order moments of the two functions:
\begin{proposition}
\label{prop:TopologicalMoments}
Given a subdomain $\subdomain\subset\Domain$, such that its border 
does not cut any particle, the integral over this subdomain of the  zeroth and first order moments of the two functions are equal:
\begin{equation}
  \begin{array}{rlc}
    \int_{\xv\in\subdomain}\int_{\phase{\phasevar}\in\phasespace}
      \SDF(t,\xv;\phase{\phasevar}) \,
    \dint{5}{\phase{\phasevar}}\dint{3}{\xv}
    &=&
    \int_{\xv\in\subdomain}\int_{\phase{\phasevar}\in\phasespace}
      \SDF^{d}(t,\xv;\phase{\phasevar}) \,
    \dint{5}{\phase{\phasevar}}\dint{3}{\xv}
    \\[10pt]
    \int_{\xv\in\subdomain}\int_{\phase{\phasevar}\in\phasespace}
      \phase{\phasevar}_l\SDF(t,\xv;\phase{\phasevar}) \,
    \dint{5}{\phase{\phasevar}}\dint{3}{\xv}
    &=&
    \int_{\xv\in\subdomain}\int_{\phase{\phasevar}\in\phasespace}
      \phase{{\phasevar}}_l\SDF^{d}(t,\xv;\phase{\phasevar}) \,
    \dint{5}{\phase{\phasevar}}\dint{3}{\xv}
  \end{array}
\end{equation}
where $ \phase{\phasevar}_l \in \left\{ \phase{\Mcurv}, \phase{\Gcurv}, \phase{\vi}_x, \phase{\vi}_y, \phase{\vi}_z \right\}$.

\end{proposition}
\begin{proof}
We illustrate this result for a moment of order one on $\Gcurv$ and zero on the other phase variables 
($\phase{\phasevar}_l = \phase{\Gcurv}$):
\begin{equation}
  \begin{array}{rcl}
    \int_{\xv\in \subdomain}\int_{\phase{\phasevar}\in\phasespace}\phase{\Gcurv}
      \SDF(t,\xv;\phase{\phasevar}) \,
      \dint{5}{\phase{\phasevar}}\dint{3}{\xv}
    &=&
    <\int_{\xv\in \subdomain}\Gcurv(t;\xv) \, \delta_I(\xv) \, \dint{3}{\xv}>\\[9pt]
    &=&
    <\sum\limits_{k}\int\limits_{\xv\in\subdomain_k\cap\subdomain}\Gcurv(t;\xv) \, \delta_I(\xv) \, \dint{3}{\xv}>\\[10pt]
    &=&
    <\sum\limits_{\{k;\; \xv_k(t)\in\subdomain\}}S_k\Aver{\Gcurv}_k>\\[10pt]
    &=&
    <\sum\limits_{\{k;\; \xv_k(t)\in\subdomain\}}\int_{\xv\in\subdomain_k}\int_{\phase{\phasevar}\in\phasespace}
      \phase{{\Gcurv}}\SDF^{d}(t,\xv,\phase{{\phasevar}}) \,
    \dint{5}{\phase{\phasevar}} \dint{3}{\xv}>\\[10pt]
    &=&
    \int_{\xv\in \subdomain}\int_{\phase{{\phasevar}}\in\phasespace}
      \phase{{\Gcurv}}\SDF^{d}(t,\xv,\phase{{\phasevar}}) \,
    \dint{5}{\phase{{\phasevar}}} \dint{3}{\xv}.
  \end{array}
\end{equation}
\end{proof}

We recall that the zeroth and first-order moments are related to the expected surface density and the two curvatures 
$(\Sigma,\Sigma\widetilde{\Mcurv},\Sigma\widetilde{\Gcurv}$), defined in \eqref{eq:expected_interfacial_variables}, 
and to the expected interface velocity, defined in \eqref{eq:expected_vi}. 
In this work, we give a great importance to these quantities because they contain some topological information 
about the gas-liquid interface, such as the number density of particles, as we will show in next section. 
From now on, we call these quantities the "\emph{topological moments}" 
and we underline that any future definition of a new SDF should preserve these moments.

\subsection{Link between the discrete SDF and a NDF: derivation of a Williams-Boltzmann-like equation}
\label{sec:link-discrete-ndf}

Let us now present the relation between the discrete SDF  \eqref{eq:DSDF}  and a number density function (NDF) for particles, 
using a geometrical property of the averaged Gauss curvature, known as the \emph{Gauss-Bonnet formula}. 

Let $M$ be a 3D bounded object and $\Surface(M)$ its surrounding surface, supposed to be smooth. 
The Gauss-Bonnet formula applied to $\Surface(M)$ states that \cite{gray}:
\begin{equation}
  \int_{\xv\in\Surface(M)}\Gcurv(\xv)dS(\xv)=2\pi\,\chi(M),
\end{equation}
where $\chi(M)$ is the Euler characteristic of the surface $\Surface(M)$
and is a topological invariant, meaning that two homeomorphic surfaces have the same Euler characteristic. 
In the following, we consider smooth particle surfaces, that are all homeomorphic to a sphere, the Euler characteristic of
which is 2. 
Therefore, we can relate the averaged Gauss curvature over the surface of a particle $\particle_k$ to its surface area by: 
\begin{equation}
  S_k=\dfrac{4\pi}{\Aver{\Gcurv}_k}.
\label{eq:particle-gauss-bonnet}
\end{equation}
Since $\ds\frac{S_k\Aver{\Gcurv}_k}{4\pi}=1$ for each particle, we are able to count the total number of particles, just by looking at 
the interface average of $\Gcurv$. Then, we can derive a relation between the discrete SDF $\SDF^d$ 
and a Number Density Function (NDF) $\NDF$. 
This relation is given by the following proposition.
\begin{proposition}
When the two-phase flow is purely disperse, 
the distribution function $\dfrac{\phase{{\Gcurv}}}{4\pi}\SDF^{d}(t,\xv;\phase{{\phasevar}})$ is equal to the NDF 
$\NDF(t,\xv;\phase{{\phasevar}})$ 
for a phase space made of the averaged curvatures and interface velocity 
$\phase{{\phasevar}}=\{\phase{{\Mcurv}},\phase{{\Gcurv}},\phase{{\Vi}}\}$:
\begin{equation}
  \dfrac{\phase{{\Gcurv}}}{4\pi}\SDF^{d}(t,\xv;\phase{{\phasevar}})
  =
  \NDF(t,\xv;\phase{{\phasevar}})
\end{equation}
\end{proposition}
\begin{proof}
\begin{equation}
  \begin{array}{rcl}
    \dfrac{\phase{{\Gcurv}}}{4\pi}\SDF^{d}(t,\xv;\phase{{\phasevar}})
    &=& \ds
    <\sum_{k=1}^{N_{max}}S_k(t)\delta(\xv-\xv_k(t))\left\{
      \dfrac{\phase{{\Gcurv}}}{4\pi}
    \right\}\delta(\phase{{\phasevar}}-\Aver{\phasevar}_k(t)))>,
    \\[10pt]
    &=& \ds
    <\sum_{k=1}^{N_{max}}S_k(t)\delta(\xv-\xv_k(t))\left\{
      \dfrac{\Aver{\Gcurv}_k(t)}{4\pi}
    \right\}\delta(\phase{{\phasevar}}-\Aver{\phasevar}_k(t))>,
    \\[10pt]
    &=& \ds
    <\sum_{k=1}^{N_{max}}\delta(\xv-\xv_k(t))\delta(\phase{{\phasevar}}-\Aver{\phasevar}_k(t))>,\\[9pt]
    &=&
    \NDF(t,\xv;\phase{{\phasevar}}).
  \end{array}
\end{equation}
\end{proof}

From the transport equation of the discrete SDF \eqref{eq:GBPE-DSDF}, we can derive a General Popolation Balance Equation (GPBE) for 
$\NDF(t,\xv;\phase{{\phasevar}})$, that is similar to the Williams-Boltzmann equation:
\begin{equation}
  \partial_t \NDF
  +
  \nabla_{\xv}\cdot\{\phase{{\Vi}}\NDF\}
  +
  \nabla_{\phase{\phasevar}}\cdot\{<\dot{\Aver{\phasevar}}>_c\NDF\}
  =
  \left(<\dot{\Aver{\Stretch}}>_c+<\dfrac{\dot{\Aver{\Gcurv}} }{\Aver{\Gcurv}}>_c\right)\NDF.
\label{eq:GBPE-NDF-ST}
\end{equation}
The right-hand term behaves as a source term. We can relate this term to topological evolutions of the Euler characteristic of droplets due to fragmentation, coalescence, ects. 

Let us now consider a set of droplets that do not break up, nor coalesce and which are all the time 
homeomorph to spheres. 
Using the Gauss-Bonnet result \eqref{eq:particle-gauss-bonnet}, we can show that:
\begin{equation}
  \dot{\Aver{\Stretch}}_k
  = 
  -\dfrac{\dot{\Aver{\Gcurv}}_k(t)}{\Aver{\Gcurv}_k(t)}.
\end{equation}
Therefore, without break-up nor coalescence, equation \eqref{eq:GBPE-NDF-ST} reduces to:
\begin{equation}
  \partial_t \NDF
  +
  \nabla_{\xv}\cdot\{\phase{{\Vi}}\NDF\}
  +
  \nabla_{\phase{\phasevar}}\cdot\{<\dot{\Aver{\phasevar}}>_c\NDF\}
  =0.
\label{eq:GBPE-NDF}
\end{equation}

Equation \eqref{eq:GBPE-NDF} represents the time evolution of the particles NDF due to transport in physical space 
at the velocity of their center of mass, 
and the evolution of the internal variables, which represents the particle surface deformations and their acceleration by 
interaction with the carrier gas.

Now, we consider the volume of particles as an additional variable of the phase-space. 
Let us underline that adding the volume to the phase-space variables is specific to particles 
and can not be easily conducted within the general statistical approach for gas-liquid interface, as introduced 
in section \ref{sect:general-SDF}. 
Furthermore, we also replace the averaged Gauss curvature by the surface area of the particles, 
since the two variables are related by \eqref{eq:particle-gauss-bonnet}. 
In the following, we denote the new phase-space variables by $\phasevarv=\{\Aver{\Mcurv},\Areadrop,\vol,\Aver{\Vi}\}$,
where $\Areadrop=4\pi/\Aver{\Gcurv}$ is the surface area of the particles and $\vol$ is their volume. 
A new NDF may be defined as follows:
\begin{equation}
  \NDF(t,\xv,\phase{\phasevarv})=<\sum_{k=1}^{N_{max}}\delta(\xv-\xv_k(t))\delta(\phase{\phasevarv}-\phasevarv_k(t))>.
\end{equation}

Next, we focus on the evolution of the expected-average terms of the gas-liquid interface: the interface density area $\Sigma(t,\xv)$, 
the mean and Gauss curvatures, $\widetilde{\Mcurv}(t,\xv)$ and $\widetilde{\Gcurv}(t,\xv)$, introduced in section \ref{sect:general-SDF}
and the volume fraction $\alpha(t,\xv)$. In the current context, these terms can be expressed as first order moments of the NDF as follows:
\begin{equation}\label{eq:Geom-NDF}
  \begin{array}{rcl}
    \alpha(t,\xv)&=&M_{0,0,1,\tens{0}}(t,\xv),\\
    \Sigma(t,\xv)&=&M_{0,1,0,\tens{0}}(t,\xv),\\
    \Sigma(t,\xv)\tilde{\Mcurv}&=&M_{1,1,0,\tens{0}}(t,\xv),\\
    \Sigma(t,\xv)\tilde{\Gcurv}&=&4\pi M_{0,0,0,\tens{0}}(t,\xv),\\
  \end{array}
\end{equation}
where the moments of the NDF now read:
\begin{equation}\label{eq:Moments-NDF}
  M_{i,j,k,\tens{l}}(t,\xv)
  =
  \int_{\phase{\phasevarv}\in\phasespace}
    \Mcurv^{i}\Areadrop^{j}\vol^{k}\Vic{x}^{l_x}\Vic{y}^{l_y}\Vic{z}^{l_z}\NDF(t,\xv;\phase{\phasevarv}) \,
  \dint{5}{\phase{\phasevarv}},
\end{equation}
and $\tens{l}=(l_x,l_y,l_z)$. By considering the corresponding moments of the dynamics of the NDF \eqref{eq:GBPE-NDF}, we can obtain
the equations of the three expected interfacial variables and the volume fraction:
\begin{equation}
  \begin{array}{rcl}
    \partial_t\alpha + \partial_{\xv}\cdot\{\alpha(t,\xv)\widetilde{\Vi}\}
    &=& \ds
    \diverg{\xv}{\{\alpha(t,\xv)(\widetilde{\Vi}-\widetilde{\Vi}^{\alpha})\}}
    +
    \int_{\phase{\phasevarv}}<\dot{\vol}>_c\NDF(t,\xv;\phase{\phasevarv})\dint{5}{\phase{\phasevarv}},
    \\[10pt]
    \partial_t\Sigma + \partial_{\xv}\cdot\{\Sigma(t,\xv)\widetilde{\Vi}\}
    &=& \ds 
    \int_{\phase{\phasevarv}}<-\dot{\Areadrop}>_c\NDF(t,\xv;\phase{\phasevarv})\dint{5}{\phase{\phasevarv}},
    \\[10pt]
    \partial_t\Sigma\widetilde{\Mcurv} + \partial_{\xv}\cdot\{\Sigma\widetilde{\Mcurv}\widetilde{\Vi}\}
    &=& \ds 
    \diverg{\xv}{\{\Sigma\widetilde{\Mcurv}(\widetilde{\Vi}-\widetilde{\Vi}^{\Mcurv})\}}
    +
    \int_{\phase{\phasevarv}}<-\dot{\Areadrop}\phase{H}+\dot{H}\phase{\Areadrop}>_c
      \NDF(t,\xv;\phase{\phasevarv})
    \dint{5}{\phase{\phasevarv}},
    \\[10pt]
    \partial_t\Sigma\widetilde{\Gcurv} + \partial_{\xv}\cdot\{\Sigma\widetilde{\Gcurv}\widetilde{\Vi}\}
    &=& \ds
    \diverg{\xv}{\{\Sigma\widetilde{\Gcurv}(\widetilde{\Vi}-\widetilde{\Vi}^{\Gcurv})\}}.
  \end{array}
  \label{eq:geometrical_moments_equations}
\end{equation}
The divergence terms appearing in the right-hand side of the system express the correlation between the velocity 
and the rest of the phase space variables. 
In the case where the velocity is uncorrelated with $\{\vol,\Mcurv,\Gcurv\}$, these terms vanish. 
Under this hypothesis, the only remaining source terms are related to the shape deformation and to the compressibility of the particles. 
Also, the last equation would have no source term: the deformation of particles does not affect the quantity $\Sigma\widetilde{\Gcurv}$. 
Indeed, according to Gauss-Bonnet this quantity is the expected number density of particles, the evolution of which is only
due to convection or coalescence and break-up of the particles. Since we have supposed that the particles do not break-up 
nor coalesce, the particles number density is simply convected at velocity $\widetilde{\Vi}$.

In the following section, we propose a simplified closed model of this system in the case of a polydisperse evaporating spray, 
where the droplets are supposed to remain spherical at any time.

\section{High order geometrical size moments for polydisperse evaporating sprays of spherical droplets}

From now on, we omit the subscript $\phase{\bullet}$ for the phase-space variables.

\subsection{Spray modeling framework}

In this part, we focus on the simplified case of an evaporating spray, which consists 
in a cloud of polydiserse and spherical droplets. We also suppose that the spray is dilute enough and the Weber number low enough, 
so that coalescence and fragmentation of the droplets can be neglected. 
Such configuration occurs downstream the fuel injection in direct injection engines. 
Finally, for the sake of simplicity and clarity, we assume that thermal transfer can also be neglected, 
so that the temperature of the droplets can be ignored. Under the hypothesis that the droplets stay spherical at all times, the NDF writes:
\begin{equation}
  \NDF(t,\xv;\Mcurv,\Areadrop,\vol,\Vi)
  =
  \NNDF(t,\xv;\Areadrop,\Vi) \, \delta(\Mcurv-\conditionvar{\Mcurv}(\Areadrop)) \, \delta(\vol-\conditionvar{\vol}(\Areadrop)),
\end{equation}
where the mean curvature, $\conditionvar{\Mcurv}(\Areadrop)=\sqrt{4\pi/\Areadrop}$, 
and the volume of the drops, $\conditionvar{\vol}(\Areadrop)=\Areadrop^{3/2}/(6\sqrt{\pi})$, are two functions of the surface $\Areadrop$. 
Therefore, we can reduce the phase-space variables to the surface area $\Areadrop$ and the averaged  interface velocity $\Vi$. 
Let us also note that under the assumption of the incompressibility of the droplet and a uniform evaporation at the droplet surface,
$\Vi$ is now equal to $\Vp$, the velocity of the mass center of the droplet. 
In the following, we consider dimensionless variables, so that the surface area $\Areadrop\in[0,1]$. 
The GPBE equation satisfied by the new NDF, $\NNDF(t,\xv;\Areadrop,\Vp)$ can be derived from \eqref{eq:GBPE-NDF}:
\begin{equation}
\partial_t\NNDF+\nabla\cdot\{\Vp\NNDF\}+\nabla\cdot\{\dot{\Vp}\NNDF\}+\partial_{\Areadrop}\{\dot{\Areadrop}\NNDF\}=0.
\label{eq:WB-sphere}
\end{equation}

The droplet acceleration $\dot{\Vp}$ is equal to the sum of forces per unit of mass acting on the droplet. 
In the following, we consider that the only force acting on the droplets is the drag due to the carrier gas.
This can be modeled in a first time by a linear Stokes law:
\begin{equation}
\dot{\Vp}=\dfrac{\Uf_g-\Vp}{\Stokes(\Areadrop)},
\end{equation}
where $\Stokes(\Areadrop)=\theta\Areadrop$ is the Stokes number, which depends linearly on the surface area $\Areadrop$. 
The coefficient $\theta$ depends on the physical properties of the gas and the liquid. 

Since we consider spherical and non-deforming droplets, the evolution of the surface area is only due to evaporation. 
The Lagrangian-time derivative of the surface area is equal to the evaporation:
\begin{equation}
  \dot{\Areadrop}=\EvapRate(\Areadrop).
\label{eq:evap-rate}
\end{equation}
Assuming a $d^2$ law for the phenomenon, the evaporation rate is constant: $\EvapRate(S)=-\KEvap$. 

Finally, the NDF satisfies the following simplified and dimensionless form of the Williams-Boltzmann Equation (WBE) \cite{williams1958}:
\begin{equation}
  \partial_t\NNDF+\nabla_{\xv}\cdot\{\Vp\NNDF\}
  =
  -\nabla_{\Vp}\cdot\left\{\frac{\Uf_g-\Vp}{\theta\Areadrop}\NNDF\right\}+\KEvap\partial_{\Areadrop}\{\NNDF\}.
\end{equation}
Here, we would like to emphasize the fact that this equation is well known in statistical modeling of the spray and
has been derived here from a statistical gas-liquid interface approach.

Here again, the high-dimensional nature of the phase-space of the WBE makes its discretization 
unreachable for complex industrial applications within a reasonable CPU time. 
Since the exact NDF is not required and only macroscopic quantities of the flows are needed, 
an Eulerian moment method can be used to reduce its complexity. 
The derivation of a system of equations for the evolution of a set of some NDF moments often involves other moments 
of lower or higher order or even other terms which depend on the NDF. 
Therefore, we need to propose a closure for the size and velocity distributions. 
This closure step depends on the considered moments used to describe the spray. 
The set of moments and their order are chosen depending on the application 
and the desired accuracy to capture some specific spray feature: for example, 
the droplet segregation or the polydispersion effect on the evaporation \cite{sabat2016,vie2012ftc}.

First, we close the NDF in the velocity direction. The main issue in modeling the velocity distribution in an Eulerian framework 
is the particle trajectory crossing (PTC), occurring especially for inertial droplets \cite{vie2015,sabat2016}. 
For accurate modeling of the PTC, one could use high order velocity-moment closure, such as the anisotropic Gaussian velocity distribution 
\cite{vie2015,sabat2016}. In the present work, we consider a simplified model based on a monokinetic assumption \cite{deChaisemartin2009}:
\begin{equation}
  \NNDF(t,\xv;\Areadrop,\Vp)=\ndf(t,\xv,\Areadrop) \, \delta(\Vp-\Up(t,\xv)),
\end{equation}
so that at a given point in space and time, all the particles have the same velocity. 
This assumption is valid for low inertia droplets ($\Stokes(\Areadrop)<1$), 
when the droplet velocities are rapidly relaxed to the local gas velocity and droplets do not experience too much crossing. 
Then, from equation \eqref{eq:WB-sphere} we derive the following semi-kinetic equation:
\begin{equation}
  \begin{array}{rcl}
    \partial_t\ndf+\partial_{\xv}\cdot\left(\ndf\Up\right)&=&\KEvap\,\partial_{\size}\ndf,\\
    \partial_t\ndf\U+\partial_{\xv}\cdot\left(\ndf\Up\otimes\Up\right)&=&
    \KEvap\,\partial_{\size}\left(\ndf\Up\right)+\ndf\dfrac{\Ug-\Up}{\Stokes(\size)}.
  \end{array}
  \label{eq:semi-kinetic}
\end{equation}
This system describes the time evolution of the size distribution $\ndf(t,\xv,\Areadrop)$ due 
to the transport in the physical space and evaporation and drag source terms. 

\subsection{Fractional size moments and relation to the interface average geometry}
\label{sect:HOGSM}

In the case of spherical droplets, the averaged interfacial quantities 
(volume fraction, surface area density and Gauss and mean curvatures), 
that have been expressed as moments of the NDF $\NDF(t,\xv;\Mcurv,\Areadrop,\vol,\Vi)$ 
in equations \eqref{eq:Geom-NDF} and \eqref{eq:Moments-NDF}, 
can now be expressed as moments of the simplified NDF $\ndf(t,\xv,\Areadrop)$ as follows:
\begin{equation}
  \begin{array}{rcl}
    \Sigmad\widetilde{\Gcurv}_d &=&4\pi \mom_0,\\
    \Sigmad\widetilde{\Hcurv}_d &=&2\sqrt{\pi} \mom_{1/2},\\
    \Sigmad &=& \mom_1,\\
    \Alphad &=& \ds{\frac{1}{6\sqrt{\pi}}\mom_{3/2}},
  \end{array}
\end{equation}
where a fractional size-moment is given by:
\begin{equation}
\label{eq:frac-mom-def}
\mom_{k/2}(t,\xv)=\int_0^1\Areadrop^{k/2}\ndf(t,\xv,\Areadrop)d\Areadrop.
\end{equation}
Now, from the semi-kinetic equation \eqref{eq:semi-kinetic} we derive a high order fractional moment model. 
This model gives the evolution of the mean geometrical interfacial variables for spherical and polydisperse droplets, 
involving transport, evaporation and drag force:
\begin{equation}
\left\{
 \begin{array}{lclcl}
  \partial_t\mom_0    &+& \partial_{\xv}\cdot\left(\mom_0\Up\right)&=&-\KEvap\ns(t,\xv,\size=0) ,\\
  \partial_t\mom_{1/2}&+& \partial_{\xv}\cdot\left(\mom_{1/2}\Up\right)&=&-\dfrac{\KEvap}{2}\mom_{-1/2},\\[4pt]
  \partial_t\mom_{1}  &+& \partial_{\xv}\cdot\left(\mom_{1}\Up\right)&=&-\KEvap \mom_{0},\\[4pt]
  \partial_t\mom_{3/2}&+& \partial_{\xv}\cdot\left(\mom_{3/2}\Up\right)&=&-\dfrac{3\KEvap}{2}\mom_{1/2},\\[4pt]
  \partial_t\left(\mom_{1}\U\right)&+& \partial_{\xv}\cdot\left(\mom_{1}\Up\otimes\Up\right)
    &=&-\KEvap \mom_{0}\Up+\mom_0\dfrac{\U_g-\Up}{\theta},
 \end{array}\right.
 \label{eq:EMSMG}
\end{equation}
where the pointwise disappearance flux $-\KEvap\ndf(t,\xv,\size=0)$ at size zero and the negative order moment 
$\mom_{-1/2}=\int_{0}^{1}\size^{-1/2}\ndf(t,\size)d\size$ depend on the size distribution, which is an unknown function in this model.

\subsection{Closure relations for the fractional moments model}

In the following, we use a smooth reconstruction of the size distribution through the maximisation of Shannon entropy 
\cite{mead84} to close the system of equation. This reconstruction has been already used in the EMSM model \cite{kah12,massot2010} 
and CSVM model \cite{vieJCP2012} and shows a high capacity to model the polydispersion effect on the evaporation 
and the drag force response of droplets compared to the Multi-fluid approach \cite{laurent2015}. 
In the case of fractional moments \eqref{eq:frac-mom-def}, the reconstruction by entropy maximization reads:
\begin{equation}
  \max\left\{ \SEntropy[n]\!=\!-\!\int_0^1\!\ns(\size)\ln(\ns(\size))d\size\right\},\quad \text{s. t.} \quad
  \mom_{k/2}\!=\!\int_0^1\!\size^{k/2}\ns(\size)d\size,
  \, k\!=\!0\hdots N,
  \label{eq:ME_optimization_problem}
\end{equation}
The existence and uniqueness of such a reconstruction have been proved in \cite{Essadki-SIAM} and
the reconstructed NDF has the following form:
\begin{equation}
  \nME(\size)=\exp\left(-(\lambda_0+\lambda_1\size^{1/2}+\lambda_2\size+\lambda_3\size^{3/2})\right).
\end{equation}
To determine the $\lambda_k$ coefficients, we use a similar algorithm as the one used for integer moment \cite{mead84}.\\

Fractional size moments can be seen as a first and simple model that we can derive successfully from the 
statical gas-liquid interface approach introduced in \ref{sect:discrte-SDF}. In this model, we suppose a spherical 
form of droplets which allows to simplify the system of equations \eqref{eq:geometrical_moments_equations}, and thus obtain
the closed system \eqref{eq:EMSMG} using the reconstruction of the size distribution through entropy maximization. 
In this case, droplet deformation is neglected and the evaporation is the only source term acting on 
the size variation. So on, the actual model can be used for polydisperse evaporating sprays, but with 
more capacity to be coupled, in future work, with a model of the form \eqref{eq:geometrical_moments_equations} to describe the gas-liquid interface.

The numerical resolution of system \eqref{eq:EMSMG} may be achieved through an operator splitting algorithm \cite{doisneau14,descombes14}, 
that integrates independently the transport part (left-hand side of system \eqref{eq:EMSMG}) and the source terms (right-hand side of 
system \eqref{eq:EMSMG}). The transport part is solved using a kinetic scheme \cite{deChaisemartin2009} 
adapted for fractional moments. A new solver for the source terms (evaporation and drag force) is proposed in \cite{Essadki-SIAM}. 
The proposed numerical schemes ensure the realizability of the moments and show a high accuracy for evaporation effects at 0D 
and 1D transport \cite{Essadki-SIAM}. 
The new model was also compared to the EMSM model for a 2D evaporating case \cite{Essadki-SIAM,OGST}. 
Finally in \cite{OGST}, a 3D adaptive mesh refinement simulation compares droplet segregation 
with a Lagrangian simulation for non evaporating droplets in the presence of a 
frozen  homogeneous isotropic turbulent gas.

\section{From separated to disperse phases: a generalized statistical approach through spatial averaging.}

In section \ref{sect:discrte-SDF}, we have drawn a link between the NDF and the discrete SDF in the context of a discrete formalism. 
However, the discrete SDF is valid only for the disperse phase and its 
definition supposes that we are able to isolate the droplets/ligaments \eqref{eq:DSDF} in a certain manner. 
Therefore, it is not 
yet obvious that we can link a statistical description of the gas-liquid interface with a  statistical description of the spray evolving under
Williams-Boltzmann's equation. In this part, we introduce a  
spatially averaged SDF, which is defined independently on the flow regime (disperse or separated phases). 
We show that the defined SDF degenerates to the discrete SDF \eqref{eq:DSDF}, when the liquid phase
is disperse and dilute. The objective of this task is twofold: first, to draw a clear link between a 
distribution on the surface, which describes general gas-liquid interfaces, and a distribution of a droplet number 
density. The second purpose is to use this definition to design a new tool and new
algorithms to analyze the gas-liquid interface of DNS computations. 

\subsection{Averaged interfacial quantities and appropriate phase space variables}
\label{sect:averg-sdf-definition}

The spatial averaging process is applied on the phase variables $\phasevar(t,\xv)$ and the instantaneous 
area concentration $\delta_I(t,\xv)$ separately for each realization. In this context, we define 
generally the Averaged SDF (ASDF) as follows:
\begin{equation}
  \ASDF(t,\xv;\phasevar)=<\aaverg{\delta_I}(t,\xv) \, \delta(\phasevar-\paverg{\phasevar}(t,\xv))>,
\label{eq:ASDF}
\end{equation}
where the averaged space variables $\paverg{\phasevar}(t,\xv)$ and the average area measure $\aaverg{\delta_I}(t,\xv)$ will be defined such that we satisfy the following properties:

\begin{enumerate}
  \item \label{item:requirement1} 
    The spatial average of the surface area does not lead to the spreading of the interface, such 
    that the 
    interface thickness remains zero for one realization, \textit{i.e.} $\aaverg{\delta_I}(t,\xv)=0$ for $\xv\notin\totalSurf(t)$.
  \item \label{item:requirement2} 
    The new distribution preserves the space integral of the  first  order moments of the classical SDF
    as in Proposition \ref{prop:TopologicalMoments}. This property reads:
    \begin{equation}
      \int_{\xv\in\Domain}\int_{\phasevar\in\phasespace}
        \phasevar^{\lbold}\ASDF(t,\xv;\phasevar)\dint{5}{\phasevar} 
      \dint{3}{\xv}
      =
      \int_{\xv\in\Domain}\int_{\phasevar\in\phasespace}
        \phasevar^{\lbold}\SDF(t,\xv;\phasevar)
      \dint{5}{\phasevar} \dint{3}{\xv}
    \label{eq:requirement2}
    \end{equation}
    where $\lbold\in \{0,1\}^5, \;l_1+\dots+l_5\leq 1$ and $\phasevar=\{\phasevar_1,\hdots,\phasevar_5 \}$.
    We recall that this property ensures the possibility to express the expected-mean interfacial quantities 
    $(\Sigma(t,\xv),\widetilde{\Mcurv}(t,\xv),\widetilde{\Gcurv}(t,\xv))$ as moments of the new averaged SDF.
  \item \label{item:requirement3}
    The new SDF can be related to the discrete SDF, when the domain contains a dilute disperse phase, 
    such that the larger particle diameter is smaller than the inter-particle distance. 
    We express this relation for a volume space $\subdomain\subset\Domain$, 
    where the border of the domain does not cross any particle, as follows:
    \begin{equation}
      \int_{\xv\in\subdomain}\ASDF(t,\xv;\phasevar)\dint{3}{\xv}=\int_{\xv\in\subdomain}\SDF^{d}(t,\xv;\phasevar)\dint{3}{\xv}.
    \end{equation}
    In these conditions, we can relate the spatial-averaged SDF to the NDF of the particles 
    in the same way as in section \ref{sec:link-discrete-ndf}.
\end{enumerate}


To fulfill the first requirement \ref{item:requirement1}, we propose to define the averaged quantities as follows:

\begin{equation}
\aaverg{\delta_I}(\xv)=\normalizeASDF(t,\xv)\delta_I(t,\xv),
\label{eq:aaverg}
\end{equation} 
\begin{equation}
  \paverg{\phasevar}(t,\xv) =
  \begin{cases}
    \dfrac{1}{\normalizeASDF(t,\xv)} \int_{\yv \in\Domain} w(\xv;\yv-\xv) \, \delta_I(t,\yv)\phasevar(t,\yv)\dint{3}{\yv}, & \text{if $\normalizeASDF(t,\xv)>0$} \\
    \\
    \text{Not defined} & \text{otherwise}
  \end{cases}
  \label{eq:paverg}
\end{equation}
and 
\begin{equation}
\normalizeASDF(t,\xv)=\int_{\yv \in\Domain}w(\xv;\yv-\xv) \, \delta_I(t,\yv)\dint{3}{\yv},
\label{eq:normalizeASDF}
\end{equation}
where $w(\xv;\rv=\yv-\xv)$ is a convolution kernel.
The function $w(\xv;\rv)$ is chosen such that it vanishes for large $||\rv||>h$, 
where $h>0$ is a characteristic spatial length scale, independent from the position $\xv$. The case where $\normalizeASDF(t,\xv)=0$, the averaged interfacial variables are not defined and can take any real value, since no interface exists in the vicinity of the position $\xv$.

The kernel function is designed such that the SDF \eqref{eq:ASDF}, satisfies the two remaining requirements (\ref{item:requirement2} and \ref{item:requirement3}).


For the second requirement \ref{item:requirement2}, we can show that equation \eqref{eq:requirement2} is satisfied if we have, for any $\yv\in\Domain$:
\begin{equation}
\int_{\xv\in\Domain}w(\xv,\yv-\xv)\delta_I(t,\xv)\dint{3}{\xv}=1
\label{eq:statis-requirement2}
\end{equation}
Finally, for the third requirement property \ref{item:requirement3}, we look for a spatial-average 
of the interfacial variables \eqref{eq:paverg} that can degenerate to an interfacial average over 
one particle surface \eqref{eq:average-over-p-surf-implicit}, 
when the inter-particle distance is much larger than the particles diameters. 
We can reach this goal by first choosing a length scale $h$ larger than the maximum of particle diameters and 
smaller than the inter-particle distance. Then we propose to use the following kernel function, 
which satisfies the condition \eqref{eq:statis-requirement2}:
\begin{equation}
w(\xv;\rv)
=\dfrac{\One_{||\rv'||<h}(\rv)}{\int_{\yv'\in\Volh{\xv+\rv}} \delta_I(\yv')\dint{3}{\yv}}
 \label{eq:ASDF-requirement2_3}
\end{equation}
where $\Volh{\xv}=\{\yv \in \Domain;\,\,||\xv-\yv||<h\}$. Therefore, $w(\xv,\rv)$ represents the inverse of the surface area included in $\Volh{\xv+\rv}$. In the case where no interface exists in the volume $\Volh{\xv+\rv}$, $w(\xv,\rv)$ is not defined and can take any real value without affecting the definition of the averaged SDF.

\subsection{Generalized Number density function}
Using the averaged SDF defined in the previous section, we can define the following distribution function:
\begin{definition}
	We define the Generalized Number Density Function (GNDF) as follows:
	\begin{equation}
		\GNDF(t,\xv;\phasevar)=\dfrac{|\Gcurv|}{4\pi}\ASDF(t,\xv;\phasevar).
	\label{eq:def-GNDF-abs}
	\end{equation}
\end{definition}
Let us now consider a disperse phase made of non-spherical particles (droplets or bubbles). By considering a sufficiently 
large length scale averaging $h$ (for example larger than the particle diameter), 
we can ensure the positivity of the averaged Gauss curvature. In this case, we can simply write the GNDF 
as follows:
\begin{equation}
	\GNDF(t,\xv;\phasevar)=\dfrac{\Gcurv}{4\pi}\ASDF(t,\xv;\phasevar).
\label{eq:def-GNDF}
\end{equation}
In the following, we consider a dilute disperse phase, such that the inter-distance between the droplets are larger than the droplet diameter. 
Using a length scale $h$ larger than the droplet diameter, we show the link between the GNDF and the classic NDF, 
as it was stated in the third requirement \ref{item:requirement3} of the last section. 
To conduct this derivation, we use the notations of section \ref{sect:discrte-SDF}, 
used to introduce the discrete formalism of the particles. 
\begin{proposition}\label{proposition:general_link}
	The two integrals over a space volume $\subdomain$, 
	whose borders do not cross any particle, of the GNDF $\GNDF(t,\xv;\phasevar)=	\dfrac{\Gcurv}{4\pi}\ASDF(t,\xv;\phasevar)$ 
	and the NDF $\NDF(t,\xv;\phasevar)$ are equal:
	\begin{equation}
		\int\limits_{\xv\in\subdomain}\GNDF(t,\xv;\phasevar)\dint{3}{\xv}= \int_{\xv\in\subdomain}\NDF(t,\xv;\phasevar)\dint{3}{\xv}.
	\end{equation}
\end{proposition}

\begin{proof}
	\begin{equation}
		\begin{array}{rcl}
			\int\limits_{\xv\in\subdomain} \Gcurv \ASDF(t,\xv;\phasevar) \dint{3}{\xv}
				&=& <\sum\limits_{k=1}^{\Nmax} \left\{\int\limits_{\xv\in\omega_k}\Gcurv \, \aaverg{\delta_I}(t,\xv) \, \delta(\phasevar-\Aver{\phasevar}(t,\xv)\dint{3}{\xv}\right\}
				+ \int\limits_{\xv\in\subdomain-\bigcup\limits_{k=1}^{N_{max}}\omega_k} \Gcurv \, \aaverg{\delta_I}(t,\xv) \, \delta(\phasevar-\Aver{\phasevar}(t,\xv)\dint{3}{\xv}>,\\[11pt]
 				&=&<\sum\limits_{k=1}^{\Nmax} \left\{ \int\limits_{\xv\in\omega_k}\Aver{\Gcurv}_k \, \aaverg{\delta_I}(t,\xv) \, \delta(\phasevar-\Aver{\phasevar}_k)\dint{3}{\xv} \right\}>
				+ 0,\\[11pt]
 				&=&<\sum\limits_{k\in \{ j,\xv_j\in\subdomain\}} \left\{\Aver{\Gcurv}_k \, S_k \, \delta(\phasevar-\Aver{\phasevar}_k) \right\}>,\\[11pt]
 				&=&4\pi<\sum\limits_{k\in\{j,\xv_j\in\subdomain\}} \left\{\delta(\phasevar-\Aver{\phasevar}_k) \right\}>,\\[11pt]
 				&=&4\pi \int_{\xv\in\subdomain}\NDF(t,\xv;\phasevar)\dint{3}{\xv},
		\end{array}
	\end{equation} 
	where $\omega_k=\{\xv\in\subdomain,\,\,\min\limits_{\yv\in\Sigma_k}{||\xv-\yv||\leq h}\}$ is a space volume that 
	includes the droplet $\particle_k$. To pass from the third to the fourth equality, we used Gauss-Bonnet formula and we suppose that all droplets are homeomorph to a sphere.
\end{proof}

 This result shows a general interpretation for the NDF, which is not restricted to a discrete formalism.
We can derive in the same way as in the section \ref{sect:discrte-SDF} the GPBE satisfied by the GNDF:
\begin{equation}
	\partial_t \GNDF+\nabla_{\xv}\cdot\{\Vi\GNDF\}+\nabla_{\phasevar}\cdot\{<\dot{\phasevar}>_c\GNDF\}=<\Breakup>_c\GNDF.
\label{eq:GBPE-GNDF-ST}
\end{equation}
where $\Breakup=\dot{\Stretch}+\frac{\dot{\Gcurv} }{\Gcurv}$ is a source term of topology variation  and which can be related mainly to the breakup and coalescence.\\

Now, it is more obvious that the SDF and NDF are strongly related. The SDF can be seen as  the origin of the NDF as it was discussed in \cite{romain2017}. The averaged SDF allows to consider a generic statistical description of the gas-liquid interface and which can also be related directly to the NDF in the case of dilute disperse phase as it is stated in the proposition \ref{proposition:general_link}. Indeed, by using the spatial averaging procedure, we are able to isolate droplets in dilute disperse phase region, and thus we can compute the NDF from the averaged SDF.

\section{Algorithms and techniques for the numerical computation of the curvatures and the SDF}
\label{sect:result-algorithm}

In this part, we use the results of the previous section to design a new numerical tool dedicated to post-processing 
DNS two-phase flows simulations. Such tools allow to enhance the analysis of the gas-liquid interface evolution.
The DNS computations are performed using the ARCHER code 
\cite{menard07}, where a combined VOF and level-set approach is used to capture the interface and a ghost method is 
applied to represent accurately the jump of variables across the liquid-gas interface. 

At first, some numerical tools to compute the curvatures and the surface area of the gas-liquid interface present in each cell
are already available in ARCHER code \cite{romain2017}. 
Unfortunately, the surface area and the curvatures being computed separately, the Gauss-Bonnet formula can not be numerically satisfied. 
For this reason, we choose to compute the curvatures and the 
surface element areas by using the algorithm presented thereafter.
Such algorithm will allow to obtain the number density function 
of droplets in DNS simulations by using simple computations on the surface and without the need of an algorithm 
that isolates the droplets \cite{kang}.

The new algorithm is implemented independently from the ARCHER code, 
from which only the distance function (level-set) data is used. 
The different algorithm steps, to compute the local curvatures and surface elements at the interface, are summarized as follows:
\begin{itemize}
	\item The gas-liquid interface is discretized with a 2D triangulated mesh using the Marching Cubes algorithm 
    \cite{marching_cubes_jgt}. 
	  This 	algorithm takes the $3D$ level-set scalar field as an input 
    and returns a 2D meshed surface. This mesh is described with two arrays: 
    the array of vertices, $\Vertsarray$ of dimension $n_v\times3$, and the array of faces (defined by three vertices),
  	$\Facearray$ of dimension $n_f\times 3$ which defines the connectivity between vertices. 
  	In this work, we use the Python package sckimage \cite{sckimage} to triangulate the surface.
	\item For each vertex $\vertex\in[1,n_v]$, $\Neighbour(\vertex)$ denotes the set of neighbors of $\vertex$: 
	  $\vertex'$ is a neighbor of $\vertex$ if both vertices share a same face $\face\in[1,n_f]$. 
    By abuse of notation, $\Neighbour(\vertex)$ also denotes the set of faces $\vertex$ belongs to and for 
    each $\face\in\Neighbour(\vertex)$, $\normal(\face)$ denotes the normal vector to the face, oriented toward 
    the gas phase.  
  \item In the neighbourhood of $\vertex$, we define its dual cell $\dual(\vertex)$ 
    as a mix between the Voronoi cell (figure \ref{fig:voronoi_cell}) and the barycentric cell 
    (figure \ref{fig:barycentric_cell}) around $\vertex$, as proposed in \cite{discret_curvature}. 
    Simply, for each triangular face $\face$ containing $\vertex$, 
    the Voronoi contour is considered if all the angle are acute. The barycentric contour is considered otherwise. 
    Also, $\theta_{\face}$ denotes the angle at vertex $\vertex$ in $\face$, see \ref{fig:angle}, and 
    $A[\face,\vertex]=\area(\face\cap\dual(\vertex))$ is the area of the dual cell within the face $\face$. 
    Then, the discrete local area $A[\vertex]$, normal $\normal[\vertex]$ and Gauss curvature $\Gcurv[\vertex]$ 
    can be computed:
    \begin{equation}\label{eq:DiscreteAreaNormalAndGauss}
      \begin{array}{c}
			\ds
      A[\vertex]=\sum_{\face\in\Neighbour(\vertex)}A[\face,\vertex], 
      \quad\quad
      \normal[\vertex] = \frac{\sum\limits_{\face\in\Neighbour(\vertex)} A[\face,\vertex]\normal(\face)}
                             {\left|\left|\sum\limits_{\face\in\Neighbour(\vertex)} A[\face,\vertex]\normal(\face)\right|\right|},
      \\ \\
      \ds \Gcurv[\vertex]=\left(2\pi-\sum_{\face\in\Neighbour(\vertex)} \theta_{\face}\right)/A[\vertex].
      \end{array}
    \end{equation}
  \item Finally, for each neighbour $\vertex'$ of $\vertex$, $\alpha_{\vertex,\vertex'}$ and $\beta_{\vertex,\vertex'}$
    denote the two angles facing the edge $[\vertex,\vertex']$, without distinction, 
    as illustrated in figure \ref{fig:angle}. Now, the local discrete mean curvature reads:
    \begin{equation}\label{eq:DiscreteMean}
      \Mcurv[\vertex] = \frac{1}{2A[\vertex]} \sum_{\vertex'\in\Neighbour(\vertex)}
        \left(\cot(\alpha_{\vertex',\vertex})+\cot(\beta_{\vertex',\vertex})\right)
        \left(\Vertsarray[\vertex]-\Vertsarray[\vertex'])\right)\cdot\normal[\vertex].
    \end{equation}
\end{itemize}


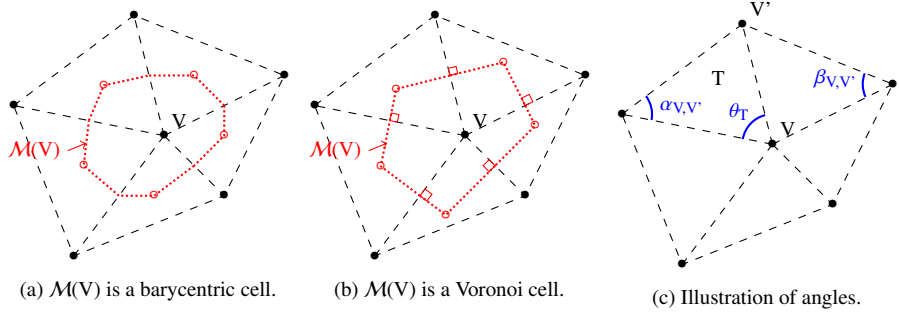
\begin{figure}
\begin{minipage}{.33\linewidth}
  \centering
  \subfloat[$\mathcal{M} (\vertex)$ is a barycentric cell. \label{fig:barycentric_cell}]{\footnotesize
\begin{tikzpicture}[scale=0.4]

\coordinate (o) at (5,4);

\coordinate (a) at (0,5);
\coordinate (b) at (4,8);
\coordinate (c) at (9,6);
\coordinate (d) at (7,2);
\coordinate (e) at (2,0);

\coordinate (oa) at (2.5,4.5);
\coordinate (ob) at (4.5,6);
\coordinate (oc) at (7,5);
\coordinate (od) at (6,3);
\coordinate (oe) at (3.5,2);

\coordinate (oab) at (3,5.67);
\coordinate (obc) at (6,6);
\coordinate (ocd) at (7,4);
\coordinate (ode) at (4.67,2);
\coordinate (oea) at (2.33,3);

\coordinate (oo0) at (2.41, 3.75);
\coordinate (oo1) at (1.8, 3.5);

\node at (o) {$\bullet$};
\node[above right] at (o) {V};

\node at (a) {$\bullet$};
\node at (b) {$\bullet$};
\node at (c) {$\bullet$};
\node at (d) {$\bullet$};
\node at (e) {$\bullet$};

\draw[dashed] (a) -- (b) -- (c) -- (d) -- (e) -- cycle;
\draw[dashed] (o) -- (a);
\draw[dashed] (o) -- (b);
\draw[dashed] (o) -- (c);
\draw[dashed] (o) -- (d);
\draw[dashed] (o) -- (e);

\node[red] at (oab) {$\circ$};
\node[red] at (obc) {$\circ$};
\node[red] at (ocd) {$\circ$};
\node[red] at (ode) {$\circ$};
\node[red] at (oea) {$\circ$};

\draw[red, densely dotted, thick] (oa) -- (oab) -- (ob);
\draw[red, densely dotted, thick] (ob) -- (obc) -- (oc);
\draw[red, densely dotted, thick] (oc) -- (ocd) -- (od);
\draw[red, densely dotted, thick] (od) -- (ode) -- (oe);
\draw[red, densely dotted, thick] (oe) -- (oea) -- (oa);

\draw[red, ->] (oo1) -- (oo0);
\node[red, left] at (oo1) {$\mathcal{M}$(V)};

\end{tikzpicture}}
\end{minipage}%
\begin{minipage}{.33\linewidth}
  \centering
  \subfloat[$\mathcal{M} (\vertex)$ is a Voronoi cell.  \label{fig:voronoi_cell}]{\footnotesize
\begin{tikzpicture}[scale=0.4]
\coordinate (o) at (5,4);

\coordinate (a) at (0,5);
\coordinate (b) at (4,8);
\coordinate (c) at (9,6);
\coordinate (d) at (7,2);
\coordinate (e) at (2,0);

\coordinate (oa) at (2.5,4.5);
\coordinate (ob) at (4.5,6);
\coordinate (oc) at (7,5);
\coordinate (od) at (6,3);
\coordinate (oe) at (3.5,2);

\coordinate (oa1) at (2.8,4.7);
\coordinate (ob1) at (4.7,6.3);
\coordinate (oc1) at (7.13,5.37);
\coordinate (od1) at (5.6,3);
\coordinate (oe1) at (3.9,2.05);

\coordinate (oab) at (2.71,5.55);
\coordinate (obc) at (6.28,6.44);
\coordinate (ocd) at (7.33,4.33);
\coordinate (ode) at (4.36,1.36);
\coordinate (oea) at (2.2,2.98);

\coordinate (oo0) at (2.41, 3.75);
\coordinate (oo1) at (1.8, 3.5);

\node at (o) {$\bullet$};
\node[above right] at (o) {V};

\node at (a) {$\bullet$};
\node at (b) {$\bullet$};
\node at (c) {$\bullet$};
\node at (d) {$\bullet$};
\node at (e) {$\bullet$};

\draw[dashed] (a) -- (b) -- (c) -- (d) -- (e) -- cycle;
\draw[dashed] (o) -- (a);
\draw[dashed] (o) -- (b);
\draw[dashed] (o) -- (c);
\draw[dashed] (o) -- (d);
\draw[dashed] (o) -- (e);

\node[red] at (oab) {$\circ$};
\node[red] at (obc) {$\circ$};
\node[red] at (ocd) {$\circ$};
\node[red] at (ode) {$\circ$};
\node[red] at (oea) {$\circ$};

\draw[red, densely dotted, thick] (oa) -- (oab) -- (ob);
\draw[red, densely dotted, thick] (ob) -- (obc) -- (oc);
\draw[red, densely dotted, thick] (oc) -- (ocd) -- (od);
\draw[red, densely dotted, thick] (od) -- (ode) -- (oe);
\draw[red, densely dotted, thick] (oe) -- (oea) -- (oa);

\draw[red, ->] (oo1) -- (oo0);
\node[red, left] at (oo1) {$\mathcal{M}$(V)};

\draw[red, rotate=-11.4] (oa) rectangle (oa1);
\draw[red, rotate=14] (ob) rectangle (ob1);
\draw[red, rotate=26.6] (oc) rectangle (oc1);
\draw[red, rotate=44.9] (od) rectangle (od1);
\draw[red, rotate=53.1] (oe) rectangle (oe1);

\end{tikzpicture}} 
\end{minipage}
\begin{minipage}{.33\linewidth}
  \centering
  \subfloat[Illustration  of angles.  \label{fig:angle}]{\footnotesize
\begin{tikzpicture}[scale=0.4]
\coordinate (o) at (5,4);

\coordinate (a) at (0,5);
\coordinate (b) at (4,8);
\coordinate (c) at (9,6);
\coordinate (d) at (7,2);
\coordinate (e) at (2,0);

\coordinate (oa) at (2.5,4.5);
\coordinate (ob) at (4.5,6);
\coordinate (oc) at (7,5);
\coordinate (od) at (6,3);
\coordinate (oe) at (3.5,2);

\coordinate (oa1) at (2.8,4.7);
\coordinate (ob1) at (4.7,6.3);
\coordinate (oc1) at (7.13,5.37);
\coordinate (od1) at (5.6,3);
\coordinate (oe1) at (3.9,2.05);

\coordinate (oab) at (2.71,5.55);
\coordinate (obc) at (6.28,6.44);
\coordinate (ocd) at (7.33,4.33);
\coordinate (ode) at (4.36,1.36);
\coordinate (oea) at (2.2,2.98);

\coordinate (oo0) at (2.41, 3.75);
\coordinate (oo1) at (1.8, 3.5);

\node at (o) {$\bullet$};
\node[above right] at (o) {V};

\node at (a) {$\bullet$};
\node at (b) {$\bullet$};
\node[above right] at (b) {V'};
\node at (c) {$\bullet$};
\node at (d) {$\bullet$};
\node at (e) {$\bullet$};
\node at (3.2,6.2) {T};

\draw[dashed] (a) -- (b) -- (c) -- (d) -- (e) -- cycle;
\draw[dashed] (o) -- (a);
\draw[dashed] (o) -- (b);
\draw[dashed] (o) -- (c);
\draw[dashed] (o) -- (d);
\draw[dashed] (o) -- (e);

\draw[blue, thick] (a) ++(-11:1) arc (-11:37:1)  (2.0,5.2) node {$\alpha_{\text{V,V'}}$};
\draw[blue, thick] (c) ++(159:1) arc (159:205:1) node[above left] {$\beta_{\text{V,V'}}$};
\draw[blue, thick] (o) ++(104:1) arc (104:170:1) (3.9,5.1) node {$\theta_{\text{T}}$};

\end{tikzpicture}}
\end{minipage}
 \caption{Neighboring vertices and surface elements $\mathcal{M} (\vertex)$ around the vertex $\vertex$.}
\label{fig:neighour_vertex}
\end{figure}
After computing the local geometrical quantities, we can evaluate the fine-grain localized SDF for one simulation, \textit{i. e.} one realization. 
In the following, the phase space variables 
are restricted to the two curvatures $\phasevar=(\Mcurv,\Gcurv) \in \phasespace$: we do not consider the 
interface velocity. This 2D space is discretized into $n_h\times n_g$ discrete sub-spaces, such that the 
discrete phase space variables are expressed as follows: 
\[
	\phasevar_{i,j}=(\Mcurv_i,\Gcurv_j)=\phasevar_{min}+(i\Delta\Mcurv,j\Delta\Gcurv), \quad (i,j)\in[1,n_h]\times[1,n_g],
\] 
where $\Delta\Mcurv=(\Mcurv_{max}-\Mcurv_{min})/n_h$, $\Delta\Gcurv=(\Gcurv_{max}-\Gcurv_{min})/n_g$ and the
subscripts $min$ and $max$ refer respectively to the minimum and maximum values obtained by 
\eqref{eq:DiscreteAreaNormalAndGauss} and \eqref{eq:DiscreteMean}. 
The integral of the fine-grain localized SDF over a subdomain $\subdomain\subset\Domain$ can be approximated by: 
\begin{equation}
	\int_{\xv\in\subdomain}\sdf(t,\xv;\phasevar_{i,j})d^3\xv\simeq\dfrac{1}{\Delta\Mcurv \Delta\Gcurv}\sum\limits_{\vertex=1}^{n_v}A[\vertex]
		\, I_{C_{i,j}}(\phasevar[\vertex]) \, I_{\subdomain}(\Vertsarray[\vertex])
\end{equation}
where $\sdf(t,\xv;\phasevar_{i,j})$ is the fine-grain SDF defined in \eqref{eq:SDF-implicit}, $I_{C_{i,j}}(.)$ is the characteristic function of the set $C_{i,j}$:
\[
	C_{i,j}=\{\phasevar', |\Mcurv'-\Mcurv_{i,j}|\leq \Delta\Mcurv/2\,\,\text{and}\,\,|\Gcurv'-\Gcurv_{i,j}|\leq \Delta\Gcurv/2\},
\] 
and $I_{\subdomain}(.)$ is 
the characteristic function of the sub-domain $\subdomain$. In the following, we take $\subdomain=\Domain$. Thus 
the last equation becomes:
\begin{equation}
	\int_{\xv\in\Domain}\sdf(t,\xv;\phasevar_{i,j})d^3\xv
		\simeq
		\dfrac{1}{\Delta\Mcurv \Delta\Gcurv}\sum\limits_{\vertex=1}^{n_v}A[\vertex] \, I_{C_{i,j}}(\phasevar[\vertex]).
\end{equation}
Using this numerical approximation, we define the numerical fine-grain SDF integrated in the whole 
domain as follows:
\begin{equation}
	\nsdf(\phasevar_{i,j})=\dfrac{1}{\Delta\Mcurv \Delta\Gcurv}\sum\limits_{\vertex=1}^{n_v}A[\vertex] \, I_{C_{i,j}}(\phasevar[\vertex]).
\end{equation}
To evaluate the averaged fine-grain SDF \eqref{eq:ASDF}, we first calculate the averaged area and curvatures 
at each vertex $\vertex\in[1,n_v]$. 
These averaged quantities depend on the length scale $h$. In the following, 
we take $h=(\averglen/2)\Delta x$, where $\Delta x$ is the size of the level-set computation cells 
and $\averglen$ is an integer value to be set by the user. We 
denote by  $\widetilde{A}^{\averglen}[\vertex]$ and  $\paverg{\phasevar}^{\averglen}[\vertex]$ the numerical
approximation of the averaged area and curvatures, defined in \eqref{eq:aaverg}, \eqref{eq:paverg} 
and \eqref{eq:ASDF-requirement2_3} and which are numerically
evaluated as follows:
\begin{equation}
	\begin{array}{rcl}
		\widetilde{A}^{\averglen}[\vertex]&=&\sum\limits_{\vertex'\in \Neighbour_{\averglen}(\vertex)}\dfrac{A[\vertex']}{S_k[\vertex']}A[\vertex],\\
		\paverg{\phasevar}^{\averglen}[\vertex]&=&\dfrac{1}{\widetilde{A}^{\averglen}[\vertex]}\sum\limits_{\vertex'\in \Neighbour_{\averglen}(\vertex)}\phasevar[\vertex']\dfrac{A[\vertex']}{S_{\averglen}[\vertex']}A[\vertex],\\
	\end{array}
\end{equation}
where $\Neighbour_k(\vertex)$ is the set of vertices $\vertex'$, such that 
$||\Vertsarray[\vertex]-\Vertsarray[\vertex']||_2\leq k/2\Delta x$ for $k\geq0$ 
and 
$$S_k[\vertex']=\sum\limits_{\vertex^{''}\in O_k(\vertex')}A(\vertex^{''}).$$

The numerical fine-grain averaged SDF is given by:
\begin{equation}
	\nasdf(\phasevar_{i,j})=\dfrac{1}{\Delta\Mcurv \Delta\Gcurv}\sum\limits_{\vertex=1}^{n_v}\widetilde{A}^{\averglen}[\vertex]I_{C_{i,j}}(\paverg{\phasevar}^{\averglen}[\vertex]).
\end{equation}
The GNDF can be then evaluated numerically as follows:
\begin{equation}
	\ngndf(\phasevar_{i,j})
		=
		\dfrac{1}{\Delta\Mcurv \Delta\Gcurv}\sum\limits_{\vertex=1}^{n_v}\dfrac{|\paverg{\Gcurv}^{\averglen}[\vertex]|}{4\pi}\widetilde{A}^{\averglen}[\vertex]I_{C_{i,j}}(\paverg{\phasevar}^{\averglen}[\vertex]).
		\label{eq:numerical_GNDF}
\end{equation}

We recall that the spatial averaged of curvatures and surface area measure, introduced in the section \ref{sect:averg-sdf-definition}, has been designed such that the first order moments of the SDF \eqref{eq:requirement2} are preserved. We can verify numerically this property as follows:
\begin{equation}
\begin{array}{lcccl}
\int_{\Domain}\Sigma(t,\xv)d\xv&\simeq&\sum\limits_{\vertex}A[\vertex]&=&
\sum\limits_{\vertex}\widetilde{A}^{\averglen}[\vertex],\\[8pt]
\int_{\Domain}\Sigma\widetilde{\Mcurv}(t,\xv)d\xv&\simeq&\sum\limits_{\vertex}\Mcurv[\vertex]A[\vertex]&=&\sum\limits_{\vertex}\paverg{\Mcurv}^{\averglen}[\vertex]\widetilde{A}^{\averglen}[\vertex],\\[8pt]
\int_{\Domain}\Sigma\widetilde{\Gcurv}(t,\xv)d\xv&\simeq&\sum\limits_{\vertex}\Gcurv[\vertex]A[\vertex]&=&\sum\limits_{\vertex}\paverg{\Gcurv}^{\averglen}[\vertex]\widetilde{A}^{\averglen}[\vertex].\\[8pt]
\end{array}
\end{equation}


\section{Assessment of the theoretical approach and DNS post-processing}

In this part, we use the ARCHER code to perform some two-phase flows direct numerical simulations. 
The algorithms described in the previous section will be used to post-process the 
level-set data in order to compute the  
fine-grain SDF and NDF for one simulation (one realization). Two numerical tests, with and without topological changes,  allow to test 
the algorithm and the capacity to find some geometrical and topological feature of the gas-liquid interface, depending on the average 
scale $h$ defined in the previous section.
\subsection{Droplets homeomorphic to spheres}
\label{sect:result-result}

In the first configuration, we consider five initial spherical water
droplets of radius $r\in[0.33 \text{ mm},1 \text{ mm}]$ injected at a velocity 5 m/s in a periodic domain 
of a size $1 \text{ cm } \times1\text{ cm }\times1\text{ cm }$ of initially still gas. Classical water/air properties are used here. The initial positions of
droplets centers are chosen randomly, with the constraint however that the inter-droplet distance $d_{drop}$
(computed from the droplet surface) is larger than $2r_{max}$, where initially $r_{max}=1 \text{ mm}$. 
Due to the difference between the gas and the droplet velocities, the shape of the droplets is deformed in time without break-up (the maximum Weber number is $We_{max}<1$).
In Figure \ref{fig:five_droplet_1}, 
we display an illustration of the droplets at two successive times. The color indication at the droplet surfaces shows 
an estimate of the mean curvature. This first case is envisioned in the framework of no topological changes.

\begin{figure}
\begin{minipage}{.4\linewidth}
  \centering
  \setlength{\fboxsep}{0pt}
  \subfloat[$t=8.54 \times 10^{-5}$s]{\includegraphics[width=0.95\linewidth]{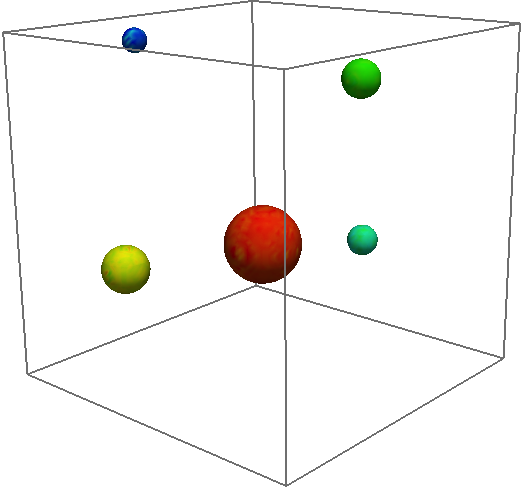}}
\end{minipage}%
\begin{minipage}{.4\linewidth}
  \centering
  \setlength{\fboxsep}{0pt}
  \subfloat[ $t= 7.624\times 10^{-3}$s]{\includegraphics[width=0.95\linewidth]{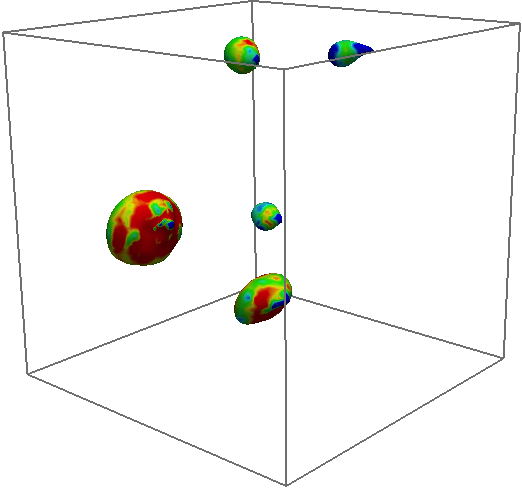}}
\end{minipage}
\begin{minipage}{.16\linewidth}
  \centering
  \subfloat{\includegraphics[width=0.7\linewidth]{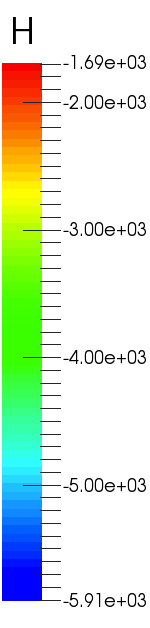}}
\end{minipage}
 \caption{Droplet surface colored according to the mean curvature values at the surfaces for two different times.}
\label{fig:five_droplet_1}
\end{figure}
In Figure \ref{fig:sdfG} (respectively \ref{fig:sdfH}), we plot the SDF for different averaged scales as function 
of Gauss (respectively  mean) curvature. The localized SDF, plotted with the blue line, corresponds to the 
length scale average $h=\Delta x/2$. In this case, the SDF is continuous and we can not identify the surface area 
contribution of each droplet.
But, by using larger scales of average ($h=25\Delta x/2$ and $h=55\Delta x/2$), we see that we can 
identify five peaks which correspond to the five droplets. 
Indeed, when $d_{drop}>h=55\Delta x/2 > r_{max}$, for any point $\xv$ situated at the interface of a droplet, 
the space volume $\Volh{\xv}=\{\yv \in \Domain;\,\,||\xv-\yv||<h\}$ contains that, and only that, droplet.
Thus, for each droplet we compute one average mean and Gauss curvatures. 
Consequently, in this case the averaged SDF becomes a sum of Dirac delta function.

In figure \ref{fig:ndfG} (respectively \ref{fig:ndfH}), we display the GNDF for different average scales. These distributions are obtained from the averaged SDFs using \eqref{eq:numerical_GNDF}. When the average scale  $h$ is larger than the droplet diameters, we obtain five peaks of value $1$. We can then count the right number of droplets for each averaged Gauss or mean curvature 
thanks to Gauss-Bonnet formula with a proper evaluation of the curvature and using our averaging approach. This assesses the proposed theoretical and algorithmic approach  within the framework of no topological change.


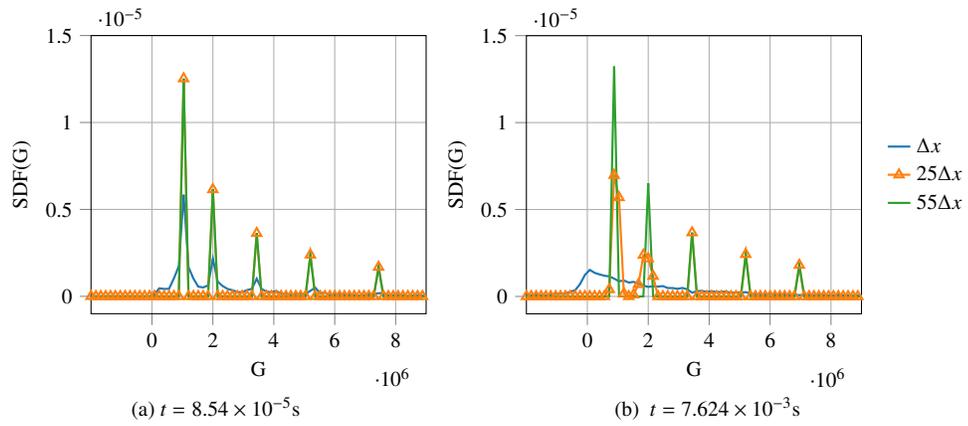
\begin{figure}
\begin{minipage}{.48\linewidth}
  \centering
  \subfloat[$t=8.54 \times 10^{-5}$s]{
\pgfplotsset{
compat=1.3,
legend image code/.code={
\draw[mark repeat=2,mark phase=2]
plot coordinates {
(0cm,0cm)
(0.15cm,0cm)        
(0.3cm,0cm)         
};%
}
}
\footnotesize

\begin{tikzpicture}

\definecolor{color1}{rgb}{1,0.498039215686275,0.0549019607843137}
\definecolor{color0}{rgb}{0.12156862745098,0.466666666666667,0.705882352941177}
\definecolor{color2}{rgb}{0.172549019607843,0.627450980392157,0.172549019607843}

\begin{axis}[
scale=0.65,
xlabel={G},
ylabel={SDF(G)},
xmin=-2000000, xmax=9000000,
ymin=-0.1e-5, ymax=1.5e-05,
xtick={0,2000000,4000000,6000000,8000000},
xticklabels={0,2,4,6,8},
tick align=outside,
tick pos=left,
xmajorgrids,
x grid style={lightgray!92.026143790849673!black},
ymajorgrids,
y grid style={lightgray!92.026143790849673!black},
]
\addplot [thick, color0]
table {%
-2320000 0
-2160000 0
-2000000 0
-1840000 0
-1680000 0
-1520000 0
-1360000 0
-1200000 0
-1040000 0
-880000 0
-720000 0
-560000 0
-400000 0
-240000 0
-80000 0
80000 1.08493068755e-07
240000 4.62309056023e-07
400000 4.24693004024e-07
560000 4.28771433419e-07
720000 1.03213551535e-06
880000 1.72981046566e-06
1040000 5.84464278742e-06
1200000 1.67866945153e-06
1360000 1.02541822963e-06
1520000 5.56204936921e-07
1680000 5.12298059645e-07
1840000 6.23723547989e-07
2000000 2.1581881585e-06
2160000 8.37854849753e-07
2320000 5.93895449792e-07
2480000 4.27979656526e-07
2640000 3.51228714796e-07
2800000 2.45396491565e-07
2960000 2.37127809163e-07
3120000 3.58425109074e-07
3280000 4.29556948218e-07
3440000 1.01510517905e-06
3600000 4.20934598658e-07
3760000 2.71578164578e-07
3920000 2.35134655841e-07
4080000 3.11053529326e-07
4240000 1.43890450166e-07
4400000 1.49337898668e-07
4560000 7.44244550359e-08
4720000 1.33795625087e-07
4880000 1.3020488847e-07
5040000 1.43356145282e-07
5200000 3.14877578455e-07
5360000 5.02768017043e-07
5520000 1.72783006417e-07
5680000 2.103664608e-07
5840000 7.53218929157e-08
6000000 9.91992019498e-08
6160000 1.31987272707e-07
6320000 8.61309457954e-08
6480000 9.32067470104e-08
6640000 7.07792493564e-08
6800000 3.33374197445e-08
6960000 6.11717556851e-08
7120000 1.05739166091e-07
7280000 1.53765195779e-07
7440000 1.75756840388e-07
7600000 1.97271409824e-07
7760000 5.67032464327e-08
7920000 6.72192742114e-08
8080000 8.46580585054e-08
8240000 6.03076098673e-08
8400000 1.56041955885e-08
8560000 3.91550096546e-08
8720000 6.51751504319e-08
8880000 3.7990519505e-08
9040000 3.25245189605e-08
9200000 4.68954509116e-08
9360000 2.79652602162e-08
9520000 5.64344924873e-09
9680000 8.34979908137e-09
9840000 2.60600331804e-08
10000000 2.14387312312e-08
};
\addplot [thick, mark=triangle, color1]
table {%
-2320000 0
-2160000 0
-2000000 0
-1840000 0
-1680000 0
-1520000 0
-1360000 0
-1200000 0
-1040000 0
-880000 0
-720000 0
-560000 0
-400000 0
-240000 0
-80000 0
80000 0
240000 0
400000 0
560000 0
720000 0
880000 0
1040000 1.25237368138e-05
1200000 0
1360000 0
1520000 0
1680000 0
1840000 0
2000000 6.14581562593e-06
2160000 0
2320000 0
2480000 0
2640000 0
2800000 0
2960000 0
3120000 0
3280000 0
3440000 3.62889605373e-06
3600000 0
3760000 0
3920000 0
4080000 0
4240000 0
4400000 0
4560000 0
4720000 0
4880000 0
5040000 0
5200000 2.3853978597e-06
5360000 0
5520000 0
5680000 0
5840000 0
6000000 0
6160000 0
6320000 0
6480000 0
6640000 0
6800000 0
6960000 0
7120000 0
7280000 0
7440000 1.68089444254e-06
7600000 0
7760000 0
7920000 0
8080000 0
8240000 0
8400000 0
8560000 0
8720000 0
8880000 0
9040000 0
9200000 0
9360000 0
9520000 0
9680000 0
9840000 0
10000000 0
};
\addplot [thick, color2]
table {%
-2320000 0
-2160000 0
-2000000 0
-1840000 0
-1680000 0
-1520000 0
-1360000 0
-1200000 0
-1040000 0
-880000 0
-720000 0
-560000 0
-400000 0
-240000 0
-80000 0
80000 0
240000 0
400000 0
560000 0
720000 0
880000 0
1040000 1.25237368138e-05
1200000 0
1360000 0
1520000 0
1680000 0
1840000 0
2000000 6.14581562593e-06
2160000 0
2320000 0
2480000 0
2640000 0
2800000 0
2960000 0
3120000 0
3280000 0
3440000 3.62889605373e-06
3600000 0
3760000 0
3920000 0
4080000 0
4240000 0
4400000 0
4560000 0
4720000 0
4880000 0
5040000 0
5200000 2.3853978597e-06
5360000 0
5520000 0
5680000 0
5840000 0
6000000 0
6160000 0
6320000 0
6480000 0
6640000 0
6800000 0
6960000 0
7120000 0
7280000 0
7440000 1.68089444254e-06
7600000 0
7760000 0
7920000 0
8080000 0
8240000 0
8400000 0
8560000 0
8720000 0
8880000 0
9040000 0
9200000 0
9360000 0
9520000 0
9680000 0
9840000 0
10000000 0
};
%
%
%

\end{axis}

\end{tikzpicture}}

\end{minipage}%
\begin{minipage}{.48\linewidth}
  \subfloat[ $t= 7.624\times 10^{-3}$s]{
\pgfplotsset{
compat=1.3,
legend image code/.code={
\draw[mark repeat=2,mark phase=2]
plot coordinates {
(0cm,0cm)
(0.15cm,0cm)        
(0.3cm,0cm)         
};%
}
}
\footnotesize

\begin{tikzpicture}

\definecolor{color1}{rgb}{1,0.498039215686275,0.0549019607843137}
\definecolor{color0}{rgb}{0.12156862745098,0.466666666666667,0.705882352941177}
\definecolor{color2}{rgb}{0.172549019607843,0.627450980392157,0.172549019607843}

\begin{axis}[
scale=0.65,
xlabel={G},
ylabel={SDF(G)},
xmin=-2000000, xmax=9000000,
ymin=-0.1e-5, ymax=1.5e-05,
xtick={0,2000000,4000000,6000000,8000000},
xticklabels={0,2,4,6,8},
tick align=outside,
tick pos=left,
xmajorgrids,
x grid style={lightgray!92.026143790849673!black},
ymajorgrids,
y grid style={lightgray!92.026143790849673!black},
legend style={at={(1.05,0.5)}, anchor=west, draw=none},
legend entries={{$\Delta x$},{$25 \Delta x$},{$55 \Delta x$}},
legend cell align={left}
]
\addplot [thick, color0]
table {%
-2320000 3.03139761682e-08
-2160000 1.43929614772e-08
-2000000 5.64035876463e-08
-1840000 6.11250301623e-08
-1680000 1.00164566243e-07
-1520000 6.44528438282e-08
-1360000 4.99213864152e-08
-1200000 9.77747963291e-08
-1040000 1.53268723597e-07
-880000 1.58185293291e-07
-720000 1.53973592678e-07
-560000 2.73326217995e-07
-400000 3.56462837971e-07
-240000 7.08653824622e-07
-80000 1.23481968172e-06
80000 1.52891278271e-06
240000 1.35788360578e-06
400000 1.26689523523e-06
560000 1.19116224612e-06
720000 1.1531267007e-06
880000 1.02558445309e-06
1040000 8.70529018551e-07
1200000 9.18853126735e-07
1360000 7.97144435443e-07
1520000 8.37898158297e-07
1680000 7.29273608632e-07
1840000 6.26728849981e-07
2000000 5.47578198357e-07
2160000 5.8744532101e-07
2320000 5.64134465648e-07
2480000 5.96412451207e-07
2640000 4.8389229721e-07
2800000 4.83121994243e-07
2960000 4.3691945259e-07
3120000 4.89254144939e-07
3280000 4.10832481798e-07
3440000 2.15101026399e-07
3600000 3.06141278043e-07
3760000 3.18607232456e-07
3920000 2.65238503526e-07
4080000 3.0041934951e-07
4240000 2.53192054782e-07
4400000 2.63779610751e-07
4560000 2.78219670058e-07
4720000 2.03583458037e-07
4880000 2.26760027302e-07
5040000 2.17739691451e-07
5200000 2.45317656327e-07
5360000 1.71762392853e-07
5520000 1.24838310073e-07
5680000 1.28370124114e-07
5840000 9.63528589458e-08
6000000 1.61052253866e-07
6160000 9.43939378938e-08
6320000 1.18612059078e-07
6480000 1.09789411281e-07
6640000 8.78216567237e-08
6800000 6.88270958714e-08
6960000 7.69181508225e-08
7120000 7.9576210491e-08
7280000 9.44208638128e-08
7440000 6.20600523107e-08
7600000 9.22048893037e-08
7760000 7.66475903091e-08
7920000 5.20946527844e-08
8080000 5.4136082615e-08
8240000 1.07322189434e-07
8400000 5.22979859477e-08
8560000 5.87741679593e-08
8720000 3.55172277593e-08
8880000 4.91609059101e-08
9040000 5.76866642986e-08
9200000 4.60384809384e-08
9360000 3.62969161451e-08
9520000 5.49892760188e-08
9680000 5.44477883219e-08
9840000 2.18744635844e-08
10000000 1.8874331723e-08
};
\addplot [thick, mark=triangle, color1]
table {%
-2320000 0
-2160000 0
-2000000 0
-1840000 0
-1680000 0
-1520000 0
-1360000 0
-1200000 0
-1040000 0
-880000 0
-720000 0
-560000 0
-400000 0
-240000 0
-80000 0
80000 0
240000 0
400000 0
560000 0
720000 4.00042009429e-07
880000 6.96983806137e-06
1040000 5.69340007043e-06
1200000 1.66674579785e-07
1360000 0
1520000 1.01600506161e-07
1680000 6.94209990328e-07
1840000 2.37854163801e-06
2000000 2.17365659173e-06
2160000 1.16724786679e-06
2320000 0
2480000 0
2640000 0
2800000 0
2960000 0
3120000 0
3280000 0
3440000 3.66971227957e-06
3600000 0
3760000 0
3920000 0
4080000 0
4240000 0
4400000 0
4560000 0
4720000 0
4880000 0
5040000 0
5200000 2.42287725152e-06
5360000 0
5520000 0
5680000 0
5840000 0
6000000 0
6160000 0
6320000 0
6480000 0
6640000 0
6800000 0
6960000 1.80087659155e-06
7120000 0
7280000 0
7440000 0
7600000 0
7760000 0
7920000 0
8080000 0
8240000 0
8400000 0
8560000 0
8720000 0
8880000 0
9040000 0
9200000 0
9360000 0
9520000 0
9680000 0
9840000 0
10000000 0
};
\addplot [thick, color2]
table {%
-2320000 0
-2160000 0
-2000000 0
-1840000 0
-1680000 0
-1520000 0
-1360000 0
-1200000 0
-1040000 0
-880000 0
-720000 0
-560000 0
-400000 0
-240000 0
-80000 0
80000 0
240000 0
400000 0
560000 0
720000 0
880000 1.3229954721e-05
1040000 0
1200000 0
1360000 0
1520000 0
1680000 0
1840000 0
2000000 6.51525659302e-06
2160000 0
2320000 0
2480000 0
2640000 0
2800000 0
2960000 0
3120000 0
3280000 0
3440000 3.66971227957e-06
3600000 0
3760000 0
3920000 0
4080000 0
4240000 0
4400000 0
4560000 0
4720000 0
4880000 0
5040000 0
5200000 2.42287725152e-06
5360000 0
5520000 0
5680000 0
5840000 0
6000000 0
6160000 0
6320000 0
6480000 0
6640000 0
6800000 0
6960000 1.80087659155e-06
7120000 0
7280000 0
7440000 0
7600000 0
7760000 0
7920000 0
8080000 0
8240000 0
8400000 0
8560000 0
8720000 0
8880000 0
9040000 0
9200000 0
9360000 0
9520000 0
9680000 0
9840000 0
10000000 0
};
%
%
%

\end{axis}

\end{tikzpicture}}

\end{minipage}
 \caption{Numerical SDF over the domain space as a function of the Gauss curvature: localized SDF (dashed-line), averaged SDF with $\averglen=25$ (triangle) and $\averglen=55$ (solid line). }
\label{fig:sdfG}
\end{figure}

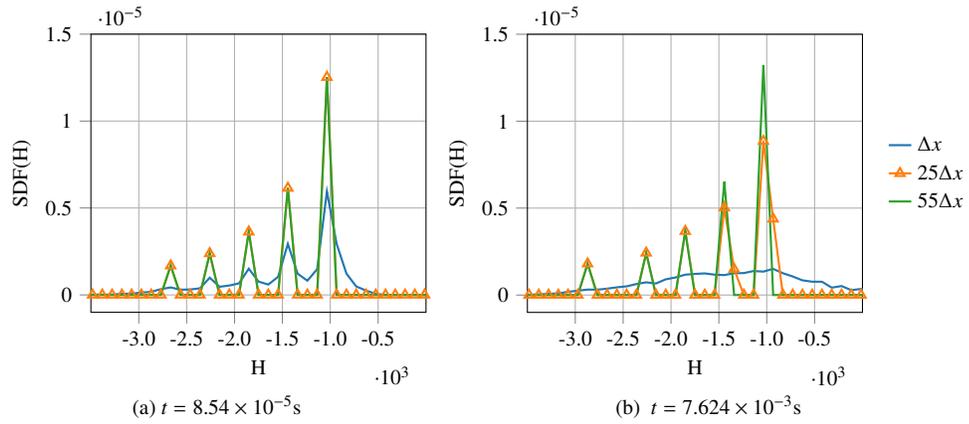
\begin{figure}
\begin{minipage}{.48\linewidth}
  \centering
  \subfloat[$t=8.54 \times 10^{-5}$s]{
\pgfplotsset{
compat=1.3,
legend image code/.code={
\draw[mark repeat=2,mark phase=2]
plot coordinates {
(0cm,0cm)
(0.15cm,0cm)        
(0.3cm,0cm)         
};%
}
}
\footnotesize

\begin{tikzpicture}

\definecolor{color1}{rgb}{1,0.498039215686275,0.0549019607843137}
\definecolor{color0}{rgb}{0.12156862745098,0.466666666666667,0.705882352941177}
\definecolor{color2}{rgb}{0.172549019607843,0.627450980392157,0.172549019607843}

\begin{axis}[
scale=0.65,
xlabel={H},
ylabel={SDF(H)},
xmin=-3500, xmax=0,
ymin=-0.1e-5, ymax=1.5e-05,
xtick={-3000, -2500, -2000, -1500, -1000, -500},
xticklabels={-3.0, -2.5, -2.0, -1.5, -1.0, -0.5},
scaled x ticks=base 10:-3,
tick align=outside,
tick pos=left,
xmajorgrids,
x grid style={lightgray!92.026143790849673!black},
ymajorgrids,
y grid style={lightgray!92.026143790849673!black},
]
\addplot [thick, color0]
table {%
-3687 3.51877446332e-09
-3585 1.80527984538e-08
-3483 1.83918318225e-08
-3381 1.41448132848e-08
-3279 2.01092552887e-08
-3177 6.27196852826e-08
-3075 7.15816624372e-08
-2973 1.35403845486e-07
-2871 1.76565456093e-07
-2769 3.31741176289e-07
-2667 4.29999406306e-07
-2565 2.995058712e-07
-2463 3.09987940152e-07
-2361 3.78689203754e-07
-2259 9.91488095744e-07
-2157 4.6837653967e-07
-2055 5.39355750827e-07
-1953 6.63781304355e-07
-1851 1.50745832406e-06
-1749 7.64863810451e-07
-1647 5.95977980838e-07
-1545 1.04079539745e-06
-1443 2.93696364041e-06
-1341 1.22171100405e-06
-1239 8.16113742744e-07
-1137 1.48044249204e-06
-1035 5.94062029132e-06
-933 2.92459806636e-06
-831 1.22520852134e-06
-729 4.87165645065e-07
-627 2.52864735116e-07
-525 5.67649234546e-08
-423 0
-321 0
-219 0
-117 0
-15 0
87 0
};
\addplot [thick, mark=triangle, color1]
table {%
-3687 0
-3585 0
-3483 0
-3381 0
-3279 0
-3177 0
-3075 0
-2973 0
-2871 0
-2769 0
-2667 1.68089444254e-06
-2565 0
-2463 0
-2361 0
-2259 2.3853978597e-06
-2157 0
-2055 0
-1953 0
-1851 3.62889605373e-06
-1749 0
-1647 0
-1545 0
-1443 6.14581562593e-06
-1341 0
-1239 0
-1137 0
-1035 1.25237368138e-05
-933 0
-831 0
-729 0
-627 0
-525 0
-423 0
-321 0
-219 0
-117 0
-15 0
87 0
};
\addplot [thick, color2]
table {%
-3687 0
-3585 0
-3483 0
-3381 0
-3279 0
-3177 0
-3075 0
-2973 0
-2871 0
-2769 0
-2667 1.68089444254e-06
-2565 0
-2463 0
-2361 0
-2259 2.3853978597e-06
-2157 0
-2055 0
-1953 0
-1851 3.62889605373e-06
-1749 0
-1647 0
-1545 0
-1443 6.14581562593e-06
-1341 0
-1239 0
-1137 0
-1035 1.25237368138e-05
-933 0
-831 0
-729 0
-627 0
-525 0
-423 0
-321 0
-219 0
-117 0
-15 0
87 0
};
%
%
%

\end{axis}

\end{tikzpicture}}

\end{minipage}%
\begin{minipage}{.48\linewidth}
  \subfloat[ $t= 7.624\times 10^{-3}$s]{

\pgfplotsset{
compat=1.3,
legend image code/.code={
\draw[mark repeat=2,mark phase=2]
plot coordinates {
(0cm,0cm)
(0.15cm,0cm)        
(0.3cm,0cm)         
};%
}
}
\footnotesize

\begin{tikzpicture}

\definecolor{color1}{rgb}{1,0.498039215686275,0.0549019607843137}
\definecolor{color0}{rgb}{0.12156862745098,0.466666666666667,0.705882352941177}
\definecolor{color2}{rgb}{0.172549019607843,0.627450980392157,0.172549019607843}

\begin{axis}[
scale=0.65,
xlabel={H},
ylabel={SDF(H)},
xmin=-3500, xmax=0,
ymin=-0.1e-5, ymax=1.5e-05,
xtick={-3000, -2500, -2000, -1500, -1000, -500},
xticklabels={-3.0, -2.5, -2.0, -1.5, -1.0, -0.5},
scaled x ticks=base 10:-3,
tick align=outside,
tick pos=left,
xmajorgrids,
x grid style={lightgray!92.026143790849673!black},
ymajorgrids,
y grid style={lightgray!92.026143790849673!black},
legend style={at={(1.05,0.5)}, anchor=west, draw=none},
legend entries={{$\Delta x$},{$25 \Delta x$},{$55 \Delta x$}},
legend cell align={left}
]
\addplot [thick, color0]
table {%
-3687 3.70505697729e-08
-3585 1.01450507251e-08
-3483 1.36304404749e-08
-3381 3.69014094473e-08
-3279 6.67363323101e-08
-3177 9.50687187887e-08
-3075 1.82856678176e-07
-2973 2.48543012739e-07
-2871 3.13230715763e-07
-2769 3.06377763002e-07
-2667 3.65649718181e-07
-2565 4.42124338761e-07
-2463 4.99416257031e-07
-2361 6.31988311345e-07
-2259 7.29860797671e-07
-2157 6.6338306804e-07
-2055 9.03140447352e-07
-1953 1.00943764953e-06
-1851 1.16984826492e-06
-1749 1.21393413055e-06
-1647 1.24225060475e-06
-1545 1.17074026167e-06
-1443 1.14576885867e-06
-1341 1.25582297207e-06
-1239 1.26218201897e-06
-1137 1.3793564438e-06
-1035 1.34618379629e-06
-933 1.50398360819e-06
-831 1.25394911801e-06
-729 1.06716429466e-06
-627 8.44552078307e-07
-525 7.66262240387e-07
-423 7.64492627212e-07
-321 4.23516048556e-07
-219 5.16632192433e-07
-117 2.79873836089e-07
-15 3.47370116733e-07
87 2.1024920068e-07
};
\addplot [thick, mark=triangle, color1]
table {%
-3687 0
-3585 0
-3483 0
-3381 0
-3279 0
-3177 0
-3075 0
-2973 0
-2871 1.80087659155e-06
-2769 0
-2667 0
-2565 0
-2463 0
-2361 0
-2259 2.42287725152e-06
-2157 0
-2055 0
-1953 0
-1851 3.66971227957e-06
-1749 0
-1647 0
-1545 0
-1443 5.0237523002e-06
-1341 1.49150429282e-06
-1239 0
-1137 0
-1035 8.85010575939e-06
-933 4.37984896162e-06
-831 0
-729 0
-627 0
-525 0
-423 0
-321 0
-219 0
-117 0
-15 0
87 0
};
\addplot [thick, color2]
table {%
-3687 0
-3585 0
-3483 0
-3381 0
-3279 0
-3177 0
-3075 0
-2973 0
-2871 1.80087659155e-06
-2769 0
-2667 0
-2565 0
-2463 0
-2361 0
-2259 2.42287725152e-06
-2157 0
-2055 0
-1953 0
-1851 3.66971227957e-06
-1749 0
-1647 0
-1545 0
-1443 6.51525659302e-06
-1341 0
-1239 0
-1137 0
-1035 1.3229954721e-05
-933 0
-831 0
-729 0
-627 0
-525 0
-423 0
-321 0
-219 0
-117 0
-15 0
87 0
};
%
%
%

\end{axis}

\end{tikzpicture}}

\end{minipage}
 \caption{Numerical SDF over the domain space as a function of the mean curvature: localized SDF (dashed-line), averaged SDF with $\averglen=25$ (triangle) and $\averglen=55$ (solid line). }
\label{fig:sdfH}
\end{figure}

\begin{figure}
\begin{minipage}{.48\linewidth}
  \centering
  \subfloat[$t=8.54 \times 10^{-5}$s]{
\pgfplotsset{
compat=1.3,
legend image code/.code={
\draw[mark repeat=2,mark phase=2]
plot coordinates {
(0cm,0cm)
(0.15cm,0cm)        
(0.3cm,0cm)         
};%
}
}
\footnotesize

\begin{tikzpicture}

\definecolor{color1}{rgb}{1,0.498039215686275,0.0549019607843137}
\definecolor{color0}{rgb}{0.12156862745098,0.466666666666667,0.705882352941177}
\definecolor{color2}{rgb}{0.172549019607843,0.627450980392157,0.172549019607843}

\begin{axis}[
scale=0.65,
xlabel={G},
ylabel={N(G)},
xmin=-2000000, xmax=9000000,
ymin=-0.05, ymax=1.05,
xtick={0,2000000,4000000,6000000,8000000},
xticklabels={0,2,4,6,8},
tick align=outside,
tick pos=left,
xmajorgrids,
x grid style={lightgray!92.026143790849673!black},
ymajorgrids,
y grid style={lightgray!92.026143790849673!black},
]
\addplot [thick, color0]
table {%
-2320000 0
-2160000 0
-2000000 0
-1840000 0
-1680000 0
-1520000 0
-1360000 0
-1200000 0
-1040000 0
-880000 0
-720000 0
-560000 0
-400000 0
-240000 0
-80000 0
80000 0.000896996787405
240000 0.00876315613506
400000 0.0127283280344
560000 0.0197169426627
720000 0.0601576800373
880000 0.123912799089
1040000 0.476024220088
1200000 0.15921744726
1360000 0.110927756107
1520000 0.0673215442103
1680000 0.068210512061
1840000 0.0918751974789
2000000 0.348258159232
2160000 0.1435923266
2320000 0.109013692727
2480000 0.0843793691715
2640000 0.0736197052423
2800000 0.0545780466089
2960000 0.0556903719708
3120000 0.0891644786516
3280000 0.112505488725
3440000 0.279732159598
3600000 0.120305167536
3760000 0.0810986997059
3920000 0.0732876230822
4080000 0.100721841945
4240000 0.0485428796297
4400000 0.0522950732651
4560000 0.026954218921
4720000 0.0501669497495
4880000 0.0506598191737
5040000 0.0574379584967
5200000 0.130671881291
5360000 0.213620849202
5520000 0.0758892913376
5680000 0.0948993364247
5840000 0.0348943237337
6000000 0.047248663922
6160000 0.0649181115526
6320000 0.0432926426011
6480000 0.0483288141199
6640000 0.0372673840857
6800000 0.0180728731793
6960000 0.0340048890529
7120000 0.0598023925229
7280000 0.0893080834504
7440000 0.104002397591
7600000 0.119016999636
7760000 0.0350289762853
7920000 0.0422561335848
8080000 0.0544680988102
8240000 0.0394686261759
8400000 0.0104545028335
8560000 0.0266388774312
8720000 0.0451428071627
8880000 0.026902455143
9040000 0.0234328644199
9200000 0.03436333113
9360000 0.0208459223567
9520000 0.00428784719751
9680000 0.00641419763564
9840000 0.0205122792503
10000000 0.0170872764923 
};
\addplot [thick, mark=triangle, color1]
table {%
-2320000 0
-2160000 0
-2000000 0
-1840000 0
-1680000 0
-1520000 0
-1360000 0
-1200000 0
-1040000 0
-880000 0
-720000 0
-560000 0
-400000 0
-240000 0
-80000 0
80000 0
240000 0
400000 0
560000 0
720000 0
880000 0
1040000 1.00000000006
1200000 0
1360000 0
1520000 0
1680000 0
1840000 0
2000000 1.00000000003
2160000 0
2320000 0
2480000 0
2640000 0
2800000 0
2960000 0
3120000 0
3280000 0
3440000 1.00000000002
3600000 0
3760000 0
3920000 0
4080000 0
4240000 0
4400000 0
4560000 0
4720000 0
4880000 0
5040000 0
5200000 1.00000000001
5360000 0
5520000 0
5680000 0
5840000 0
6000000 0
6160000 0
6320000 0
6480000 0
6640000 0
6800000 0
6960000 0
7120000 0
7280000 0
7440000 1.00000000001
7600000 0
7760000 0
7920000 0
8080000 0
8240000 0
8400000 0
8560000 0
8720000 0
8880000 0
9040000 0
9200000 0
9360000 0
9520000 0
9680000 0
9840000 0
10000000 0 
};
\addplot [thick, color2]
table {%
-2320000 0
-2160000 0
-2000000 0
-1840000 0
-1680000 0
-1520000 0
-1360000 0
-1200000 0
-1040000 0
-880000 0
-720000 0
-560000 0
-400000 0
-240000 0
-80000 0
80000 0
240000 0
400000 0
560000 0
720000 0
880000 0
1040000 1.00000000006
1200000 0
1360000 0
1520000 0
1680000 0
1840000 0
2000000 1.00000000003
2160000 0
2320000 0
2480000 0
2640000 0
2800000 0
2960000 0
3120000 0
3280000 0
3440000 1.00000000002
3600000 0
3760000 0
3920000 0
4080000 0
4240000 0
4400000 0
4560000 0
4720000 0
4880000 0
5040000 0
5200000 1.00000000001
5360000 0
5520000 0
5680000 0
5840000 0
6000000 0
6160000 0
6320000 0
6480000 0
6640000 0
6800000 0
6960000 0
7120000 0
7280000 0
7440000 1.00000000001
7600000 0
7760000 0
7920000 0
8080000 0
8240000 0
8400000 0
8560000 0
8720000 0
8880000 0
9040000 0
9200000 0
9360000 0
9520000 0
9680000 0
9840000 0
10000000 0 
};
%
%
%

\end{axis}

\end{tikzpicture}}

\end{minipage}%
\begin{minipage}{.48\linewidth}
  \subfloat[ $t= 7.624times 10^{-3}$s]{

\pgfplotsset{
compat=1.3,
legend image code/.code={
\draw[mark repeat=2,mark phase=2]
plot coordinates {
(0cm,0cm)
(0.15cm,0cm)        
(0.3cm,0cm)         
};%
}
}
\footnotesize

\begin{tikzpicture}

\definecolor{color1}{rgb}{1,0.498039215686275,0.0549019607843137}
\definecolor{color0}{rgb}{0.12156862745098,0.466666666666667,0.705882352941177}
\definecolor{color2}{rgb}{0.172549019607843,0.627450980392157,0.172549019607843}

\begin{axis}[
scale=0.65,
xlabel={G},
ylabel={N(G)},
xmin=-2000000, xmax=9000000,
ymin=-0.05, ymax=1.05,
xtick={0,2000000,4000000,6000000,8000000},
xticklabels={0,2,4,6,8},
tick align=outside,
tick pos=left,
xmajorgrids,
x grid style={lightgray!92.026143790849673!black},
ymajorgrids,
y grid style={lightgray!92.026143790849673!black},
legend style={at={(1.05,0.5)}, anchor=west, draw=none},
legend entries={{$\Delta x$},{$25 \Delta x$},{$55 \Delta x$}},
legend cell align={left}
]
\addplot [thick, color0]
table {%
-2320000 -0.00550217187379
-2160000 -0.00251368018731
-2000000 -0.00901721328757
-1840000 -0.00889887719827
-1680000 -0.0132523310386
-1520000 -0.00778378676772
-1360000 -0.00545864707104
-1200000 -0.00941957965577
-1040000 -0.012715753133
-880000 -0.0111061503606
-720000 -0.00872574911047
-560000 -0.0121399741585
-400000 -0.011176967463
-240000 -0.0132134887579
-80000 -0.0070330369181
80000 0.00946411497978
240000 0.0259430090994
400000 0.0398972602224
560000 0.0535455864871
720000 0.0661588399856
880000 0.071364399
1040000 0.0717510451061
1200000 0.0869763031249
1360000 0.0863258138866
1520000 0.101925750625
1680000 0.0971104335193
1840000 0.0916972073445
2000000 0.0871964949371
2160000 0.101020937745
2320000 0.103901654389
2480000 0.117627418227
2640000 0.101547637943
2800000 0.1075496932
2960000 0.103126808447
3120000 0.121167822862
3280000 0.106799917564
3440000 0.0590198847868
3600000 0.0874871562242
3760000 0.0950745609666
3920000 0.0823677549864
4080000 0.0975936100893
4240000 0.0850543340457
4400000 0.0922863640987
4560000 0.100992717269
4720000 0.0763958851114
4880000 0.0876801605277
5040000 0.087108423109
5200000 0.10140040328
5360000 0.0734322188774
5520000 0.0548184800668
5680000 0.0578307661843
5840000 0.0446382290582
6000000 0.0767764978415
6160000 0.046162382743
6320000 0.0598006282665
6480000 0.0565046162109
6640000 0.0464932280896
6800000 0.0373021950061
6960000 0.0425362803062
7120000 0.0450088785747
7280000 0.0545589152497
7440000 0.0367913921108
7600000 0.0556573653269
7760000 0.0472067450496
7920000 0.0328198902664
8080000 0.0347939496408
8240000 0.0703572343772
8400000 0.0348589036073
8560000 0.0401744026758
8720000 0.0246243861443
8880000 0.0349308972672
9040000 0.0415762349952
9200000 0.0336588883555
9360000 0.0270250089692
9520000 0.041479500214
9680000 0.0418767379869
9840000 0.0172187059581
10000000 0.0149467249627
};
\addplot [thick, mark=triangle,  color1]
table {%
-2320000 0
-2160000 0
-2000000 0
-1840000 0
-1680000 0
-1520000 0
-1360000 0
-1200000 0
-1040000 0
-880000 0
-720000 0
-560000 0
-400000 0
-240000 0
-80000 0
80000 0
240000 0
400000 0
560000 0
720000 0.0239633051623
880000 0.502003905049
1040000 0.458811644328
1200000 0.0152211455202
1360000 0
1520000 0.0126855993176
1680000 0.0934872760469
1840000 0.351401127175
2000000 0.345429278202
2160000 0.196996719289
2320000 0
2480000 0
2640000 0
2800000 0
2960000 0
3120000 0
3280000 0
3440000 1.00000000002
3600000 0
3760000 0
3920000 0
4080000 0
4240000 0
4400000 0
4560000 0
4720000 0
4880000 0
5040000 0
5200000 1.00000000001
5360000 0
5520000 0
5680000 0
5840000 0
6000000 0
6160000 0
6320000 0
6480000 0
6640000 0
6800000 0
6960000 1.00000000001
7120000 0
7280000 0
7440000 0
7600000 0
7760000 0
7920000 0
8080000 0
8240000 0
8400000 0
8560000 0
8720000 0
8880000 0
9040000 0
9200000 0
9360000 0
9520000 0
9680000 0
9840000 0
10000000 0
};
\addplot [thick, color2]
table {%
-2320000 0
-2160000 0
-2000000 0
-1840000 0
-1680000 0
-1520000 0
-1360000 0
-1200000 0
-1040000 0
-880000 0
-720000 0
-560000 0
-400000 0
-240000 0
-80000 0
80000 0
240000 0
400000 0
560000 0
720000 0
880000 1.00000000006
1040000 0
1200000 0
1360000 0
1520000 0
1680000 0
1840000 0
2000000 1.00000000003
2160000 0
2320000 0
2480000 0
2640000 0
2800000 0
2960000 0
3120000 0
3280000 0
3440000 1.00000000002
3600000 0
3760000 0
3920000 0
4080000 0
4240000 0
4400000 0
4560000 0
4720000 0
4880000 0
5040000 0
5200000 1.00000000001
5360000 0
5520000 0
5680000 0
5840000 0
6000000 0
6160000 0
6320000 0
6480000 0
6640000 0
6800000 0
6960000 1.00000000001
7120000 0
7280000 0
7440000 0
7600000 0
7760000 0
7920000 0
8080000 0
8240000 0
8400000 0
8560000 0
8720000 0
8880000 0
9040000 0
9200000 0
9360000 0
9520000 0
9680000 0
9840000 0
10000000 0
};
%
%
%

\end{axis}

\end{tikzpicture}}
  
\end{minipage}
 \caption{Numerical GNDF over the domain space as a function of the Gauss curvature: localized GNDF (dashed-line), averaged GNDF with $\averglen=25$ (triangle) and $\averglen=55$ (solid line). }
\label{fig:ndfG}
\end{figure}
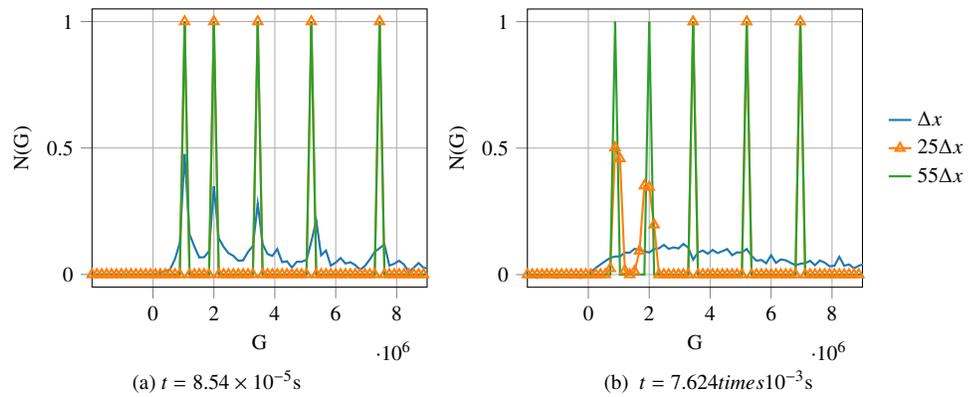

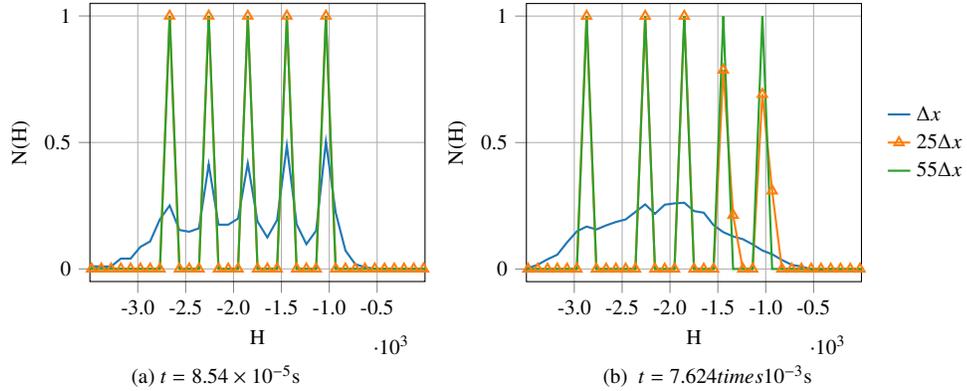
\begin{figure}
\begin{minipage}{.48\linewidth}
  \centering
  \subfloat[$t=8.54 \times 10^{-5}$s]{

\pgfplotsset{
compat=1.3,
legend image code/.code={
\draw[mark repeat=2,mark phase=2]
plot coordinates {
(0cm,0cm)
(0.15cm,0cm)        
(0.3cm,0cm)         
};%
}
}
\footnotesize

\begin{tikzpicture}

\definecolor{color1}{rgb}{1,0.498039215686275,0.0549019607843137}
\definecolor{color0}{rgb}{0.12156862745098,0.466666666666667,0.705882352941177}
\definecolor{color2}{rgb}{0.172549019607843,0.627450980392157,0.172549019607843}

\begin{axis}[
scale=0.65,
xlabel={H},
ylabel={N(H)},
xmin=-3500, xmax=0,
ymin=-0.05, ymax=1.05,
xtick={-3000, -2500, -2000, -1500, -1000, -500},
xticklabels={-3.0, -2.5, -2.0, -1.5, -1.0, -0.5},
scaled x ticks=base 10:-3,
tick align=outside,
tick pos=left,
xmajorgrids,
x grid style={lightgray!92.026143790849673!black},
ymajorgrids,
y grid style={lightgray!92.026143790849673!black},
]
\addplot [thick, color0]
table {%
-3687 0.00230310718567
-3585 0.00886754005949
-3483 0.0103206026483
-3381 0.00944621965817
-3279 0.00949984718071
-3177 0.0406645743746
-3075 0.040824298076
-2973 0.0869422781123
-2871 0.108676772701
-2769 0.196282973421
-2667 0.250617493335
-2565 0.153782106315
-2463 0.146877608258
-2361 0.160535399107
-2259 0.412916377321
-2157 0.174738190023
-2055 0.174903648167
-1953 0.198740099628
-1851 0.417297874088
-1749 0.187280877853
-1647 0.124281322972
-1545 0.193039260125
-1443 0.486992917599
-1341 0.176878023158
-1239 0.0977459467802
-1137 0.151204316321
-1035 0.502043225652
-933 0.222811882874
-831 0.0724849959193
-729 0.0187600431369
-627 0.00625206350928
-525 0.000886218482615
-423 0
-321 0
-219 0
-117 0
-15 0
87 0
};
\addplot [thick, mark=triangle, color1]
table {%
-3687 0
-3585 0
-3483 0
-3381 0
-3279 0
-3177 0
-3075 0
-2973 0
-2871 0
-2769 0
-2667 1.00000000001
-2565 0
-2463 0
-2361 0
-2259 1.00000000001
-2157 0
-2055 0
-1953 0
-1851 1.00000000002
-1749 0
-1647 0
-1545 0
-1443 1.00000000003
-1341 0
-1239 0
-1137 0
-1035 1.00000000006
-933 0
-831 0
-729 0
-627 0
-525 0
-423 0
-321 0
-219 0
-117 0
-15 0
87 0
};
\addplot [thick, color2]
table {%
-3687 0
-3585 0
-3483 0
-3381 0
-3279 0
-3177 0
-3075 0
-2973 0
-2871 0
-2769 0
-2667 1.00000000001
-2565 0
-2463 0
-2361 0
-2259 1.00000000001
-2157 0
-2055 0
-1953 0
-1851 1.00000000002
-1749 0
-1647 0
-1545 0
-1443 1.00000000003
-1341 0
-1239 0
-1137 0
-1035 1.00000000006
-933 0
-831 0
-729 0
-627 0
-525 0
-423 0
-321 0
-219 0
-117 0
-15 0
87 0
};
%
%
%

\end{axis}

\end{tikzpicture}}

\end{minipage}%
\begin{minipage}{.48\linewidth}
  \subfloat[ $t= 7.624times 10^{-3}$s]{

\pgfplotsset{
compat=1.3,
legend image code/.code={
\draw[mark repeat=2,mark phase=2]
plot coordinates {
(0cm,0cm)
(0.15cm,0cm)        
(0.3cm,0cm)         
};%
}
}
\footnotesize

\begin{tikzpicture}

\definecolor{color1}{rgb}{1,0.498039215686275,0.0549019607843137}
\definecolor{color0}{rgb}{0.12156862745098,0.466666666666667,0.705882352941177}
\definecolor{color2}{rgb}{0.172549019607843,0.627450980392157,0.172549019607843}

\begin{axis}[
scale=0.65,
xlabel={H},
ylabel={N(H)},
xmin=-3500, xmax=0,
ymin=-0.05, ymax=1.05,
xtick={-3000, -2500, -2000, -1500, -1000, -500},
xticklabels={-3.0, -2.5, -2.0, -1.5, -1.0, -0.5},
scaled x ticks=base 10:-3,
tick align=outside,
tick pos=left,
xmajorgrids,
x grid style={lightgray!92.026143790849673!black},
ymajorgrids,
y grid style={lightgray!92.026143790849673!black},
legend style={at={(1.05,0.5)}, anchor=west, draw=none},
legend entries={{$\Delta x$},{$25 \Delta x$},{$55 \Delta x$}},
legend cell align={left}
]
\addplot [thick, color0]
table {%
-3687 0.0154040154751
-3585 0.00151192070056
-3483 0.00361070932454
-3381 0.0168537338571
-3279 0.0385048963625
-3177 0.0565259859283
-3075 0.105357004149
-2973 0.148869860281
-2871 0.167793692315
-2769 0.156289995189
-2667 0.172395726678
-2565 0.18589894921
-2463 0.195895750951
-2361 0.225756112837
-2259 0.255504140851
-2157 0.218305631973
-2055 0.254305380388
-1953 0.259426528463
-1851 0.262053363736
-1749 0.229042886947
-1647 0.222568103667
-1545 0.173229730804
-1443 0.145611653829
-1341 0.128805487951
-1239 0.117961511785
-1137 0.0972839849613
-1035 0.0732285006111
-933 0.0582359724997
-831 0.0358662944371
-729 0.016395916163
-627 0.00974726794705
-525 -0.0014476992375
-423 -0.00538436039
-321 -0.00897984239305
-219 -0.00849160687697
-117 -0.00976116122744
-15 -0.00839912866228
87 -0.00561763896972
};
\addplot [thick, mark=triangle, color1]
table {%
-3687 0
-3585 0
-3483 0
-3381 0
-3279 0
-3177 0
-3075 0
-2973 0
-2871 1.00000000001
-2769 0
-2667 0
-2565 0
-2463 0
-2361 0
-2259 1.00000000001
-2157 0
-2055 0
-1953 0
-1851 1.00000000002
-1749 0
-1647 0
-1545 0
-1443 0.787173675831
-1341 0.2128263242
-1239 0
-1137 0
-1035 0.69089251829
-933 0.30910748177
-831 0
-729 0
-627 0
-525 0
-423 0
-321 0
-219 0
-117 0
-15 0
87 0
};
\addplot [thick, color2]
table {%
-3687 0
-3585 0
-3483 0
-3381 0
-3279 0
-3177 0
-3075 0
-2973 0
-2871 1.00000000001
-2769 0
-2667 0
-2565 0
-2463 0
-2361 0
-2259 1.00000000001
-2157 0
-2055 0
-1953 0
-1851 1.00000000002
-1749 0
-1647 0
-1545 0
-1443 1.00000000003
-1341 0
-1239 0
-1137 0
-1035 1.00000000006
-933 0
-831 0
-729 0
-627 0
-525 0
-423 0
-321 0
-219 0
-117 0
-15 0
87 0
};
%
%
%

\end{axis}

\end{tikzpicture}}

\end{minipage}
 \caption{Numerical GNDF over the domain space as a function of the mean curvature: localized GNDF (dashed-line), averaged GNDF with $\averglen=25$ (triangle) and $\averglen=55$ (solid line). }
\label{fig:ndfH}
\end{figure}

\FloatBarrier


\subsection{Two droplets collision}
In this section, we consider a simulation of two water droplets collision in the stretching separation regime \cite{Poo}. The two droplets are initially separated (Figure \ref{fig:collision1}) in a periodic domain, with initial velocity, 
then the collision leads to a temporary coalescence of the two droplets (Figure \ref{fig:collision2})
by forming a one unstable big droplet. Due to competition between inertial effect and surface tension force, 
the unstable droplet deforms, first into a torus shape (Figure \ref{fig:collision3})
and finally breaks up into small different droplets (Figure \ref{fig:collision4}).\\
The parameters of collision yield some important topological changes and are given in the following table \ref{tab:collision}:
\begin{table}[h]
\begin{center}
    \begin{tabular}{c| c| c| c| c}
      $D_s$ ($\mu m$) & $D_l$ ($\mu m$) & $U_c$ ($m.s^{-1}$) & $We=\frac{\rho_lU_cD_s}{\sigma}$  & $x$\\\hline
      $260$ & $400$ & 4& $60$ & $0.42$ \\
     
    \end{tabular}
\end{center}
\caption{Parameters of droplet collision. Subscript $s$ and $l$ are related respectively to the small and large droplets. $U_c$ is the relative
velocity and $x$ the dimensionless impact parameter\cite{Poo}. Classical water and air properties are used.}
\label{tab:collision} 
\end{table}

\begin{figure}
\begin{minipage}{.45\linewidth}
  \centering
  \subfloat[$t=1$ \label{fig:collision1}]{\includegraphics[width=1.1\linewidth]{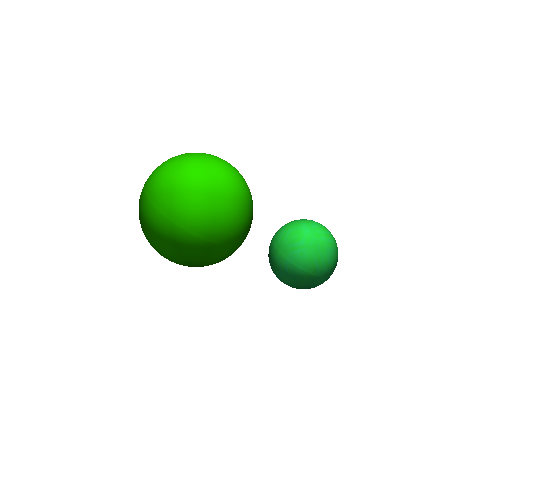}}
\end{minipage}%
\begin{minipage}{.55\linewidth}
    \subfloat[$t=7$\label{fig:collision2}]{\includegraphics[width=1.1\linewidth]{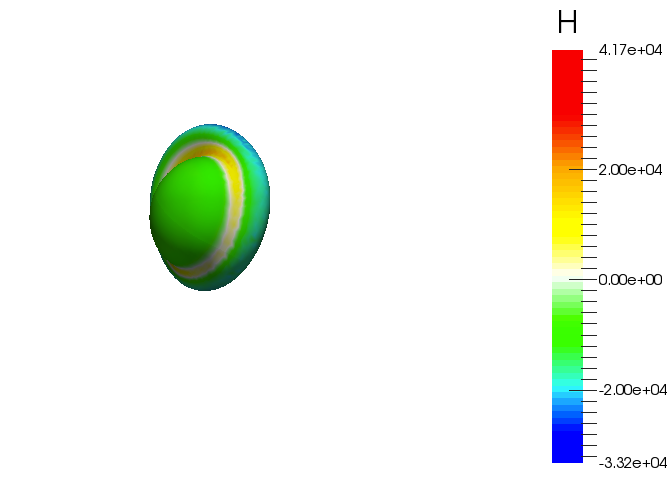}}
\end{minipage}
\begin{minipage}{.45\linewidth}
  \centering
  \subfloat[$t=14$\label{fig:collision3}]{\includegraphics[width=1.1\linewidth]{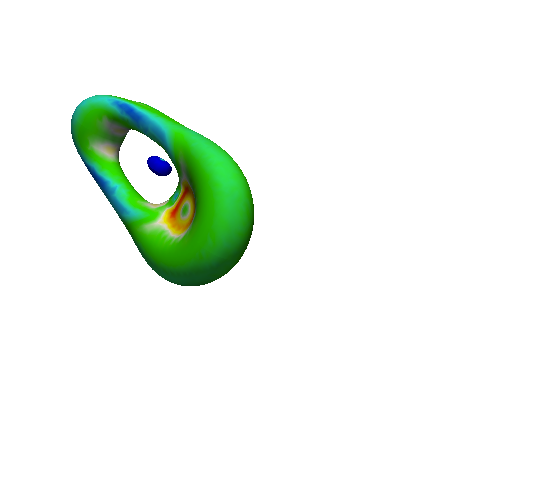}}
\end{minipage}%
\begin{minipage}{.55\linewidth}
   \subfloat[$t=60$\label{fig:collision4}]{\includegraphics[width=1.1\linewidth]{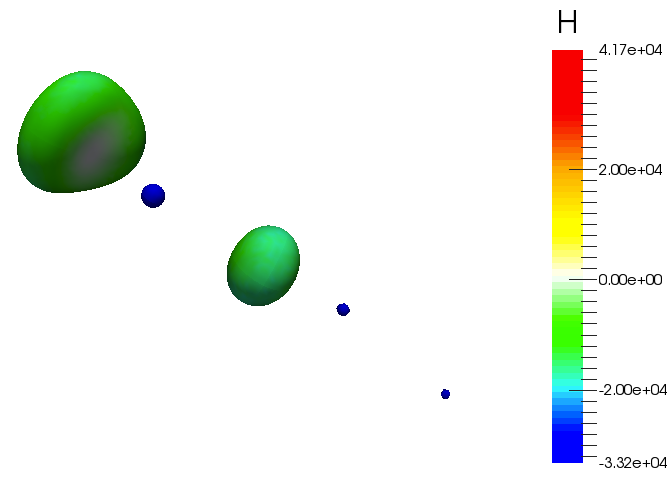}}
\end{minipage}
 \caption{Simulation of collision and stretching separation of two droplets. Surface colored according to the local mean curvature ($\Mcurv$) value.}
\label{fig:collision}
\end{figure}

In this simulation, we compute the time evolution of the volume integral over the  whole computational domain of the zero and first order moments of the localized and averaged SDFs ($\Sigma$, $\Sigma\widetilde{\Mcurv}$, $\Sigma\widetilde{\Gcurv}$). 
We can see from Figures \ref{fig:sigma_drop_collision}-\ref{fig:sigmaG_drop_collision}, 
that the moments of the averaged and localized SDF are equal at all time. This equality corresponds to the requirement \ref{item:requirement2}. The first two Figures \ref{fig:sigma_drop_collision}-\ref{fig:sigmaH_drop_collision} correspond  respectively to the total surface area and the total mean curvature. These two quantities evolve continuously and we can identify from the two curves the different droplet states:
\begin{itemize}
\item $t\in [0,2]$: the two droplets are initially separated,
\item $t\in [2,12]$: coalescence and stretching of the two droplets are characterized by a minimum total surface area at $t \in[3,4]$ just after the coalescence, then a maximum surface area at $t\in[11,12]$ just before the breakup process takes 
place, when the thin liquid film reaches its maximum. In the same time, we obtain a minimum of 
the absolute value of the total mean curvature, this can be explained by the positive values of mean curvature during coalescence. Then it increases until and during the break-up of thin liquid film.
\item $t\in[12,40]$: breakup cascade and coalescence of some small droplets take place. The total surface area decreases while the absolute value of the mean curvature increases,
\item $t\in[40,80]$: in the final state, we obtain five droplets, where we have a convergence of the total surface area and the total mean curvature toward stable values.
\end{itemize} 

The evolution of total Gauss curvature illustrated in Figure \ref{fig:sigmaG_drop_collision} is a discontinuous evolution and by dividing this quantity by $4\pi$, we obtain integer values. Indeed, the quantity $1/(4\pi)\int_{\xv}\Sigma\widetilde{\Gcurv} d\xv$ is equal to the sum of the half of the Euler characteristic of the objects included in the entire domain. This quantity will allow us to evaluate 
the topology evolution. In the case of droplets homeomorphic to spheres, we get the droplet number in the domain. This is the case here with the exception of the period of time between time  $t=10$ and $t=20$, where this quantity drops down from $1$ to $-2$ and then increases to  $0$ and $1$ before reaching again $2$ with two objects homeomorphic to a sphere. During this interval of time, the droplet formed by the coalescence form a tore shape with several holes with satellite droplets and then we come back to a regular tore with another droplet homeomorphic to a sphere, that is a total characteristic of $1$, before the tore closes, that is a characteristic of $2$, and then breaks into a total of three droplets at time $20$.

Thus, the proposed approach not only can lead to a statistics of objects through spatial averaging, when the whole set of objects are homeomorphic to spheres, but also provides some key informations about the topology of the interface.

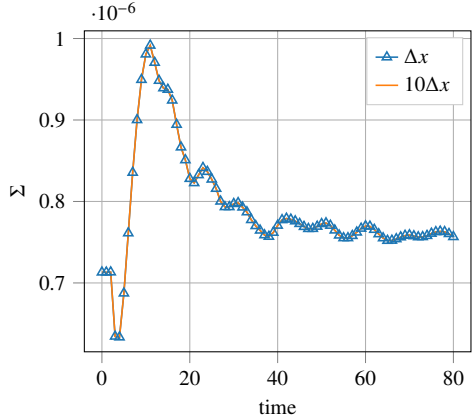
\begin{figure}
\centering{
\pgfplotsset{
compat=1.3,
legend image code/.code={
\draw[mark repeat=2,mark phase=2]
plot coordinates {
(0cm,0cm)
(0.15cm,0cm)        
(0.3cm,0cm)         
};%
}
}
\footnotesize

\begin{tikzpicture}

\definecolor{color1}{rgb}{1,0.498039215686275,0.0549019607843137}
\definecolor{color0}{rgb}{0.12156862745098,0.466666666666667,0.705882352941177}

\begin{axis}[
scale=0.75,
xlabel={time},
ylabel={$\Sigma$},
xmin=-3.99999999999235, xmax=83.9999999998393,
ymin=6.154973717977e-07, ymax=0.0000010094249215783,
tick align=outside,
tick pos=left,
xmajorgrids,
x grid style={lightgray!92.026143790849673!black},
ymajorgrids,
y grid style={lightgray!92.026143790849673!black},
legend entries={{$\Delta x$},{$10 \Delta x$}},
legend style={draw=white!80.0!black},
legend cell align={left}
]
\addplot [semithick, mark=triangle, color0]
table {%
0 7.1309907610264e-07
1 7.13015711529476e-07
2 7.13220817735677e-07
3 6.34382378145791e-07
4 6.33403169514836e-07
5 6.8734182386034e-07
6 7.61163595882156e-07
7 8.3547056120139e-07
8 9.00091704843702e-07
9 9.49447506609334e-07
10 9.80698814956157e-07
11 9.91519123860871e-07
12 9.70469298611629e-07
13 9.48339838385368e-07
14 9.3917780730405e-07
15 9.37334942158495e-07
16 9.24136955322708e-07
17 8.94338925750813e-07
18 8.6658584822515e-07
19 8.50856014724845e-07
20 8.27966207920211e-07
21 8.22739872202501e-07
22 8.32510773167108e-07
23 8.41335715064933e-07
24 8.36557986645946e-07
25 8.27140158797283e-07
26 8.15779628373981e-07
27 8.001448677149e-07
28 7.93265167755149e-07
29 7.93263094548874e-07
30 7.97146865862727e-07
31 7.98259568128301e-07
32 7.92849419177132e-07
33 7.86964564820488e-07
34 7.77434423031357e-07
35 7.7003413352216e-07
36 7.64554875523125e-07
37 7.58776411143742e-07
38 7.56903923864033e-07
39 7.61827263122192e-07
40 7.70529404967283e-07
41 7.77324771859658e-07
42 7.79093576109796e-07
43 7.77961842201853e-07
44 7.75865025947601e-07
45 7.72646440406219e-07
46 7.6900256002601e-07
47 7.66474873706379e-07
48 7.66571266890407e-07
49 7.69277400783262e-07
50 7.7247852667954e-07
51 7.73474726533e-07
52 7.70473961903981e-07
53 7.64648934407915e-07
54 7.5848438094228e-07
55 7.55176892830426e-07
56 7.55165765980093e-07
57 7.57908085744101e-07
58 7.62141420377305e-07
59 7.66852822782554e-07
60 7.69789039752357e-07
61 7.69372323983182e-07
62 7.65426237754377e-07
63 7.6003761923514e-07
64 7.55079962759117e-07
65 7.51927500993315e-07
66 7.52103674677743e-07
67 7.53930590341399e-07
68 7.5621634367505e-07
69 7.58092198610007e-07
70 7.58793604260218e-07
71 7.57579675279026e-07
72 7.56299967437021e-07
73 7.56494560128657e-07
74 7.57930387375315e-07
75 7.60671163274442e-07
76 7.62619353340674e-07
77 7.63205418444338e-07
78 7.62728405826703e-07
79 7.60281540828354e-07
80 7.56540395497416e-07
};
\addplot [semithick, color1]
table {%
0 7.1309907610264e-07
1 7.13015711529476e-07
2 7.13220817735677e-07
3 6.34382378145791e-07
4 6.33403169514836e-07
5 6.8734182386034e-07
6 7.61163595882156e-07
7 8.35470561201389e-07
8 9.00091704843702e-07
9 9.49447506609334e-07
10 9.80698814956156e-07
11 9.91519123860871e-07
12 9.70469298611629e-07
13 9.48339838385368e-07
14 9.3917780730405e-07
15 9.37334942158495e-07
16 9.24136955322707e-07
17 8.94338925750813e-07
18 8.6658584822515e-07
19 8.50856014724845e-07
20 8.27966207920211e-07
21 8.22739872202501e-07
22 8.32510773167108e-07
23 8.41335715064933e-07
24 8.36557986645946e-07
25 8.27140158797283e-07
26 8.15779628373981e-07
27 8.001448677149e-07
28 7.93265167755149e-07
29 7.93263094548874e-07
30 7.97146865862728e-07
31 7.98259568128301e-07
32 7.92849419177133e-07
33 7.86964564820488e-07
34 7.77434423031356e-07
35 7.7003413352216e-07
36 7.64554875523125e-07
37 7.58776411143742e-07
38 7.56903923864033e-07
39 7.61827263122192e-07
40 7.70529404967283e-07
41 7.77324771859658e-07
42 7.79093576109796e-07
43 7.77961842201853e-07
44 7.75865025947601e-07
45 7.72646440406219e-07
46 7.6900256002601e-07
47 7.66474873706379e-07
48 7.66571266890407e-07
49 7.69277400783262e-07
50 7.7247852667954e-07
51 7.73474726533e-07
52 7.7047396190398e-07
53 7.64648934407915e-07
54 7.5848438094228e-07
55 7.55176892830426e-07
56 7.55165765980093e-07
57 7.57908085744101e-07
58 7.62141420377305e-07
59 7.66852822782554e-07
60 7.69789039752357e-07
61 7.69372323983182e-07
62 7.65426237754377e-07
63 7.6003761923514e-07
64 7.55079962759117e-07
65 7.51927500993315e-07
66 7.52103674677743e-07
67 7.53930590341399e-07
68 7.5621634367505e-07
69 7.58092198610007e-07
70 7.58793604260218e-07
71 7.57579675279026e-07
72 7.56299967437021e-07
73 7.56494560128657e-07
74 7.57930387375315e-07
75 7.60671163274442e-07
76 7.62619353340674e-07
77 7.63205418444338e-07
78 7.62728405826703e-07
79 7.60281540828354e-07
80 7.56540395497416e-07
};
%
%
%

\end{axis}

\end{tikzpicture}}
 \caption{Time evolution of the total surface area $\int_{\xv}\Sigma(\xv)d\xv$: without averaging (dashed-line),  with scale average $\averglen=20$ (solid line).}
\label{fig:sigma_drop_collision}
\end{figure}

\begin{figure}
  \centering{
\pgfplotsset{
compat=1.3,
legend image code/.code={
\draw[mark repeat=2,mark phase=2]
plot coordinates {
(0cm,0cm)
(0.15cm,0cm)        
(0.3cm,0cm)         
};%
}
}
\footnotesize

\begin{tikzpicture}

\definecolor{color1}{rgb}{1,0.498039215686275,0.0549019607843137}
\definecolor{color0}{rgb}{0.12156862745098,0.466666666666667,0.705882352941177}

\begin{axis}[
scale=0.75,
xlabel={time},
ylabel={$\Sigma \widetilde{H}$},
xmin=-3.99999999999235, xmax=83.9999999998393,
ymin=-0.005883838465291, ymax=-0.002710189299109,
tick align=outside,
tick pos=left,
xmajorgrids,
x grid style={lightgray!92.026143790849673!black},
ymajorgrids,
y grid style={lightgray!92.026143790849673!black},
legend entries={{$\Delta x$},{$10 \Delta x$}},
legend style={draw=white!80.0!black},
legend cell align={left}
]
\addplot [semithick, mark=triangle, color0]
table {%
0 -0.00413725463497524
1 -0.00413692715909143
2 -0.00413713141461042
3 -0.00298059432591644
4 -0.00290209194451454
5 -0.00303930595261828
6 -0.00321555564242821
7 -0.00338546115124961
8 -0.00353478054512825
9 -0.0036525303382757
10 -0.00373798163518843
11 -0.00465573272993722
12 -0.00573958168500643
13 -0.00560375700085097
14 -0.00559352640558957
15 -0.00561005320690402
16 -0.00568658074718677
17 -0.00565084431871861
18 -0.00550006880284594
19 -0.00525613338666216
20 -0.00489155897402369
21 -0.00489482697786434
22 -0.00499374515799938
23 -0.00512552174626631
24 -0.00525577336596924
25 -0.00538013466137299
26 -0.00553623559733379
27 -0.00566131602577821
28 -0.00555071784584405
29 -0.00552862455817841
30 -0.00561251673730764
31 -0.00568567044630384
32 -0.00557659094535688
33 -0.00551685647669144
34 -0.00534291421593306
35 -0.00535178698788292
36 -0.00527436200973769
37 -0.00528089128833192
38 -0.00524743251756829
39 -0.00523499455679645
40 -0.00524356994233723
41 -0.00525917490729191
42 -0.00526488759917625
43 -0.00526942430761962
44 -0.00526651102144302
45 -0.00524844327686427
46 -0.00524425451488483
47 -0.00524111478110499
48 -0.0052443540876927
49 -0.00524395381328095
50 -0.00523523407861655
51 -0.00524267872072626
52 -0.0052328532275545
53 -0.00521330857897532
54 -0.00517807078108042
55 -0.00517125300903863
56 -0.00517440926275947
57 -0.00518425195858189
58 -0.00518493617785803
59 -0.00520421772394765
60 -0.00521497339236137
61 -0.00521265311953842
62 -0.00519359210473521
63 -0.00517339034380269
64 -0.00516496055891437
65 -0.0051397931331906
66 -0.00514787392019448
67 -0.00515444341090475
68 -0.00514921634944938
69 -0.00516328707834648
70 -0.0051780124234522
71 -0.00516468820243148
72 -0.00516379183610524
73 -0.00517294426614031
74 -0.00517283093045611
75 -0.00519792665213727
76 -0.00520228146700497
77 -0.00518991871992201
78 -0.00519480854403367
79 -0.00517337059197358
80 -0.0051479937097362
};
\addplot [semithick, color1]
table {%
0 -0.00413725463497525
1 -0.00413692715909143
2 -0.00413713141461042
3 -0.00298059432591644
4 -0.00290209194451454
5 -0.00303930595261828
6 -0.00321555564242821
7 -0.00338546115124961
8 -0.00353478054512825
9 -0.0036525303382757
10 -0.00373798163518843
11 -0.00465573272993722
12 -0.00573958168500643
13 -0.00560375654954831
14 -0.00559352640558957
15 -0.00561005320690402
16 -0.00568658074718677
17 -0.00565084431871861
18 -0.00550006880284594
19 -0.00525613338666216
20 -0.00489155897402369
21 -0.00489482697786434
22 -0.00499374515799938
23 -0.00512552174626631
24 -0.00525577336596924
25 -0.00538013466137299
26 -0.00553623559733379
27 -0.00566131602577821
28 -0.00555071784584405
29 -0.00552862455817841
30 -0.00561251673730764
31 -0.00568567092096187
32 -0.00557659094535688
33 -0.00551685647669144
34 -0.00534291421593306
35 -0.00535178698788292
36 -0.00527436200973769
37 -0.00528089128833192
38 -0.00524743334526083
39 -0.00523499455679645
40 -0.00524356994233723
41 -0.00525917490729191
42 -0.00526488759917625
43 -0.00526942430761962
44 -0.00526651102144302
45 -0.00524844327686427
46 -0.00524425451488483
47 -0.00524111478110499
48 -0.0052443540876927
49 -0.00524395381328095
50 -0.00523523407861655
51 -0.00524267872072626
52 -0.0052328532275545
53 -0.00521330857897532
54 -0.00517807078108042
55 -0.00517125300903863
56 -0.00517440926275947
57 -0.00518425195858189
58 -0.00518493617785803
59 -0.00520421772394765
60 -0.00521497339236137
61 -0.00521265311953842
62 -0.00519359210473522
63 -0.00517339034380269
64 -0.00516496055891437
65 -0.0051397931331906
66 -0.00514787392019448
67 -0.00515444341090475
68 -0.00514921634944937
69 -0.00516328707834648
70 -0.0051780124234522
71 -0.00516468820243148
72 -0.00516379183610524
73 -0.00517294426614031
74 -0.00517283093045611
75 -0.00519792665213727
76 -0.00520228146700497
77 -0.00518991871992201
78 -0.00519480854403367
79 -0.00517337059197358
80 -0.0051479937097362
};
%
%
%

\end{axis}

\end{tikzpicture}}
 \caption{Time evolution of the total mean curvature $\int_{\xv}\Sigma\widetilde{\Mcurv}(\xv)d\xv$: without averaging (dashed-line),  with scale average $\averglen=20$ (solid line).}
\label{fig:sigmaH_drop_collision}
\end{figure}

\begin{figure}
  \centering{

\pgfplotsset{
compat=1.3,
legend image code/.code={
\draw[mark repeat=2,mark phase=2]
plot coordinates {
(0cm,0cm)
(0.15cm,0cm)        
(0.3cm,0cm)         
};%
}
}
\footnotesize

\begin{tikzpicture}

\definecolor{color1}{rgb}{1,0.498039215686275,0.0549019607843137}
\definecolor{color0}{rgb}{0.12156862745098,0.466666666666667,0.705882352941177}

\begin{axis}[
scale=0.75,
xlabel={time},
ylabel={$\Sigma \widetilde{G}/(4\pi)$},
xmin=-3.99999999999235, xmax=83.9999999998393,
ymin=-2.450000070814, ymax=7.450001491494,
tick align=outside,
tick pos=left,
xmajorgrids,
x grid style={lightgray!92.026143790849673!black},
ymajorgrids,
y grid style={lightgray!92.026143790849673!black},
legend entries={{$\Delta x$},{$10 \Delta x$}},
legend style={draw=white!80.0!black},
legend cell align={left}
]
\addplot [semithick, mark=triangle, color0]
table {%
0 2.00000000014
0.999999999999026 2.00000000014
2.00000000000181 2.00000000014
3.00000000000083 1.00000000013
4.00000000001113 1.00000000013
5.00000000001391 1.00000000014
6.00000000001669 1.00000000015
6.99999999998192 1.00000000017
7.9999999999847 1.00000000018
8.99999999998748 1.00000000019
9.99999999999026 1.0000000002
10.999999999993 -1.9999999998
11.9999999999958 1.98114913455e-10
12.9999999999986 2.00000000019
14.0000000000014 1.00000000019
15.0000000000042 1.00000000019
16.000000000007 1.00000000019
17.0000000000097 1.00000000018
18.0000000000125 1.00000000017
19.0000000000153 1.00000025307
20.0000000000181 2.00025904217
20.9999999999833 2.0000001686
21.9999999999861 2.00000001737
22.9999999999889 1.99999996486
23.9999999999917 2.00000027903
24.9999999999944 1.99999983615
25.9999999999972 1.99999980334
27 5.00000031229
28.0000000000028 4.99999866607
29.0000000000056 4.00000000016
30.0000000000083 4.00000000016
31.0000000000111 5.95932497899
32.0000000000139 7.00000142048
33.0000000000167 6.00000252432
33.9999999999819 5.00000047536
34.9999999999847 6.00000082423
35.9999999999875 4.99999959054
36.9999999999903 6.00000000792
38.0000000001808 4.95563458917
38.9999999999583 5.00000142689
40.0000000001113 4.99999785424
40.9999999998887 5.00000036873
42.0000000000417 5.00000000016
42.9999999998192 5.00000000016
43.9999999999722 5.00000000016
45.0000000001252 5.00000000016
45.9999999999026 5.00000000016
47.0000000000556 5.00000000016
47.9999999998331 5.00000000016
48.9999999999861 5.00000000015
50.0000000001391 5.00000000015
50.9999999999165 5.00000000015
52.0000000000695 5.00000000015
52.999999999847 5.00000000015
54 5.00000000015
55.000000000153 5.00000000015
55.9999999999304 5.00000000015
57.0000000000835 5.00000000015
57.9999999998609 5.00000000015
59.0000000000139 5.00000000016
60.0000000001669 5.00000000016
60.9999999999444 5.00000000016
62.0000000000974 5.00000000016
62.9999999998748 4.99999995446
64.0000000000278 4.99999969329
65.0000000001808 5.00000000015
65.9999999999583 5.00000000015
67.0000000001113 5.00000097615
67.9999999998887 4.99999961392
69.0000000000417 5.00000032375
69.9999999998192 5.00000091432
70.9999999999722 5.00000030582
72.0000000001252 5.00000041696
72.9999999999026 4.99999980986
74.0000000000556 5.00000167208
74.9999999998331 4.99999992265
75.9999999999861 5.00000006328
77.0000000001391 5.00000078325
77.9999999999165 4.9999985901
79.0000000000695 4.99999967447
79.999999999847 4.9999997342
};
\addplot [semithick, color1]
table {%
0 2.00000000014
0.999999999999026 2.00000000014
2.00000000000181 2.00000000014
3.00000000000083 1.00000000013
4.00000000001113 1.00000000013
5.00000000001391 1.00000000014
6.00000000001669 1.00000000015
6.99999999998192 1.00000000017
7.9999999999847 1.00000000018
8.99999999998748 1.00000000019
9.99999999999026 1.0000000002
10.999999999993 -1.9999999998
11.9999999999958 1.98115196171e-10
12.9999999999986 1.99999914432
14.0000000000014 1.00000000019
15.0000000000042 1.00000000019
16.000000000007 1.00000000019
17.0000000000097 1.00000000018
18.0000000000125 1.00000000017
19.0000000000153 1.00000025307
20.0000000000181 2.00025904217
20.9999999999833 2.0000001686
21.9999999999861 2.00000001737
22.9999999999889 1.99999996486
23.9999999999917 2.00000027903
24.9999999999944 1.99999983615
25.9999999999972 1.99999980334
27 5.00000031229
28.0000000000028 4.99999866607
29.0000000000056 4.00000000016
30.0000000000083 4.00000000016
31.0000000000111 5.9593233789
32.0000000000139 7.00000142048
33.0000000000167 6.00000252432
33.9999999999819 5.00000047536
34.9999999999847 6.00000082423
35.9999999999875 4.99999959054
36.9999999999903 6.00000000792
38.0000000001808 4.95563308577
38.9999999999583 5.00000142689
40.0000000001113 4.99999785424
40.9999999998887 5.00000036873
42.0000000000417 5.00000000016
42.9999999998192 5.00000000016
43.9999999999722 5.00000000016
45.0000000001252 5.00000000016
45.9999999999026 5.00000000016
47.0000000000556 5.00000000016
47.9999999998331 5.00000000016
48.9999999999861 5.00000000015
50.0000000001391 5.00000000015
50.9999999999165 5.00000000015
52.0000000000695 5.00000000015
52.999999999847 5.00000000015
54 5.00000000015
55.000000000153 5.00000000015
55.9999999999304 5.00000000015
57.0000000000835 5.00000000015
57.9999999998609 5.00000000015
59.0000000000139 5.00000000016
60.0000000001669 5.00000000016
60.9999999999444 5.00000000016
62.0000000000974 5.00000000016
62.9999999998748 4.99999995446
64.0000000000278 4.99999969329
65.0000000001808 5.00000000015
65.9999999999583 5.00000000015
67.0000000001113 5.00000097615
67.9999999998887 4.99999961392
69.0000000000417 5.00000032375
69.9999999998192 5.00000091432
70.9999999999722 5.00000030582
72.0000000001252 5.00000041696
72.9999999999026 4.99999980986
74.0000000000556 5.00000167208
74.9999999998331 4.99999992265
75.9999999999861 5.00000006328
77.0000000001391 5.00000078325
77.9999999999165 4.9999985901
79.0000000000695 4.99999967447
79.999999999847 4.9999997342
};
%
%
%

\end{axis}

\end{tikzpicture}}
 \caption{Time evolution of the total gauss curvature $\int_{\xv}\Sigma\widetilde{\Gcurv}(\xv)d\xv$: without averaging (dashed-line),  with scale average $\averglen=20$ (solid line).}
\label{fig:sigmaG_drop_collision}
\end{figure}

\FloatBarrier

\section{Conclusion}

In this paper, we propose a new statistical approach of the gas-liquid interface dedicated to 
two-phase flow modeling based on geometrical interface variables. 

Relying on a statistical description of the interface between the two phases, 
our first  contribution has been to propose a transport equation for a surface density function valid for both regimes: disperse and separated phases. 
The related phase space has been identified: it includes the curvatures 
and the velocity of the interface. An original link between
such a surface density function formalism and its application to obtain statistics at the level of objects, such as the 
number density function (NDF) for sprays of droplets, has been proposed 
by introducing the discrete SDF (DSDF). The DSDF is only 
valid for disperse phase and supposes that we can isolate droplets/bubbles in small volumes. However, we can then describe droplets and bubbles with arbitrary shapes as long as they are homeomorphic to a sphere and this provides us with an interesting theoretical framework, which naturally degenerates to previous contributions when the objects are spheres \cite{Essadki-SIAM}.

In a second main contribution, we have  defined the spatially averaged SDF (ASDF), with an averaging kernel bounded to a 
small region around the interface and which preserves some information of the standard SDF 
given by the first moments of this distribution. 
We have shown that the ASDF degenerates to the DSDF, when 
the liquid or gas phase is disperse. In this case, the link with the NDF can be identified 
straightforwardly and explicitly. 
We have then shown how we can derive 
reduced-order models from an equation on the SDF using the moments of these 
distributions. However, we still need a closure modeling in situations of complex topological changes, 
while we have illustrated that in some simplified situations (spherical droplets), a
closure model can be derived \cite{Essadki-SIAM,OGST}. 
This model has already  been used to simulate evaporating polydisperse sprays with variables related to interface geometry. 

Finally, to illustrate and assess the theoretical part, we have 
designed a new algorithm to extract the curvatures and the two different distributions NDF and SDF from 
a level-set field obtained with the DNS ARCHER code. 
This new algorithm preserves some geometrical and topological 
information, which essentially allows us to compute a NDF from an ASDF. This new tool can serve, in a
future work, to post-process more representative two-phase flows DNS simulations \cite{romain2017}, with and without topological changes, and 
thus propose closure modeling of the curvatures evolutions and possibly closure of the distribution
from its moments.  
\section*{Acknowledgments}
This research was supported by a Ph.D. grant for M. Essadki from IFP Energies nouvelles. The support of 
Ecole doctorale de math\'ematiques Hadamard (EDMH) and EM2C Laboratory are gratefully acknowledged.
This research was also supported by a  Ph.D. grant for F. Drui from  CEA/DGA at Maison de la Simulation and EM2C Laboratory. 
The support of ANR project "Mobilit\'e durable et syst\`emes de transport (DS0603) - 2014" 
entitled "SUB SUPER JET : Mod\'elisation de l'atomisation d'un jet liquide avec transition sous- et super-critique"
 (ANR-14-CE22-0014 - PI T. Schmitt) is also gratefully acknowledged. We would 
like to express our special thanks to our colleagues from CORIA laboratory: F.X. Demoulin, S. Puggelli, R. Canu and 
J. R\'eveillon for the numerous interesting and helpful discussions. 
The very helpful comments of S. Kokh and R. Di Battista are eventually gratefully acknowledged.

\begin{appendix}

\section{Evolution equation for the area density measure}
\label{app:eq-deltaI}

In this section, let us derive the evolution equation for the area density measure $\delta_I$. First, let us recall that $\delta_I$ is defined
as a distribution function by:
\[
	\delta_I (t, \xv) = ||\gradient{\xv}{g(t,\xv)}|| \, \delta(g(t,\xv)).
\]
The Lagrangian derivative $\dot{\bullet}$ of $\delta_I$ reads:
\begin{align*}
	\dot{\delta}_I (t, \xv) 
		& = \partial_t\delta_I + \Vi \cdot \gradient{\xv}{\delta_I} \\
		& = \left[ 
			\partial_t \left( ||\gradient{\xv}{g(t,\xv)}|| \right)  
			+ \Vi \cdot \gradient{\xv}{\left( ||\gradient{\xv}{g(t,\xv)}|| \right)}
			\right]
			\, \delta(g(t,\xv)) \\
		& \quad + \left[
			\partial_t \left(  \delta(g(t,\xv)) \right)
			+ \Vi \cdot \gradient{\xv}{ \left(  \delta(g(t,\xv)) \right)}
		 \right] \,  ||\gradient{\xv}{g(t,\xv)}|| \\
\end{align*}

The second term in the right-hand side of the previous equality is null, because of the equation on $g(t, \xv)$
which is:
\[
	\partial_t g + \Vi \cdot \gradient{\xv}{g} = 0.
\]

Let us now compute the first term. On a one hand, we have:
\begin{align*}
	\partial_t \left( ||\gradient{\xv}{g(t,\xv)}|| \right) 
		& = \frac{\gradient{\xv}{g(t,\xv)} }{ ||\gradient{\xv}{g(t,\xv)}|| } \cdot \partial_t \left( \gradient{\xv}{g(t,\xv)} \right) \\
		& = \normal \cdot \gradient{\xv}{\left(\partial_t g(t,\xv)\right)} \\
		& = - \normal \cdot \gradient{\xv}{  \left(\Vi \cdot \gradient{\xv}{g(t, \xv)} \right)} \\
		& =  - \normal \cdot \gradient{\xv}{  \left(\vi ||\gradient{\xv}{g(t, \xv)}|| \right)}.
\end{align*}
On the other hand, one computes:
\begin{align*}
	 \Vi \cdot \gradient{\xv}{\left( ||\gradient{\xv}{g(t,\xv)}|| \right)}
		& = \vi \normal \cdot \gradient{\xv}{\left( ||\gradient{\xv}{g(t,\xv)}|| \right)} \\
		& = \normal \cdot \gradient{\xv}{\left( \vi ||\gradient{\xv}{g(t,\xv)}|| \right)} -  ||\gradient{\xv}{g(t,\xv)}||\normal \cdot \gradient{\xv}{\left( \vi  \right)}
\end{align*}
And finally, one can find that:
\begin{align}
	\nonumber \dot{\delta}_I (t, \xv) 
		& = - ||\gradient{\xv}{g(t,\xv)}||\normal \cdot \gradient{\xv}{\left( \vi  \right)}
			\, \delta(g(t,\xv)) \\
	\label{eq:deltaI-time-equation}	& = - \delta_I \, \normal \cdot \gradient{\xv}{\left( \vi  \right)}.
\end{align}

Equation \eqref{eq:deltaI-time-equation} can be further developed, by noting that:
\begin{align*}
	 - \delta_I \, \normal \cdot \gradient{\xv}{\left( \vi  \right)} 
	 	& = - \delta_I \, \left[ \nabla_{\xv} \cdot \Vi - \vi \nabla_{\xv} \cdot \normal \right]
\end{align*}
and
\[
	\nabla_{\xv} \cdot \normal = 2H.
\]
Finally, one has:
\begin{equation}
\label{eq:deltaI-time-equation2}
	\dot{\delta}_I (t, \xv)  =  - \delta_I \, \nabla_{\xv} \cdot \Vi + 2H \delta_I \vi.
\end{equation}


\end{appendix}

\section*{References}
\bibliography{biblio}

\end{document}